%% file: dAEL_journal.tex
\documentclass{article}
\usepackage{multirow}
\usepackage{array}
\usepackage[usenames,dvipsnames,tikz]{xcolor}
\usepackage[all,cmtip]{xy}

\usepackage{tikz}
\usepackage{amsthm}

\usepackage{thmtools}
\usetikzlibrary{calc,decorations.markings,arrows,positioning,fit,shapes,calc}
\usetikzlibrary{shadows}

\makeatletter
\newcommand{\thickhline}{%
    \noalign {\ifnum 0=`}\fi \hrule height 1pt
    \futurelet \reserved@a \@xhline
}
\newcolumntype{"}{@{\hskip\tabcolsep\vrule width 1pt\hskip\tabcolsep}}
\makeatother
\usepackage{fullpage}
\usepackage{times}
\usepackage{helvet}
\usepackage{courier}

\usepackage{comment}
\usepackage{framed}
\usepackage{wrapfig}
\input{idp-latex/idp-latex}
\usepackage{natbib}
\renewcommand\cite[1]{\citep{#1}}
\usepackage{booktabs}
\usepackage{amsmath}
\usepackage{amsfonts}
\usepackage{xspace}

\usepackage{algorithm}
\usepackage[noend]{algorithmic}

\newcommand{\mc}[1]{\mathcal{#1}}

\usepackage{tikz}

\definecolor{lavander}{cmyk}{0,0.48,0,0}
\definecolor{violet}{cmyk}{0.79,0.88,0,0}
\definecolor{burntorange}{cmyk}{0,0.52,1,0}

\def\lav{black!90}
\def\oran{black!30}

\tikzstyle{peers}=[draw,circle,violet,bottom color=\lav,
                  top color= white, text=violet,minimum width=10pt]
\tikzstyle{superpeers}=[draw,circle,burntorange, left color=\oran,
                       text=violet,minimum width=20pt]
\tikzstyle{legendsp}=[rectangle, draw, burntorange, rounded corners,
                     thin,bottom color=\oran, top color=white,
                     text=burntorange, minimum width=2.5cm]
\tikzstyle{legendp}=[rectangle, draw, violet, rounded corners, thin,
                     bottom color=\lav, top color= white,
                     text= violet, minimum width= 2.5cm]
\tikzstyle{legend_general}=[rectangle, rounded corners, thin,
                           burntorange, fill= white, draw, text=violet,
                           minimum width=2.5cm, minimum height=0.8cm]

\newcommand{\T}{\m{\mc{T}}} 
\renewcommand{\I}{I}

\newcommand{\A}{\mc{A}} 
\renewcommand{\fo}{\mc{L}} 

\newcommand{\dael}{\m{\mc{L}_d}} 
\newcommand{\ael}{\m{\mc{L}_k}}
\newcommand{\pws}{\m{Q}} 
\newcommand{\upws}{\m{\mc{Q}}} 
\newcommand{\Tt}{\m{\mc{T}}} 

\newcommand{\daelr}{\m{\textnormal{dAEL}_R}} 

\newcommand{\ubp}{\m{\mc{B}}} 

\newcommand{\leqpws}{\leq_K}
\newcommand{\lequpws}{\leq_K}
\newcommand{\leqbp}{\leq_{p}}
\declaretheorem[style=plain,name=Task,numberlike=thm]{task}

\newcommand{\upwsrevision}{\m{\mc{D}_{\mc{T}}}}

\newcommand{\ubprevision}{\m{\mc{D}^*_{\mc{T}}}}

\newcommand{\ubprevisionL}{\m{\mc{D}^c_{\mc{T}}}}
\newcommand{\ubprevisionU}{\m{\mc{D}^l_{\mc{T}}}}

\newcommand{\upwsstrevision}{S_{\mc{D}^*_{\mc{T}}}}
\newcommand\citetDMT[1]{DMT [\citeyear{#1}]\xspace}
\newcommand\citeDMT[1]{[DMT \citeyear{#1}]\xspace}

\newcommand{\upwsrev}[1]{\mc{D}_{#1}}
\newcommand{\upwsrevL}[1]{\mc{D}^c_{#1}}

\newcommand{\ubprev}[1]{\mc{D}^*_{#1}}

\newcommand{\upwsrevst}[1]{S_{\mc{D}^*_{#1}}}

\newcommand{\says}{\,\m{\mathit{says}}\,}
\newcommand{\say}{\m{\mathit{says}}}
\newcommand{\access}{\m{\mathit{access}}}
\newcommand{\delegate}{\m{\mathit{deleg\_{}to}}}
\newcommand{\revoke}{\m{\mathit{revoke}}}

\newcommand{\Terms}{\mathbb{T}}
\newcommand{\dummy}{\delta}

\renewcommand{\theory}{T}

\newboolean{showproofs}
\setboolean{showproofs}{true}
\newcommand{\aproof}[1]{
\ifthenelse{\boolean{showproofs}}{#1}{}}

 \newcommand{\varass}{\m{a}}
 \newcommand\Apred{\m{\mathrm{Agt}}}
 \newcommand\vyes{\m{\mathrm{yes}}}
 \newcommand\vno{\m{\mathrm{no}}}

\renewcommand{\leqk}{\m{\leq_K}}
\renewcommand{\geqk}{\m{\geq_K}}

\begin{document}
\title{Distributed Autoepistemic Logic: \\ Semantics, Complexity, and Applications to Access Control}
 \author{Marcos Cramer\\ 
   TU Dresden\\ Dresden, Germany\\ marcos.cramer@tu-dresden.de\\ Tel.: +49 351 463 38426  \and
	Pieter Van Hertum\\ 
   pietervanhertum@gmail.com \and
   Bart Bogaerts\\
   Vrije Universiteit Brussel (VUB)\\ Brussels, Belgium\\
   bart.bogaerts@vub.be \and
   Marc Denecker\\ KU Leuven\\ Leuven, Belgium\\ marc.denecker@kuleuven.be 
}
\maketitle              


\begin{abstract}

In this paper we define and study a multi-agent extension of autoepistemic logic (AEL) called \emph{distributed autoepistemic logic} (dAEL).
We define the semantics of dAEL using approximation fixpoint theory, an abstract algebraic framework that unifies different knowledge representation formalisms by describing their semantics as fixpoints of semantic operators.
We define 2- and 3-valued semantic operators for dAEL.
Using these operators, approximation fixpoint theory allows us to define a class of  semantics for dAEL, each based on different intuitions that are well-studied in the context of AEL.
We define a mapping from dAEL to AEL and identify the conditions under which the mapping preserves semantics, and furthermore argue that when it does not, the dAEL semantics is more desirable than the AEL-induced semantics since dAEL manages to contain inconsistencies.

The development of dAEL has been motivated by 
an application in the domain of \emph{access control}.
We explain how dAEL can be fruitfully applied to this domain and discuss how well-suited the different semantics are for the application in access control.

\end{abstract}


\section{Introduction}
\input{chapters/introduction.tex}

\paragraph{Publication History}
A preliminary version of this paper was presented at the IJCAI conference \cite{ijcai/HertumCBD16}. 
The current paper extends the previous work with examples, proofs, a detailed account of the semantical relationship between dAEL and AEL and a more detailed discussion of the applicability of dAEL to access control. 
In contrast to the conference paper, the current version no longer defines how the construct of inductive definitions can be incorporated into dAEL, as we found that it complicated the logic without being necessary for our applications.

\section{Motivation}\label{sec:motiv}
\input{chapters/motivation.tex}

\section{Formal Preliminaries}
\input{chapters/preliminaries.tex}

\section{dAEL: Syntax and Semantics}\label{sec:dAEL}
\input{chapters/syntax.tex}

\input{chapters/semantics.tex}

\section{dAEL and AEL}\label{sec:mapping}
\input{chapters/mapping.tex}


\section{Applying dAEL to Access Control}\label{sec:AC}\label{sec:use-cases}
\input{chapters/AC.tex}

\section{Complexity}\label{sec:compl}
\input{chapters/procedure.tex}

\section{Related Work}\label{sec:related}
\input{chapters/related.tex}

\section{Conclusion and Future Work} \label{sec:challenge}
\input{chapters/conclusion.tex}

\newpage
\section*{Acknowledgements}
This research was supported by project GOA 13/010 of the Research Fund KU Leuven and projects G.0489.10, G.0357.12, and G.0922.13 of the Research Foundation - Flanders and more specifically was part of the FNR-FWO project \emph{Specification logics and Inference tools for verification and Enforcement of Policies}. Part of this work was also performed while 
Bart Bogaerts was a postdoctoral fellow of the Research Foundation -- Flanders (FWO).

\bibliographystyle{plainnat}

\input{dAEL_journal.bbl}
\appendix

\input{chapters/appendix}

\end{document}

%% file: idp-latex/idp-latex.tex
\usepackage{import}

\subimport{idp-latex/}{includes-no-theorems}

\subimport{idp-latex/}{common-citations}
\subimport{idp-latex/}{theorems}

\ignore{

\include{includes-no-theorems}
\include{common-citations}
\include{theorems}
}

%% file: idp-latex/includes-no-theorems.tex
\usepackage{ifthen}
\usepackage{url}
\usepackage[usenames,dvipsnames,tikz]{xcolor}
\usepackage{amssymb}
\usepackage{amsmath}
\usepackage{import}
\usepackage{xparse}
\usepackage{amsmath}
\usepackage{xspace}
\usepackage{datatool}
\usepackage{glossaries}

\providecommand\m[1]{\ensuremath{#1}\xspace}
\renewcommand{\m}[1]{\ensuremath{#1}\xspace}
\newcommand{\trval}[1]{\m{\mathbf{#1}}}



	\newcommand{\limplies}{\Rightarrow}
	
	\newcommand{\lequiv}{\Leftrightarrow}

	\newcommand{\lrule}{\leftarrow}
	\newcommand{\cause}{\stackrel{c}{\lrule}}
	
	\newcommand{\ltrue}{\trval{t}}
	\newcommand{\lfalse}{\trval{f}}
	\newcommand{\lunkn}{\trval{u}}
	


	\newcommand{\ra}{\rightarrow}

	\newcommand{\Tr}{\ltrue}
	\newcommand{\Fa}{\lfalse}
	\newcommand{\Un}{\lunkn}

	\newcommand{\voc}{\m{\Sigma}}

	\newcommand{\struct}{\m{I}}
	
	\newcommand{\I}{\m{\mathcal{I}}}

	\newcommand{\theory}{\m{\mathcal{T}}}

	\newcommand{\PP}{\m{\mathcal{P}}}


	\NewDocumentCommand\inter{g+g}{%
	  \IfNoValueTF{#1}
	    {\struct}
	    {\m{#1^{#2}}}}



	\newcommand{\xxx}{\m{\overline{x}}}

	\newcommand{\ddd}{\m{\overline{d}}}

	\newcommand{\ttt}{\m{\overline{t}}}





	\renewcommand{\int}{\m{\mathbb{Z}}}

	\newcommand{\leqp}{\m{\leq_p}}
	\newcommand{\geqp}{\m{\geq_p}}
	\newcommand{\leqk}{\m{\leq_k}}
	\newcommand{\geqk}{\m{\geq_k}}
	\newcommand{\leqt}{\m{\leq_t}}
	\newcommand{\geqt}{\m{\geq_t}}
	
	\DeclareMathOperator\glb{glb}
	\DeclareMathOperator\lub{lub}
	\DeclareMathOperator\lfp{lfp}



	\NewDocumentCommand\subs{g+g}{%
	  \IfNoValueTF{#1}
	    {\m{/}}
	    {\m{#1/ #2}}}


	\newcommand{\logicname}[1]{\textsc{#1}\xspace}

	\newcommand{\idp}{\logicname{IDP}}


	\newcommand{\fodot}{\logicname{FO(\ensuremath{\cdot})}}


	\newcommand{\fo}{\FO}

\newcommand{\ouracronym}[3]{%
	\newacronym{#1}{#2}{#3}
	\expandafter\newcommand\csname #1\endcsname{\gls{#1}\xspace}%
}
	\ouracronym{FO}{FO}{first-order logic}
	\ouracronym{PC}{PC}{propositional calculus}
	\ouracronym{MX}{MX}{Model Expansion}
	\ouracronym{MO}{MO}{Model Optimization}
	\ouracronym{ASP}{ASP}{Answer Set Programming}
	\ouracronym{TP}{TP}{Theorem Proving}
	\ouracronym{LP}{LP}{Logic Programming}
	\ouracronym{CP}{CP}{Constraint Programming}
	\ouracronym{FP}{FP}{Functional Programming}
	\ouracronym{KR}{KR}{Knowledge Representation}
	\ouracronym{CSP}{CSP}{Constraint Satisfaction Problem}
	\ouracronym{SMT}{SMT}{SAT Modulo Theories}
	\ouracronym{KBS}{KBS}{knowledge base system}
	\ouracronym{NNF}{NNF}{Negation Normal Form}
	\ouracronym{FNNF}{FNNF}{Flat Negation Normal Form}
	\ouracronym{DefNNF}{DefNNF}{Definition Negation Normal Form}
	\ouracronym{DEFNF}{DEFNF}{Definition Normal Form}
	\ouracronym{CDCL}{CDCL}{Conflict-Driven Clause-Learning}
	\ouracronym{WFS}{WFS}{Well-Founded Semantics}
	\ouracronym{LCG}{LCG}{Lazy Clause Generation}
	\ouracronym{AEL}{AEL}{Autoepistemic Logic}
	\ouracronym{OEL}{OEL}{Ordered Epistemic Logic}
	\ouracronym{AFT}{AFT}{Approximation Fixpoint Theory}



%

	\makeatletter
	\def\ifenv#1{
	\def\@tempa{#1}%
	\def\@ttempa{#1*}%
	\ifx\@tempa\@currenvir
	\expandafter\@firstoftwo
	\else
	\expandafter\@secondoftwo
	\fi
	}
	\makeatother

	\newcommand{\ddrule}[4]{\ensuremath{#1 \leftarrow #2 & \{#3\} & #4}}
	\newcommand{\drule}[2]{\ensuremath{#1 & \leftarrow & #2}}

	\newcommand{\darule}[4]{\ensuremath{#1 \leftarrow #2 & \{#3\} & #4}}
	\newcommand{\arule}[2]{\ensuremath{#1 \, &\leftarrow \, #2}}

	\newcommand{\LNDRule}[2]{
	\ifenv{array}
	{\drule{#1}{#2}}
	{ \ifenv{align}
		{\arule{#1}{#2}}
		{\ifenv{align*}
		{\arule{#1}{#2}}
		{ERROR: using LDRule in unsupported environment: \@currenvir}
		}
	}
	}

	\newcommand{\LDRule}[4]{
	\ifenv{array}
	{\ddrule{#1}{#2}{#3}{#4}}
	{ \ifenv{align}
		{\darule{#1}{#2}{#3}{#4}}
		{\ifenv{align*}
		{\darule{#1}{#2}{#3}{#4}}
		{ERROR: using LDRule in unsupported environment: \@currenvir}
		}
	}
	}

	\NewDocumentCommand\LRule{m+g+g+g}{%
		\IfNoValueTF{#2}%
		{#1.&}{%
		\IfNoValueTF{#3}
		{\LNDRule{#1}{#2.}}
		{\LDRule{#1}{#2.}{#3}{#4}}%
		}
	}


	\NewDocumentCommand\CLRule{m+g}{%
	\ifenv{array}
	{\cdrule{#1}{#2}}
	{ \ifenv{align}
		{\carule{#1}{#2}}
		{\ifenv{align*}
			{\carule{#1}{#2}}
			{ERROR: using CLRule in unsupported environment: \@currenvir}
		}
	}
	}

	\NewDocumentCommand\carule{m+g}{%
		\IfNoValueTF{#2}
			{\ensuremath{#1.}}
			{\ensuremath{#1 \, &\cause \, #2}}}
	\NewDocumentCommand\cdrule{m+g}{%
		\IfNoValueTF{#2}
			{\ensuremath{#1.}}
			{\ensuremath{#1 & \cause & #2}}}



	\newcommand{\algrule}[4]{
	\hbox{{#1}:}& 
	\quad #2 ~\longrightarrow~ #3 
	\hbox{~ if } #4\\
	}

	\newcommand{\AlgoRule}[4]{
	\ifenv{array}
	{\algrule{#1}{#2}{#3}{#4}}
		{ERROR: using AlgoRule in unsupported environment: \@currenvir}
	}


	\newcommand{\ignore}[1]{}

	\newboolean{nocomments}
	\setboolean{nocomments}{false}

	\newboolean{commentmargin}
	\setboolean{commentmargin}{true}

	\newcommand{\namedcomment}[3]{%
		\ifthenelse{\boolean{nocomments}}%
		{}
		{
			\ifthenelse{\boolean{commentmargin}}%
				{ {\color{#3} \marginpar{\color{#3}\sc #2}#1}  }
				{  {\color{#3} {\sc #2}: #1}  }
		}%
	}
	\newcommand{\mnamedcomment}[3]{\ifthenelse{\boolean{nocomments}}{}{{\marginpar{\tiny \color{#3}{\sc #2}:#1}}}}



	\usepackage{soul}


\usepackage[normalem]{ulem} 
\makeatletter
\font\uwavefont=lasyb10 scaled 700
\def\spelling{\bgroup\markoverwith{\lower3.5\p@\hbox{\uwavefont\textcolor{Red}{\char58}}}\ULon}
\def\grammar{\bgroup\markoverwith{\lower3.5\p@\hbox{\uwavefont\textcolor{LimeGreen}{\char58}}}\ULon}
\def\phrasing{\bgroup\markoverwith{\lower3.5\p@\hbox{\uwavefont\textcolor{RoyalBlue}{\char58}}}\ULon}

\newcommand\remove{\bgroup\markoverwith{\textcolor{red}{\rule[0.5ex]{2pt}{0.4pt}}}\ULon}
\makeatother







%% file: idp-latex/common-citations.tex
\usepackage{etoolbox}

\ExplSyntaxOn
\newcommand\setcitation[2]{%
	\csdef{mycommoncitation\text_uppercase:n{#1}}{#2}}
\newcommand\getcitation[1]{%
	\csuse{mycommoncitation\text_uppercase:n{#1}}}

\ExplSyntaxOff

\setcitation{IDP}{WarrenBook/DeCatBBD18}
\setcitation{fodot}{tocl/DeneckerT08}
\setcitation{CPsupport}{ictai/DeCat13}
\setcitation{minisatid}{ictai/DeCat13}
\setcitation{FunctionDetection}{iclp/DeCatB13}
\setcitation{lazyGrounding}{jair/CatDBS15} 
\setcitation{Bootstrapping}{ngc/BogaertsJDJBD16}
\setcitation{BootstrappingIDP}{ngc/BogaertsJDJBD16}
\setcitation{FOSymmetry}{tplp/DevriendtBBD16}
\setcitation{IDP2}{lash08/WittocxMD08}
\setcitation{FOLASP}{tplp/VanDesselDV21}

\setcitation{foid}{tocl/DeneckerT08}
\setcitation{cplogic}{journal/tplp/VennekensDB10}

\setcitation{fodot2asp}{corr/DeneckerLTV19} 
\setcitation{Tarskian}{corr/DeneckerLTV19} 
\setcitation{TarskianSemanticsASP}{corr/DeneckerLTV19} 

\setcitation{SAT}{faia/336}
\setcitation{HandbookOfSAT}{faia/336}


\setcitation{DPLLT}{cav/GanzingerHNOT04}
\setcitation{SMT}{faia/BarrettSST21}

\setcitation{CP}{fai/Rossi06}
\setcitation{Inca}{iclp/DrescherW12}
\setcitation{csp2asp}{ijcai/DrescherW11}
\setcitation{EZCSP}{lpnmr/Balduccini11}

\setcitation{ASPComp2}{lpnmr/DeneckerVBGT09}
\setcitation{ASPComp3}{journals/tplp/CalimeriIR14}
\setcitation{ASPComp4}{conf/lpnmr/AlvianoCCDDIKKOPPRRSSSWX13}
\setcitation{ASPComp5}{journals/ai/CalimeriGMR16}
\setcitation{ASPComp6}{jair/GebserMR17}
\setcitation{ASPComp7}{tplp/GebserMR20}
\setcitation{AspInPractice}{synthesis/2012Gebser}
\setcitation{clasp}{ai/GebserKS12}
\setcitation{oclingo}{kr/GebserGKOSS12}
\setcitation{clingo}{iclp/GebserKKOSW16}
\setcitation{gringo}{lpnmr/GebserST07}
\setcitation{cmodels}{aaai/GiunchigliaLM04}
\setcitation{DLV}{tocl/LeonePFEGPS06}
\setcitation{GroundingBottleneck}{padl/SonPL14}

\setcitation{MaxSAT}{faia/LiM21}


\setcitation{KR}{Baral:2003}


\setcitation{inputster}{tplp/Jansen13}
\setcitation{LearningPaper}{TPLP/BruynoogheBBDDJLRDV} 
\setcitation{clog}{iclp/BogaertsVDV14} 
\setcitation{foc}{iclp/BogaertsVDV14}
\setcitation{inferenceClog}{ecai/BogaertsVDV14}
\setcitation{inferenceFOC}{ecai/BogaertsVDV14}
\setcitation{examplesClog}{nmr/BogaertsVDV14b} 
\setcitation{AFT}{DeneckerMT00}
\setcitation{KBS}{iclp/DeneckerV08}
\setcitation{KBS-invitedtalk}{jelia/Denecker16}
\setcitation{KBPE}{inap/DePooterWD11}
\setcitation{lazygroundingASP}{ijcai/BogaertsW18} 
\setcitation{justifications}{lpnmr/DeneckerBS15} 
\setcitation{justificationTheory}{lpnmr/DeneckerBS15} 
\setcitation{justificationsAlpha}{ijcai/BogaertsW18} 
\setcitation{ASP}{marek99stable} 
\setcitation{satid}{sat/MarienWDB08}
\setcitation{lazyclausegeneration}{constraints/OhrimenkoSC09}
\setcitation{FP}{ACMCS/Hudak89}
\setcitation{GroundingWithBounds}{jair/WittocxMD10}
\setcitation{GroundWithBounds}{jair/WittocxMD10}

\setcitation{LTC}{iclp/Bogaerts14}
\setcitation{SPSAT}{ictai/DevriendtBMDD12}
\setcitation{BreakID}{sat/DevriendtBBD16}
\setcitation{LCG}{stuckeyLCG}
\setcitation{MiniZinc}{conf/cp/NethercoteSBBDT07}
\setcitation{GroundedFixpoints}{ai/BogaertsVD15}
\setcitation{PartialGroundedFixpoints}{ijcai/BogaertsVD15}
\setcitation{LogicBlox}{datalog/GreenAK12}
\setcitation{proB}{journals/sttt/LeuschelB08}
\setcitation{NaturalInductions}{KR/DeneckerV14} 
\setcitation{LP}{jacm/EmdenK76}

\setcitation{AF}{ai/Dung95}
\setcitation{ADF}{kr/BrewkaW10}
\setcitation{af}{ai/Dung95}
\setcitation{adf}{kr/BrewkaW10}
\setcitation{ADFRevisited}{ijcai/BrewkaSEWW13}
\setcitation{DefaultLogic}{ai/Reiter80}
\setcitation{DL}{ai/Reiter80}
\setcitation{AEL}{mo85}
\setcitation{minisat}{sat/EenS03}
\setcitation{completion}{adbt/Clark78}
\setcitation{ClarkCompletion}{adbt/Clark78}
\setcitation{wasp}{lpnmr/AlvianoDFLR13}
\setcitation{lcg}{stuckeyLCG}
\setcitation{CEGAR}{jacm/ClarkeGJLV03}
\setcitation{CuttingPlane}{or/DantzigFJ54}
\setcitation{CuttingPlanes}{or/DantzigFJ54}
\setcitation{kodkod}{tacas/TorlakJ07}
\setcitation{CDCL}{Marques-SilvaS99}
\setcitation{1UIP}{iccad/ZhangMMM01}
\setcitation{relevance}{ijcai/JansenBDJD16}
\setcitation{relevance-implementation}{aspocp/JansenBDJD16}
\setcitation{WFS}{GelderRS91}
\setcitation{UnfoundedSet}{GelderRS91}
\setcitation{UFS}{GelderRS91}
\setcitation{stablesemantics}{iclp/GelfondL88}
\setcitation{shatter}{Shatter}
\setcitation{sbass}{drtiwa11a}
\setcitation{lparsemanual}{url:lparsemanual}
\setcitation{AIC}{ppdp/FlescaGZ04}
\setcitation{templates}{tplp/DassevilleHJD15}
\setcitation{templates2}{iclp/DassevilleHBJD16}
\setcitation{sat-to-sat}{aaai/JanhunenTT16}
\setcitation{sat-to-sat-qbf}{bnp/BogaertsJT16}
\setcitation{sat-to-sat-SO}{kr/BogaertsJT16}
\setcitation{XSB}{SwiW12}
\setcitation{KCmap}{jair/DarwicheM02}
\setcitation{KnowledgeCompilationMap}{jair/DarwicheM02}
\setcitation{TLA}{DBLP:books/aw/Lamport2002}
\setcitation{EventB}{BookAbrial2010}
\setcitation{MX}{MitchellT05}
\setcitation{MIP}{Sierksma96}
\setcitation{PerfectModel}{minker88/Przymusinski88}
\setcitation{SafeInductions}{ijcai/BogaertsVD17}
\setcitation{AIC}{ppdp/FlescaGZ04}
\setcitation{alpha}{lpnmr/Weinzierl17}
\setcitation{omiga}{jelia/Dao-TranEFWW12}
\setcitation{gasp}{fuin/PaluDPR09}
\setcitation{asperix}{lpnmr/LefevreN09a}
\setcitation{CTL}{lop/ClarkeE81}
\setcitation{AFT-AIC}{ai/BogaertsC18}
\setcitation{UltimateApproximator}{DeneckerMT04}
\setcitation{KripkeKleene}{Fitting85}
\setcitation{AFT-HO}{corr/CharalambidisRS18} 
\setcitation{HereThere}{Heyting30}
\setcitation{dAEL}{ijcai/HertumCBD16}
\setcitation{SDD}{ijcai/Darwiche11}
\setcitation{HEX}{ijcai/EiterIST05}
\setcitation{wADF}{aaai/BrewkaSWW18}
\setcitation{wADFfix}{corr/BrewkaSWW18}
\setcitation{TransitionSystems}{jacm/NieuwenhuisOT06}
\setcitation{galliwasp}{lopstr/MarpleG12}
\setcitation{GalliWasp}{lopstr/MarpleG12}
\setcitation{clingcon}{tplp/BanbaraKOS17}
\setcitation{lp2sat}{birthday/JanhunenN11}
\setcitation{lp2mip}{LIU12}
\setcitation{lp2diff}{lpnmr/JanhunenNS09}
\setcitation{lp2acyc}{ecai/GebserJR14}
\setcitation{PB}{faia/RousselM09}
\setcitation{pb}{faia/RousselM09}
\setcitation{CuttingPlanes}{dam/CookCT87}
\setcitation{RoundingSAT}{ijcai/ElffersN18}
\setcitation{PRS}{aaai/DixonG02}
\setcitation{sat4j}{jsat/BerreP10}
\setcitation{mingo}{LIU12}
\setcitation{pbmodels}{lpnmr/LiuT05}
\setcitation{HEF-LP}{lpnmr/GebserLL07}
\setcitation{HCF-LP}{amai/Ben-EliyahuD94}
\setcitation{aspcore2}{tplp/CalimeriFGIKKLM20}
\setcitation{AspCore2}{tplp/CalimeriFGIKKLM20}
\setcitation{SemanticWeb}{book/AntoniouGHH12}
\setcitation{QBF}{faia/BuningB09}
\setcitation{SMUS}{cp/IgnatievPLM15}
\setcitation{CDCLsym}{tacas/MetinBCK18}
\setcitation{SEL}{sat/DevriendtBB17}
\setcitation{Coq}{txtcs/BertotC04}
\setcitation{Isabelle}{book/NipkowPW02}
\setcitation{HOL}{tphol/SlindN08}
\setcitation{lean}{cade/MouraU21}
\setcitation{CakeML}{jfp/TanMKFON19}
\setcitation{FRAT}{lmcs/BaekCH22}
\setcitation{DRUP}{date/GoldbergN03}
\setcitation{RUP}{date/GoldbergN03}
\setcitation{GRIT}{tacas/Cruz-FilipeMS17}
\setcitation{TraceCheck}{url:TraceCheck}
\setcitation{LRAT}{cade/Cruz-FilipeHHKS17}
\setcitation{PR}{cade/HeuleKB17}
\setcitation{DRAT}{sat/WetzlerHH14}
\setcitation{satcomp2016}{aaai/BalyoHJ17}
\setcitation{satcomp2013}{url:SATcomp2013}
\setcitation{Zinc}{constraints/MarriottNRSBW08}
\setcitation{MiniZinc}{cp/NethercoteSBBDT07}
\setcitation{Essence}{constraints/FrischHJHM08}
\setcitation{smodels}{lpnmr/SyrjanenN01}
\setcitation{lparse}{Syrjanen98}
\setcitation{IDPZ3}{corr/CarbonelleVVD22} 
\setcitation{IDP-Z3}{corr/CarbonelleVVD22} 
\setcitation{Z3}{tacas/MouraB08}

\setcitation{DMN}{DMN} 
\setcitation{Planning}{book/GhallabNT04}
\setcitation{AutomatedPlanning}{book/GhallabNT04}
\setcitation{RC2}{jsat/IgnatievMM19}
\setcitation{enfragmo}{phd/Aavani14}
\setcitation{}{}
  
 \ExplSyntaxOn
\newcommand\refto[1]{%
      \ifcsname  mycommoncitation\text_uppercase:n{#1}\endcsname%
      \getcitation{#1}%
      \else%
      #1%
      \fi%
      }

\newcommand\mycite[1]{%
      \ifcsname mycommoncitation\text_uppercase:n{#1}\endcsname%
   \cite{\getcitation{#1}}%
  \else%
    \cite{#1}%
  \fi%
}	
  
\newcommand\mycitet[1]{%
      \ifcsname mycommoncitation\text_uppercase:n{#1}\endcsname%
   \citet{\getcitation{#1}}%
  \else%
    \citet{#1}
  \fi%
}	
  
\ExplSyntaxOff

%% file: idp-latex/theorems.tex
\usepackage{amsthm}
\usepackage{thmtools}

\declaretheorem[style=plain,	name=Theorem,		numberwithin=section]{thm}
\declaretheorem[style=plain,	name=Theorem,		numberlike=thm]{theorem}

\declaretheorem[style=plain,	name=Proposition,	numberlike=thm]{proposition}

\declaretheorem[style=plain,	name=Lemma,		numberlike=thm]{lemma}
\declaretheorem[style=plain,	name=Lemma,		numbered=no]   {lem*}

\declaretheorem[style=definition,	name=Definition,	numberlike=thm]{definition}

\declaretheorem[style=definition,	qed=$\blacktriangle$,	numberlike=thm]{example}

\declaretheorem[style=definition,	qed=$\blacktriangle$,	numbered=no]{ex*}

\declaretheorem[style=remark,	name=Notation,		numbered=no]{nota*}

\declaretheorem[style=remark,	qed=$\blacktriangle$,	name=Note,		numbered=no]{note*}

%% file: chapters/introduction.tex




\emph{Access control} is concerned with methods to determine which principal (i.e.\ user or program) has the right to access a resource, e.g.\ the right to read or modify a file. 
Many logics have been proposed for distributed access control \cite{Abadi03,Gurevich07,Abadi08,Garg12,Genovese12}. 
Most of these logics use a modality $k \says$indexed by a principal $k$. 
\say-based access control logics are designed for systems in which different principals can issue statements that become part of the access control policy. 
$k \says \varphi$ is usually rendered as ``$k$ supports $\varphi$'', which can be interpreted to mean that $k$ has issued statements that -- together with additional information present in the system -- imply $\varphi$.
Different access control logics vary in their account of which additional information may be assumed in deriving the statements that $k$ supports.

In Section \ref{sec:motiv}, we argue that it is reasonable to assume that the statements issued by a principal are a complete characterization of what the agent supports.
This is similar to the ``All I know''-assumption \cite{ai/Levesque90} in autoepistemic logic (AEL) \cite{\refto{AEL},nonmon30/DeneckerMT11}, which states that an AEL theory is considered to be a complete characterization of what the agent knows. 
As such, one might wonder if AEL can be a suitable logic for representing access control policies. 

In order to illustrate the ``All I know''-assumption of AEL, consider the AEL theory $T_1$ containing only the formula $\neg K \neg p \rightarrow p$. This formula says that if it is not known that $p$ does not hold, then $p$ holds. The idea behind the ``All I know''-assumption is that a formula is considered to be known precisely when it is a consequence of the AEL theory under consideration. The formula $\neg p$ cannot be derived from $T_1$, so $\neg p$ is not known. Thus $\neg K \neg p$ is true. Taken together with the formula $\neg K \neg p \rightarrow p$, this implies $p$. So $p$ is a consequence of $T_1$. 
This can be contrasted with what happens in the theory $T_2$ that contains $\neg K \neg p \rightarrow p$ as well as $\neg p$. Due to the presence of $\neg p$ in the theory, $\neg p$ can be derived from $T_2$, so $\neg p$ is known. Therefore $K \neg p$ holds, i.e.\ $\neg K \neg p$ does not hold. Thus unlike in the case of $T_1$, $p$ cannot be derived from the theory. This example also illustrates the non-monotonicity of AEL: In $T_1$ we could derive $p$, but after adding the information $\neg p$ to the theory, we could no longer derive $p$. Intuitively, the formula $\neg K \neg p \rightarrow p$ says that as long as no claims implying $\neg p$ are included in the theory, $p$ may be concluded. 

Now consider an access control scenario in which a user $A$ wants to grant access to a file to a user $B$ but also wants to allow user $C$ to deny this access to user $B$. In this case, $A$ wants that the access right of user $B$ is derivable as long as user $C$ does not issue a denial of that access right. So the derivability of the access right behaves similarly as the derivability of $p$ in the AEL example in the previous paragraph. In this paper, we will make use of this similarity between AEL and reasoning about access rights to motivate the application of a modified version of AEL to access control.

A first restriction that prohibits the application of AEL to access control is that AEL is designed to only model the state of mind of a \emph{single agent}, while in the domain of access control, typically multiple agents are in play. 
An extension to AEL with multiple agents has been defined by \citet{VlaeminckVBD/KR2012}, but this extension requires a global stratification on the agents, i.e., an order on the agents, where agents higher in the order can only refer to knowledge/statements of agents lower in the order. 
This is undesirable for a distributed system; e.g., in Section \ref{sec:use-cases} we present situations where such an order simply does not exist. 
Therefore, we extend AEL to a truly distributed multi-agent setting, and name our extension \emph{distributed autoepistemic logic} (\emph{dAEL}).
We argue in Section \ref{sec:AC} that the proposed extension provides a good formal model of the \say-modality. 

As the term ``autoepistemic logic'' suggests, AEL was designed to model 
(a single agent's) 
\emph{knowledge}, including knowledge derived from reasoning about knowledge. 
However, the formalism of AEL can be applied to model other modalities too. 
Note that claims about an agent's knowledge 
are claims about that agent's internal state of mind. 
However, the formalism of AEL does not presuppose that its $K$ modality represents an internal state of mind of an agent. 
For example, we can interpret the $K$ modality to refer to the public commitments of an agent, i.e.\ interpret $K \phi$ to mean that the agent in question has publicly made statements that imply $\phi$, and as such identify $K$ with the \say modality. 
In what follows, we will keep the AEL terminology and refer to $K$ as ``knowledge'' without thereby implying that it represents an internal state of mind.


In dAEL, agents have full (positive and negative) introspection into other agents' knowledge. 
This is of course an unreasonable assumption when the $K$ modality  represents an internal state of mind like actual knowledge. 
It is, however, reasonable when $K \phi$ is interpreted to mean that an agent has (publicly) issued statements that imply $\phi$. 


Section \ref{sec:motiv} gives some preliminary motivation for the design choices of dAEL based on the access control application that we have in mind.
Section \ref{s:preliminaries} contains preliminaries from AEL and approximation fixpoint theory. In Section \ref{sec:dael}, we first define the syntax of dAEL, then define a 2-valued and a 3-valued semantic operator for dAEL, and then show how approximation fixpoint theory can be applied to these operators to define a class of semantics for dAEL corresponding to equally-named,  well-known semantics for AEL. In Section \ref{sec:mapping}, we define a mapping from dAEL to AEL and show that for a subset of the logic defined by a consistency requirement, the mapping preserves all semantics. 
The class of theories in which the semantics coincide are, intuitively, those in which all agents have consistent knowledge; outside of this class,  we show that our new logic dAEL manages to contain inconsistencies within a single agent, while in AEL this is impossible. 
Next, in Section \ref{sec:AC}, we discuss some use cases of applying dAEL to access control, and in Section \ref{sec:compl}, we study complexity of inference in our logic. After discussing related work in Section \ref{sec:related}, we conclude the paper in Section \ref{sec:conclusion} with remarks about possible topics for future work.

%% file: chapters/motivation.tex
Before defining the syntax and semantics of dAEL, we discuss some general features of the \say-based approach to access control to give some preliminary motivations for the formalism. 

An \emph{access control policy} is a set of norms defining which principal is to be granted access to which resource under which circumstances.
Specialized logics called \emph{access control logics} were developed for representing policies and access requests and reasoning about them.
A general principle adopted by most logic-based approaches to access control is that access is granted if and only if it is logically entailed by the policy.

\subsection{\say-based access control logics and denial}
There is a large variety of access control logics, but most of them use a modality $k \says$ indexed by a principal $k$ \cite{Genovese12}.
\say-based access control logics are designed for systems in which different principals can issue statements that become part of the access control policy.
$k \says \phi$ is usually explained informally to mean that $k$ supports $\phi$ \cite{Abadi08,Garg12,Genovese12};
this means that $k$ has issued statements that -- together with additional information present in the system -- imply $\phi$.
Different access control logics vary in their account of which rules of inference and which additional information may be used in deriving  statements that $k$ supports from the statements that $k$ has explicitly issued.
For instance, if $\mathit{Alice}$ issues the statement $(\mathit{Bob}\says \mathit{ok}) \limplies \mathit{ok}$ and $\mathit{Bob}$ issues the statement $\mathit{ok}$, it is to be expected that also $\mathit{Alice}\says \mathit{ok}$ holds. 


Many state-of-the-art \say-based access control logics, e.g., Binder Logic (BL) \cite{Garg12}, are designed for application in a system based on \emph{proof-carrying authorization}  \cite{ccs/AppelF99}. In such a system, an access request is always submitted together with a proof that establishes that the requester has access, and the task of the reference monitor is only to check the validity of this proof. If the proof is based on the assumption that some other principal $k$ supports some formula $\phi$, the proof will contain the signed certificate that establishes that $k$ has made statements implying $\phi$. In this way, assumptions of the form $k \says \phi$ can be discharged. However, there is no way to discharge of assumptions of the form $\neg k \says \phi$. If $k$ has not made any statements implying $\phi$, there is no way to prove this to the reference monitor by submitting some certificates issued by $k$, as the reference monitor can never be convinced that there are no other statements made by $k$, which have not been presented to the reference monitor.
For this reason, many state-of-the-art \say-based access control logics do not provide the means for deriving statements of the form $\neg k \says \phi$ or $j \says \neg k \says \phi$ on the basis of the observation that $k$ has not issued any statements that could imply $\phi$. 

\label{sec:denial}
However, precisely formulas of this form make it possible to model access denials naturally in a \say-based access control logic, as illustrated in the following example, which we already briefly sketched in the introduction.
\begin{example}
\label{ex:postdoc}Suppose $A$ is a professor with control over a resource $r$, $B$ is a PhD student of $A$ who needs access to $r$, and $C$ is a postdoc of $A$ supervising $B$. $A$ wants to grant $B$ access to $r$, but wants to grant $C$ the right to deny $B$'s access to $r$, for example in case $B$ misuses her rights.
A natural way for $A$ to do this  using the \say-modality is to issue the statement $\neg C \says \neg \access(B,r) \Rightarrow \access(B,r)$. This should have the effect that $B$ has access to $r$ unless $C$ denies her access.
However, this effect can only be achieved if our logic allows $A$ to derive $\neg C \says \neg \access(B,r)$ from the fact that $C$ has not issued any statements implying $\neg \access(B,r)$.
\end{example}

Such denials can be realized in a system in various ways: One way is to have a central server where all statements belonging to the access control policy are stored, independently of who has issued them. In this case, the reference monitor can confirm that $C$ has not issued a statement implying $\neg \access(B,r)$ and thus grant $B$ access. This way the access control system is not truly distributed, even though the access control policy is still produced in a distributed way.

Such denials can also be realized in a truly distributed system if a certain degree of cooperativity of the principals with the reference monitor is assumed. Suppose for example that $C$ does not want to deny $B$ access right to $r$. In this case he will not issue any statement implying $\neg \access(B,r)$. Furthermore, it is reasonable to assume that he will be cooperative with the reference monitor in this respect: If the reference monitor asks $C$ whether he has issued statements implying $\neg \access(B,r)$, he will say no. If $C$ were not cooperative in this way, it would have the same effect as him stating $\neg \access(B,r)$, which goes against his goal of not denying $B$ access right. So given that the cooperativity needed here is in the interest of the concerning principals, we do not consider this cooperativity assumption problematic.

Note that the applicability of our logic does not depend on how precisely the access control system is realized in practice.

\subsection{Autoepistemic Logic}
\label{sec:autoepistemic}
The derivation of $\neg C \says \neg \access(B,r)$ described above, i.e., its derivation from the fact that $C$ has not issued any statements implying $\neg \access(B,r)$, is non-monotonic: If $C$ later issues a statement implying $\neg \access(B,r)$, the formula $\neg C \says \neg \access(B,r)$ can no longer be derived. In other words, adding a formula to the access control policy causes that something previously implied by the policy is no longer implied. Existing \say-based access control logics are monotonic and hence they cannot support the type of reasoning described above for modelling denial with the \say-modality.

In order to derive statements of the form $\neg k \says \phi$, we have to assume the statements issued by a principal to be a complete characterization of what the principal supports. This is similar to the motivation behind Moore's autoepistemic logic (AEL) to consider an agent's theory to be a complete characterization of what the agent knows \cite{mo85,ai/Levesque90,ijcai/Niemela91,nonmon30/DeneckerMT11}.
This motivates an application of AEL to access control.

However, AEL cannot model more than one agent. In order to extend it to the multi-agent case, one needs to specify how the knowledge of the agents interacts. To the best of our knowledge, all state-of-the-art access control logics allow $j \says (k \says \phi)$ to be derived from $k \says \phi$, as this is required for standard delegation to be naturally modelled using the \say-modality. In the knowledge terminology of AEL, this can be called mutual positive introspection between agents. In order to also model denial as described above, we also need mutual negative introspection, i.e., that $j \says \neg k \says \phi$ to be derived from $\neg k \says \phi$.

\subsection{Approximation Fixpoint Theory}

Several semantics have been proposed for AEL. Approximation fixpoint theory (AFT) (see Subsection \ref{ss:AFT}) is an algebraic framework that captures most of those. Furthermore, AFT provides us with a unified methodology for lifting AEL semantics to a distributed setting. What is required to apply AFT is to define semantic operators for our distributed version of AEL. 

Among the many semantics induced by AFT, we find, in line with the claims from \citet{nonmon30/DeneckerMT11}, the well-founded semantics to be best suited when applying dAEL to access control. Unlike other widely studied semantics of autoepistemic logic like the Moore's original expansion semantics, the Kripke-Kleene semantics and the stable semantics, the well-founded semantics is both \emph{grounded} \cite{ai/BogaertsVD15}, meaning that derivable formulas are supported by cycle-free justifications, and \emph{constructive} \cite{lpnmr/DeneckerV07}, meaning that the model can be characterized as the limit of a construction process. 
Both of these features are important for the access control application, as it means that agents can only access, or delegate control over a resource if there is a (non-cyclic) reason for it, and furthermore, that the provenance of this access can be traced back (by means of following the construction process). 
In Section \ref{sec:AC} we discuss application scenarios that illustrate these advantages of the well-founded semantics over other semantics.


\ignore{
\subsection{Inductive definitions}
\label{sec:IDs}
Apart from extending AEL with the well-founded semantics to the multi-agent case, dAEL also incorporates inductive definitions, thus allowing principals to define access rights and other properties relevant for access control in an inductive way. Inductive (recursive) definitions are a common concept in all branches of mathematics. Inductive definitions in dAEL are intended to be understood in the same way as in the general purpose specification language \fodot of the \idp system \cite{WarrenBook/DeCatBBD14}. \citet{Denecker:CL2000} showed that in classical logics, adding definitions leads to a strictly more expressive language.

Because of their rule-based nature, formal inductive definitions also bear strong similarities in syntax and formal semantics with logic programs. A formal inductive definition could also be understood intuitively as a logic program which has arbitrary formulas in the body and which defines only a subset of the
predicates in terms of parameter predicates not defined in the definition.

Most of the semantics that have been proposed for logic programs can be adapted to inductive definitions. \citet{KR/DeneckerV14} have argued that the well-founded semantics correctly formalises our intuitive understanding of inductive definitions,
and hence that it is actually the \emph{right} semantics. Following them, we use the well-founded semantics for inductive definitions.

In Section \ref{sec:AC} we give an example policy to show the usefulness of including inductive definitions in dAEL. }

%% file: chapters/preliminaries.tex


\label{s:preliminaries}
We assume familiarity with the basic concepts of first-order logic. 
We assume throughout this paper that a first-order vocabulary $\Sigma$ is fixed, use $\Terms$ for the set of terms over $\Sigma$ (which we call \emph{$\Sigma$-terms}) and $\fo$ for the language of standard first-order logic over $\Sigma$. Furthermore, we assume that $\Sigma$ is the disjoint union of $\Sigma_o$ and $\Sigma_s$, where $\Sigma_o$ represents a set of \emph{objective} symbols and $\Sigma_s$ a set of \emph{subjective} symbols. Symbols in $\Sigma_o$ could for instance be arithmetic symbols, equality,  or other symbols whose interpretation is shared among all involved agents. 
We assume that an infinite supply of variables is available and fixed throughout the paper. 
 A variable assignment $\varass$ assigns to each variable an object of a given domain. If $x$ is a variable and $d$ an element of the given domain, we use $\varass[x:d]$ for the variable assignment that assigns $d$ to $x$ and otherwise equals \varass. 
We consider the set of logical symbols of $\mc{L}$ to formally consist of $\land$, $\neg$ and $\forall$. The symbols $\lor$, $\Rightarrow$, $\Leftrightarrow$ and $\exists$ are, as usual, treated as abbreviations in the standard way:
\begin{align*}
 (\varphi \lor \psi) & = \lnot (\lnot \varphi \land \lnot \psi)\\
 (\varphi \limplies \psi) & = (\lnot \varphi \lor \psi)\\
 (\varphi \lequiv \psi) &= ((\varphi \limplies \psi) \land (\psi \limplies \varphi))\\
 \exists x: \varphi &= \lnot \forall x: \lnot \varphi. 
\end{align*}
Brackets may be dropped when this does not lead to ambiguity.

We use truth values $\ltrue$ for truth, $\lfalse$ for falsity and additionally, in a three-valued setting, we use $\lunkn$ for unknown.  
The truth order $<_t$ on truth values is induced
by $\lfalse<_t\lunkn<_t\ltrue$. The precision order $<_p$ on truth values is
induced by $\lunkn<_p\ltrue, \lunkn<_p\lfalse$.
We define $\ltrue^{-1}=\lfalse, \lfalse^{-1}=\ltrue$ and $\lunkn^{-1}=\lunkn$.

\subsection{Autoepistemic Logic} \label{ss:AEL}
The language $\ael$ of \emph{autoepistemic logic}~\cite{mo85}\footnote{Technically, Moore only defined the propositional fragment of the logic we define below. 
Here, we are interested in a first-order variant of Moore's logic. Also the extension with \emph{objective information} (the distinction between $\Sigma_o$ and $\Sigma_s$) was not part of the original presentation. }
is defined recursively using the standard rules for the syntax of first-order logic, augmented with one modal rule. This syntax is standard in modal logics.  The language is thus defined by 
\begin{align*}
 &P(\ttt) \in \ael &&\text{ if } P \text{ is an $n$-ary predicate in $\Sigma$ } \text{ and }\ttt\text{ an $n$-tuple of terms}\\
 &(\varphi \land \psi) \in \ael &&\text{ if } \varphi \in \ael \text{ and }\psi \in \ael\\
& \lnot \varphi \in \ael &&\text{ if } \varphi \in \ael \\ 
 &\forall x: \varphi \in \ael&&\text{ if } \varphi \in \ael \\
& K \varphi \in \ael &&\text{ if }  \varphi\in\ael
\end{align*}

%

An AEL theory $T$ is a set of sentences (that is, formulas without free occurrences of variables) in $\ael$. 
AEL uses the semantic concepts of standard modal logic. A \emph{structure} is defined as usual in first-order logic.  It formally represents a potential state of affairs of the world. 
We assume a domain $D$, shared by all structures, to be fixed throughout the paper. We also assume a $\Sigma_o$-structure $I_o$ is fixed, representing the shared knowledge among all involved agents (currently, there is only one, but in the next section, there will be multiple agents). 
A \emph{possible world structure} is a set of $\Sigma$-structures that coincide with $I_o$ on all symbols of $\Sigma_o$.
It contains all structures that are consistent with an agent's knowledge. 
Possible world structures are ordered with respect to the amount of knowledge they contain. 
Possible world structures that contain fewer structures possess more knowledge. 
Indeed, an agent $A$ ``knows'' that a certain claim holds if this claim holds in all worlds $A$ deems possible. Thus, in smaller possible world structures, more knowledge is present. 
Formally, given possible world structures $\pws_1$ and $\pws_2$, we define $\pws_1 \leqpws \pws_2$ to hold if and only if $\pws_2 \subseteq \pws_1$.

The semantics of AEL is based on the standard S5 truth assignment \cite{Lewis32,hughescresswell.niml}.
The \emph{value} of a formula $\varphi \in \ael$ with respect to a possible world structure $\pws$, a structure $\I$ and a variable assignment \varass (denoted $\varphi^{\pws,I,\varass}$) is defined using the standard recursive rules for first-order logic augmented with one additional rule for the modal operation. Formally,  we define
\begin{align*} 
& (P(\ttt))^{\pws,I,\varass} &&= \left\{ \begin{array}{l}\ltrue \text{ if } \ttt^{I, \varass} \in P^I\\\lfalse\text{ otherwise}\end{array}\right. \\
 &(\neg \varphi)^{\pws,I,\varass}&&=(\varphi^{\pws,I,\varass})^{-1}\\
 &(\varphi\land \psi)^{\pws,I,\varass}&&=\glb_{{\leqt}} (\varphi^{\pws,I,\varass},\psi^{\pws,I,\varass})\\
 &(\forall x:\varphi)^{\pws,I,\varass}&&=\glb_{{\leqt}} \{\varphi^{\pws,I,\varass [x:d]}\mid d \in D\}\\
 &(K\varphi)^{\pws,I,\varass}&&=\left\{\begin{array}{ll}
                                  \Tr & \text{if } \varphi^{\pws,I',\varass}=\Tr \text{ for all $I'\in \pws$}\\
                                  \Fa & \text{otherwise}
                                 \end{array}\right.
\end{align*}

If $\varphi$ is an AEL sentence, it is easy to see that $\varphi^{\pws,I,\varass}$ is independent of $\varass$. In this case we use $\varphi^{\pws, I}$ to denote this value.

Moore proposed to formalise the intuition that an AEL theory $\theory$ expresses ``all the agent knows'' in semantic terms, as a  condition on the possible world structure $Q$  representing the agent's belief state. The condition is: a world $I$ is possible according to $Q$  if and only if $I$ satisfies $\theory$ given $Q$.
Formally, 
Moore defines that $Q$ is an {\em autoepistemic expansion} of $\theory$ if for every world $I$, it holds that $I\in Q$ if and only if $\theory^{Q,I}=\ltrue$. 

The above definition is essentially a fixpoint characterisation.
The underlying operator $D_\theory$ maps $Q$ to 
\[
	D_\theory(Q)=\{\struct\mid \theory^{Q,I}=\ltrue\}.
\]
Autoepistemic expansions are exactly the fixpoints of $D_\theory$;
they are the possible world structures that, according to Moore,  express candidate belief states  of an autoepistemic agent with knowledge base \theory.

Soon, researchers pointed out certain ``anomalies'' in the expansion semantics \cite{nato/HalpernM85,Konolige88}.  In the following years, many different semantics for AEL were proposed. It was only with the introduction of the abstract algebraical framework \emph{approximation fixpoint theory} (AFT) that a uniform view on those different semantics was obtained. 
We define several of the semantics of AEL later, after introducing the algebraical preliminaries on AFT.

\subsection{Logic Programming}\label{ss:lp}

When introducing approximation fixpoint theory in the next section, we will illustrate some of our definitions in the context of logic programming.  
We recall some preliminaries.
For presentational purposes, we here restrict our attention to propositional logic programs.

\newcommand\head{\m{\mathit{head}}}
\newcommand\body{\m{\mathit{body}}}

Let \voc be an alphabet, i.e., a collection of symbols which are called \emph{atoms}. A \emph{literal} is an atom $p$ or the negation $\lnot q$ of an atom $q$.  
A logic program $\mathcal{P}$ is a set of \emph{rules} $r$ of the form 
\[h\lrule \phi,\] where
$h$ is an atom called the \emph{head} of $r$, denoted $\head(r)$, and $\phi$ is a propositional formula called the \emph{body} of $r$, denoted $\body(r)$.
An interpretation $\struct$ of the alphabet \voc is an element of $2^\voc$, i.e., a subset of $\voc$. 
The truth value (\ltrue or \lfalse) of a propositional formula $\varphi$ in a structure $\struct$, denoted $\varphi^\struct$ is defined as usual.
With a logic program \PP, we associate an immediate consequence operator \cite{jacm/EmdenK76} $T_\PP$ that maps a structure $\struct$ to 
	\[T_\PP(\struct) = \{p\mid \exists r \in \PP: \head(r)=p\land \body(r)^\struct=\ltrue\}.\] 

\subsection{Approximation Fixpoint Theory}\label{ss:AFT}

Approximation fixpoint theory (AFT) is an algebraic framework that allows unifying semantics of various fields of non-monotonic reasoning. 
It was originally defined to unify default logic, auto-epistemic logic, and logic programming  \cite{\refto{AFT},DeneckerMT03}, but nowadays also has applications in various other domains, including abstract argumentation \citep{journals/ai/Strass13,aaai/Bogaerts19}, extensions to deal with inconsistencies \citep{RR/BiJF14}, causality \citep{\refto{foc}},  higher-order logic programs \citep{tplp/CharalambidisRS18},  active integrity constraints \citep{ai/BogaertsC18}, stream reasoning \citep{ai/Antic20}, constraint languages for the semantic web \citep{iclp/BogaertsJ21}, as well as neuro-symbolic logic programming \cite{Antic}.

In this section, we recall the basics of lattice theory and approximation fixpoint theory by Denecker, Marek and Truszczy{\'n}ski [\citeyear{DeneckerMT00}] (further shortened as DMT).

\paragraph{Lattices and Operators}

A \emph{complete lattice} ${\langle L,\leq\rangle}$ is a set $L$ equipped with a partial order $\leq$, such that every 
set $S\subseteq L$ has both a  least upper bound and a greatest lower bound, denoted $\lub(S)$ and $\glb(S)$ respectively.
A complete lattice has a least element $\bot$ and a greatest element $\top$.

\begin{example}
 If $\voc$ is a set of atoms, then the set $2^\voc$ of interpretations forms a complete lattice equipped with the order $\subseteq$. 
 The least upper bound and greatest lower bound in that case are the union and intersection respectively.
\end{example}

An operator $O:L\to L$ is \emph{monotone} if $x\leq y$ implies that $O(x)\leq O(y)$. An element $x\in L$ is 
a \emph{fixpoint} of $O$
if 
$O(x)=x$.
Every monotone operator $O$ in a 
complete lattice has a least fixpoint, denoted $\lfp(O)$, which is 
the limit (i.e., the least upper bound) of the sequence $x_\alpha$ given by:
 \begin{align*}
 &x_0 = \bot\\
 &x_{\alpha+1} = O(x_\alpha)\\
 &x_{\lambda} = \lub(\{x_\alpha \mid \alpha<\lambda\}) \textnormal{, with $\lambda$ a limit ordinal}
 \end{align*}
 
 \begin{example}
  In the context of logic programming, if a program is positive (contains no negation in the body), then the associated operator is monotone. 
  For instance, if $\PP$ is the program:
  	\begin{align*}
	  &p.\\
	  &q\lrule p \lor q.  
	\end{align*}
	then the associated lattice and operator $T_\PP$ are visualized below
	\[
		\xymatrix@R-1pc{
		& \top = \{p,q\} \ar@{.>}@/^/@(dl,dr)& \\
		\{p\}\ar[ru]\ar@{.>}@/^/[ru]      &&         \{q\}\ar[lu]\ar@{.>}@/_/[lu]\\
		& \bot = \emptyset\ar[ru]\ar[lu]\ar@{.>}@/^/[lu]
		}
	\]
	where dotted arrows represent the operator and full arrows the $\leq$ relation on the lattice (i.e., the subset relation on interpretations). 
	The least fixpoint of this operator is $\top=\{p,q\}$, the interpretation in which both $p$ and $q$ hold. 
 \end{example}

\paragraph{Approximation Fixpoint Theory}\label{004:sec:prelimAFT}
Given a lattice $L$, AFT makes use of the set 
$L^2$.  We call elements of $L^2$ \emph{approximations}.
We define \emph{projections} for pairs as usual:
$(x,y)_1=x$ and $(x,y)_2=y$.  Pairs $(x,y)\in L^2$ are used to
approximate all elements in the interval $[x,y] = \{z\mid x\leq
z\wedge z\leq y\}$. 
We call $(x,y)\in L^2$ \emph{consistent} if $x \leq y$, i.e.\ if $[x,y]$ is non-empty. We use $L^c$ to denote the set of consistent elements. 
Elements $(x,x) \in L^c$ are called \emph{exact}.  
We identify a point $x\in L$ with the exact bilattice point $(x,x)\in L^c$.
$L^2$ is equipped with a \emph{precision order}, defined as $(x,y) \leqp (u,v)$ if $x\leq u$
and $v\leq y$. If $(u,v)$ is consistent, the latter means that $(x,y)$
approximates all elements approximated by $(u,v)$, or in other words
that $[u,v]\subseteq [x,y]$.
If $L$ is a complete lattice, then so is $\langle L^2,\leqp\rangle$.

\begin{example}
 In the context of logic programming, elements of the bilattice $L^2$ are couples $(\struct_1,\struct_2)$ of two interpretations. 
 Together they correspond to a so-called \emph{four-valued interpretation}: an atom $x\in\voc$ is 
 \begin{itemize}
  \item \emph{true} if $x\in \struct_1$ and $x\in\struct_2$
  \item \emph{false} if $x\not\in \struct_1$ and $x\not\in\struct_2$
  \item \emph{unknown} if $x\not\in \struct_1$ and $x\in\struct_2$
  \item \emph{inconsistent} if $x\in \struct_1$ and $x\not\in\struct_2$
   \end{itemize}
  Elements of $L^c$ are those four-valued interpretations in which no element is inconsistent. 
\end{example}

AFT studies fixpoints of lattice operators $O:L\ra L$ through operators approximating $O$.
An operator $A: L^2\to L^2$  is an \emph{approximator} of $O$ if it is \leqp-monotone,  and has the property that for all $x$, $O(x)\in [x',y']$, where $(x',y')=A(x,x)$.
Approximators 
map $L^c$ into $L^c$.
As usual, we restrict our attention to \emph{symmetric} approximators: approximators $A$ such that for all $x$ and $y$, $A(x,y)_1 = A(y,x)_2$. 
\citetDMT{DeneckerMT04} showed that the consistent fixpoints of interest (supported, stable, well-founded) are uniquely determined by an approximator's restriction to $L^c$, hence, sometimes we only define approximators on $L^c$. 
Given an approximator $A$, we define the (complete) stable operator $S_A: L \to L: S_A(x)=\lfp(A(\cdot,x)_1)$, where $A(\cdot,y)_1$ denotes the operator $L\to L:x\mapsto A(x,y)_1$.

AFT studies fixpoints of $O$ using fixpoints of $A$. 
\begin{enumerate}
\item The \emph{$A$-Kripke-Kleene fixpoint} is the $\leqp$-least fixpoint of $A$. It approximates all fixpoints of $O$. 
\item A \emph{partial $A$-stable fixpoint} is a pair  $(x,y)$ such that $x=S_A(y)$ and $y=S_A(x)$.
\item An \emph{$A$-stable fixpoint} of $O$ is a fixpoint $x$ of $O$ such that $(x,x)$ is a partial $A$-stable fixpoint. 
\item The \emph{$A$-well-founded fixpoint} is the least precise partial $A$-stable fixpoint. 
\end{enumerate}
%
%

\begin{example}
 In the context of logic programming, the most common approximator is Fittings three- or four-valued immediate consequence operator. We present the three-valued version here. 
 Given a logic program $\PP$, Fittings operator $\Psi_\PP$ \cite{tcs/Fitting02} maps a three-valued interpretation $(\struct_1,\struct_2)$ to a three-valued interpretation $\Psi_\PP(\struct_1,\struct_2)$ with the property that 
 \begin{itemize}
  \item $x$ is true in $\Psi_\PP(\struct_1,\struct_2)$ iff there is some rule $x\lrule \phi$ with $\phi$ true in $(\struct_1,\struct_2)$,
  \item $x$ is false in $\Psi_\PP(\struct_1,\struct_2)$ iff for all rules $x\lrule \phi$, it holds that $\phi$ true in $(\struct_1,\struct_2)$,
  \item otherwise, $x$ is unknown in $\Psi_\PP(\struct_1,\struct_2)$,  
 \end{itemize}
 where $\phi$ is evaluated in a three-valued interpretation in the standard way, following Kleene's truth tables (see Figure \ref{fig:KT}). 
 
\citetDMT{DeneckerMT00} showed that $\Psi_\PP$ is indeed an approximator of $T_\PP$, that the well-founded fixpoint of $\Psi_\PP$ is the well-founded model of $\PP$ as defined by \citeauthor{GelderRS91} and that $\Psi_\PP$-stable fixpoints are exactly the stable models of $\PP$ as defined by \citeauthor{iclp/GelfondL88}. In this case, 
the operator $\Psi_\PP(\cdot,y)_1$ coincides with the immediate consequence operator of the Gelfond-Lifschitz reduct \cite{iclp/GelfondL88}. 

 For instance, if $\PP$ is the logic program:
  	\begin{align*}
	  p \lrule \lnot q.\\
	  q\lrule \lnot p.
	\end{align*}
  then it can be verified that the stable operator has three consistent fixpoints, namely the three-valued interpretations $(\{p\},\{p\})$, $(\{q\},\{q\})$, and $(\emptyset,\{p,q\})$. 
  The first two are the exact stable fixpoints; in one of them $p$ is true and $q$ false, in the other one it is the other way round. The last one is the well-founded fixpoint, in which all atoms are unknown.
\end{example}

The $A$-Kripke-Kleene fixpoint of $O$ can be constructed by iteratively applying $A$, starting from $(\bot,\top)$. 
For the $A$-well-founded fixpoint, the following constructive characterisation has been worked out by \citet{lpnmr/DeneckerV07}.
\begin{definition}
An \emph{$A$-refinement} of $(x,y)$ is a pair $(x',y')\in L^2$ satisfying one of the following two conditions:
\begin{itemize}
	\item $(x,y)\leqp(x',y')\leqp A(x,y)$, or
	\item $x'=x$ and  $A(x,y')_2\leq y'\leq y$. 
\end{itemize}
An $A$-refinement is \emph{strict} if $(x,y)\neq (x',y')$.
\end{definition}
 \begin{definition}
 A \emph{well-founded induction} of $A$  is a sequence 
$(x_i,y_i)_{i\leq \beta}$
with $\beta$ an ordinal such that 
\begin{itemize}
	\item $(x_0,y_0) = (\bot,\top)$;
	\item $(x_{i+1},y_{i+1})$ is an $A$-refinement of $(x_{i},y_{i})$, for  all $i<\beta$;
	\item $(x_{\lambda},y_{\lambda})= \lub_{\leqp} \{(x_i,y_i)\mid i<\lambda\}$
	      for each limit ordinal $\lambda\leq\beta$.
\end{itemize}
A well-founded induction is \emph{terminal} if its limit $(x_\beta,y_\beta)$ has no strict $A$-refinements.
\end{definition}
For an approximator $A$, there are many different terminal well-founded inductions of $A$.
\citet{lpnmr/DeneckerV07}  showed that they all have the same limit, and that this limit equals the $A$-well-founded fixpoint of $O$. If $A$ is symmetric, the $A$-well-founded fixpoint of $O$ (and in fact, every tuple in a well-founded induction of $A$) is consistent. 

%

An alternative constructive characterisation of the $A$-well-founded fixpoint is the following.  \mycitet{DeneckerMT00} also defined the four-valued stable operator $S_A^*:L^2 \rightarrow L^2$ by $S_A^*((x,y)) = (S_A(x),S_A(y))$. Then the $A$-well-founded fixpoint is the $\leqp$-least fixpoint of $S_A^*$, so it can be constructed by a transfinite iterative application of $S_A^*$ to $(\bot,\top)$ until a fixpoint is reached. We make use of this characterization of the $A$-well-founded fixpoint in some of the proofs that are included in the appendix.

Since the $A$-well-founded fixpoint is a consistent partial $A$-stable fixpoint, there always exists at least one consistent partial $A$-stable fixpoint. Furthermore, it easily follows from the definition of partial $A$-stable fixpoints that partial $A$-stable fixpoints are always fixpoints of $A$. These two properties together imply that if $A$ has a unique consistent fixpoint, this is also the unique consistent partial $A$-stable fixpoint.

 \subsection{AFT and Autoepistemic Logic}
\citetDMT{aaai/DMT98} showed that many semantics from AEL can be obtained by direct applications of AFT. In order to do this, they defined a three-valued version of the semantic operator. 

In order to approximate an agent's state of mind, i.e., to represent partial information about possible world structures, \citetDMT{aaai/DMT98}  defined a \emph{belief pair} as a tuple $(P,S)$ of two possible world structures.
They say that a belief pair \textit{approximates} a possible world structure $\pws$ if $P\leqpws \pws \leqpws S$, or equivalently if $S\subseteq \pws \subseteq P$.
Intuitively, $P$ is an underestimation or a \emph{conservative} bound of the agent's knowledge, and $S$ is an overestimation or \emph{liberal} bound of the agent's knowledge.
That is, 
$P$ contains all interpretations that are potentially contained in the agent's possible world structure, and $S$ all interpretations that are certainly contained in the agent's possible world structure. 
Stated even differently, $P$ represents the knowledge the agent certainly has and $S$ the knowledge the agent possibly has. 
We call a belief pair $(P,S)$ \emph{consistent} if $P\leqpws S$, i.e., if it approximates at least one possible world structure. 
From now on, we assume all belief pairs to be consistent. 
Belief pairs can be ordered by a precision ordering $\leqbp$:
 Given two belief pairs $(P,S)$ and $(P',S')$, we say that $(P,S)$ is less precise than $(P',S')$ (notation $(P,S)\leqbp (P',S')$) if $P\leqpws P'$ and $S'\leqpws S$.

We now define a three-valued valuation of sentences with respect to a belief pair (which represents an approximation of the state of mind of an agent) and a structure, representing the state of the world.

\begin{definition}\label{def:AEL3val}
The \emph{value} of $\varphi$ with respect to belief pair $B$, an interpretation $I$ and a variable assignment \varass (notation $\varphi^{B,I,\varass}$) is defined inductively as follows: 
\begin{align*}
 &(P(\ttt))^{B,I,\varass}&&= \left\{ \begin{array}{ll}\ltrue &\text{if }\ttt^{I,\varass} \in P^I\\\lfalse&\text{otherwise}\end{array}\right.\\
 &(\neg \varphi)^{B,I,\varass}&&=(\varphi^{B,I,\varass})^{-1}\\
 &(\varphi\land \psi)^{B,I,\varass}&&=\glb_{{\leqt}} (\varphi^{B,I,\varass},\psi^{B,I,\varass})\\
 &(\forall x:\varphi)^{B,I,\varass}&&=\glb_{{\leqt}} \{\varphi^{B,I,\varass [x:d]}\mid d \in D\}\\
 &(K\varphi)^{(P,S),I,\varass}&&=\left\{\begin{array}{ll}
                                  \Tr & \text{if } \varphi^{(P,S),I',\varass}=\Tr \text{ for all $I'\in P$}\\
                                  \Fa & \text{if } \varphi^{(P,S),I',\varass}=\Fa \text{ for some $I'\in S$}\\
                                  \Un & \text{otherwise}
                                 \end{array}\right.
\end{align*}
\end{definition}

The logical connectives combine three-valued truth values based on Kleene's truth tables (see Figure \ref{fig:KT}). As before, in case $\varphi$ is a sentence, $\varphi^{B,I,\varass}$ is independent of $\varass$ and we use $\varphi^{B,I}$ for $\varphi^{B,I,\varass}$ for any $\varass$. 

\begin{figure}
	\centering
\begin{minipage}{0.3\linewidth}
	\begin{tabular}{|cc"c|c|c|}
\cline{1-5}
\multicolumn{2}{|c"}{\multirow{2}{*}{{ $A\land B$} }}&\multicolumn{3}{c|}{B}\\

\multicolumn{2}{|c"}{}& \ltrue & \lfalse & \lunkn \\
\thickhline
\multirow{3}{*}{A}&\ltrue &	\ltrue & \lfalse & \lunkn	\\
\cline{2-5}
&\lfalse	& \lfalse & \lfalse & \lfalse		\\
\cline{2-5}
&\lunkn 	& \lunkn& \lfalse & \lunkn 		\\
\cline{1-5}
	\end{tabular}
\end{minipage}
\begin{minipage}{0.3\linewidth}
	\begin{tabular}{|cc"c|c|c|}
\cline{1-5}
\multicolumn{2}{|c"}{\multirow{2}{*}{{ $A\lor B$} }}&\multicolumn{3}{c|}{B}\\

\multicolumn{2}{|c"}{}& \ltrue & \lfalse & \lunkn \\
\thickhline
\multirow{3}{*}{A}&\ltrue &	\ltrue & \ltrue & \ltrue	\\
\cline{2-5}
&\lfalse	& \ltrue & \lfalse & \lunkn		\\
\cline{2-5}
&\lunkn 	& \ltrue& \lunkn & \lunkn 		\\
\cline{1-5}
	\end{tabular}
\end{minipage}
\begin{minipage}{0.2\linewidth}
	\begin{tabular}{|cc"c|}
\cline{1-3}
&&{{{ $ \lnot A$} }}\\

\thickhline
\multirow{3}{*}{A}&\ltrue &	\lfalse\\
\cline{2-3}
&\lfalse	& \ltrue	\\
\cline{2-3}
&\lunkn 	& \lunkn 		\\
\cline{1-3}
	\end{tabular}
\end{minipage}
\caption{The Kleene truth tables \cite{Kleene38}.}
\label{fig:KT}

\end{figure}

Let $T$ be a fixed AEL theory. \citetDMT{DeneckerMT00} defined the approximating operator $D^*_\theory$ that maps a belief pair $(P,S)$ to another belief pair $(P',S')$ where
\begin{align*}
P' = \{ I \mid \theory^{(P,S),I} \neq \Fa\} \text{ and }
S' = \{ I \mid \theory^{(P,S),I} = \Tr\}
\end{align*}
Intuitively, the new conservative bound contains all worlds in which the theory evaluates to true (with the current knowledge) and the new liberal bound all worlds in which $\theory$ does not evaluate to false. Thus 
$P'$ contains all knowledge that can \emph{certainly} be derived from the current state of mind and $Q'$ all knowledge that can \emph{possibly} be derived from it.
DMT showed that $D^*_\theory$ is an approximator of $D_\theory$. 
Hence, the operators induce a class of semantics for AEL: Moore's expansion semantics (supported fixpoints), Kripke-Kleene expansion semantics \citeDMT{aaai/DMT98} (Kripke-Kleene fixpoints), (partial) stable extension semantics ((partial) stable fixpoints) and well-founded extension semantics  (well-founded fixpoints) \citeDMT{DeneckerMT03}. The latter two were new semantics induced by  AFT.

%% file: chapters/syntax.tex
\label{sec:dael}
In this section, we describe the syntax and semantics of distributed autoepistemic logic. Theories in this logic describe the knowledge of a set of different agents. Throughout the rest of this paper, we assume a set of agents $\A$ to be fixed, with $\A$ a subset of the domain $D$ over which all structures are defined. The reason for this assumption is that it allows us to reuse the quantifications from first-order logic to quantify over the set of agents at hand. 
Furthermore, we assume that for each agent $A\in \A$, there is a constant $A\in \Sigma_o$, interpreted as $A$ in the objective structure $I_o$.

\subsection{Syntax and Basic Semantic Notions}
\begin{definition}
We define the language $\dael$ of distributed autoepistemic logic recursively as follows. \begin{align*}
 &P(\ttt) \in \dael &&\text{ if } P \text{ is an $n$-ary predicate in $\Sigma$ } \text{ and }\ttt\text{ an $n$-tuple of terms}\\
 &(\varphi \land \psi) \in \dael &&\text{ if } \varphi \in \dael \text{ and }\psi \in \dael\\
 &\lnot \varphi \in \dael &&\text{ if } \varphi \in \dael \\ 
 &\forall x: \varphi \in \dael&&\text{ if } \varphi \in \dael \\
 &K_t(\psi) \in \dael &&\text{ if }  \psi\in\dael  \text{ and $t \in \Terms$}
\end{align*}
\end{definition}
This definition consists of the standard recursive rules of first-order logic, augmented with a modal operator. 
The intuitive reading of $K_t(\psi)$ is ``$t$ is an agent and $t$ knows $\psi$''. Hence, if the term $t$ does not denote an agent, $K_t(\psi)$ will be interpreted to be false.

We assume that $\Sigma_o$ contains a dedicated unary predicate $\Apred$, whose interpretation is assumed to always be $\A$.

In a distributed setting, different agents each have their own theory describing their beliefs or knowledge about the world:
\begin{definition}
 A \emph{distributed theory} is an indexed family $\mc{T}=(\T_A)_{A\in \A}$ where each $\T_A$ is a set of sentences in $\dael$.
\end{definition}

\begin{example}\label{ex:vote}
Consider a situation with three agents, $A,B$ and $C$ who will vote openly on some issue. 
Agent $A$ decides to vote $\vyes$ if at least one of the other agents votes yes; otherwise he votes $\vno$.
Agent $B$ decides to follow the crowd: if all other agents are unanimous, she follows their vote. Otherwise, she abstains. 
Agent $C$ decides to vote $\vyes$ no matter what the other agents vote. 
The intended result of this vote is clear. $C$ votes \vyes, hence $A$ follows and in the end, the result is unanimous: every agent votes \vyes.

In dAEL, we model this situation as follows. We assume a single nullary predicate symbol $\vyes$ and use $K_A \vyes$ (respectively $K_A\lnot\vyes$) to denote the fact that agent $A$ votes \vyes (respectively \vno). In this example, we thus use $K_A\varphi$ to denote public announcements (not knowledge) of agents. 
Consider the following three theories. 
\begin{align*}
 \T_A &= \left\{\begin{array}{l}
			(K_B \vyes \lor K_C\vyes) \lequiv \vyes                      
                     \end{array}\right\}\\
 \T_B &= \left\{\begin{array}{l}
			(K_A \vyes \land K_C\vyes) \limplies \vyes\\                       
			(K_A \lnot \vyes \land K_C \lnot \vyes) \limplies \lnot \vyes                   
                     \end{array}\right\}\\
\T_C &= \{\vyes\}. 
\end{align*}
Now, $\T = (\T_A, \T_B, \T_C)$ is a distributed autoepistemic theory. As we shall show later (in Example \ref{ex:vote:cont}), all semantics we define for dAEL agree on this theory. Furthermore, its unique model equals the intended model sketched above, i.e., it is such that $K_A\vyes, K_B\vyes$ and $K_C\vyes$ are all true while $K_A \neg\vyes$, $K_B \neg\vyes$ and $K_C \neg\vyes$ are all false. 
%
 \end{example}

To represent the knowledge of multiple agents, we generalise the notion of a possible world structure:
\begin{definition}\label{def:upws}
A \textit{distributed possible world structure (DPWS)} is an indexed family $\upws=(\upws_A)_{A\in \A}$, where $\upws_A$ is a possible world structure for each $A\in\A$.
\end{definition}
The knowledge order can be extended pointwise to DPWSs.
One DPWS contains more knowledge than another if each agent has more knowledge: 
\begin{definition}\label{def:latticeupws}
 Given two DPWSs $\upws^1$ and $\upws^2$, we define $\upws^1 \lequpws \upws^2$ if $\upws^1_A\leqpws \upws^2_A$ for each $A\in \A$. 
\end{definition}
%
%

The value of a sentence is obtained like in AEL by evaluating each modal operator with respect to the right agent. 
\begin{definition}
\label{def:2val}
The \emph{value} of a sentence $\varphi$ with respect to a DPWS $\upws$, an interpretation $I$ and a variable assignment $\varass$ (denoted $\varphi^{\upws,I,\varass}$) is  defined inductively by the following recursive rules:
\begin{align*} 
& (P(\ttt))^{\upws,I,\varass} &&= \left\{ \begin{array}{l}\ltrue \text{ if } \ttt^{I,\varass} \in P^I\\\lfalse\text{ otherwise}\end{array}\right. \\
 &(\neg \varphi)^{\upws,I,\varass}&&=(\varphi^{\upws,I,\varass})^{-1}\\
 &(\varphi\land \psi)^{\upws,I,\varass}&&=\glb_{{\leqt}} (\varphi^{\upws,I,\varass},\psi^{\upws,I,\varass})\\
 &(\forall x:\varphi)^{\upws,I,\varass}&&=\glb_{{\leqt}} \{\varphi^{\upws,I,\varass [x:d]}\mid d \in D\}\\
 &(K_t\varphi)^{\upws,I,\varass}&&=\left\{\begin{array}{ll}
                                  \Tr & \text{if }  t^{I,\varass} \in \A \text{ and } \\ &\ \varphi^{\upws,J,\varass}=\Tr \text{ for each $J\in \upws_{t^{I,\varass}}$ }\\
                                  \Fa & \text{otherwise}
                                 \end{array}\right.
\end{align*}
\end{definition}
%
As before, if $\varphi$ is a sentence, $\varphi^{\upws,J,\varass}$ is independent of $\varass$ and we omit $\varass$ in the notation.

In order to generalise this valuation to a partial setting, we define a generalisation of belief pairs. 
\begin{definition}\label{ubp}
 A \emph{distributed belief pair} is a pair $\ubp = (\mc{P},\mc{S})$ of distributed possible world structures. 
 \end{definition}
 If $\ubp = (\mc{P},\mc{S})$ is a distributed belief pair, we denote the conservative bound $\mc{P}$ as $\ubp^c$ and the liberal bound $\mc{S}$ as $\ubp^l$. Furthermore, we use $\ubp_A$ to denote the belief pair $(\mc{P}_A,\mc{S}_A)$. 
The precision order on the approximating lattice is defined as usual, in the following definition.
\begin{definition}
If $\ubp^1$ and $\ubp^2$ are two distributed belief pairs, we say that  $\ubp^1$ is \emph{less precise} than  $\ubp^2$ if $\ubp^1_A\leqp  \ubp^2_A$ for each agent $A$. We denote this fact by  $\ubp^1 \leqp \ubp^2$.
\end{definition}
The following proposition follows easily from the equivalent result in AEL. 
\begin{proposition}
 The set of all DPWSs forms a complete lattice when equipped with the order $\leq_K$. The set of all distributed belief pairs forms a lattice when equipped with the order $\leqbp$. 
\end{proposition}
\begin{proof}
Follows immediately from the fact the order $\leq_K$ is the product order of the knowledge orders for each agent and the same for $\leqbp$.

\end{proof}

 As before, we restrict our attention to \emph{consistent} distributed belief pairs.
Note that for a set $\mc{S}$ of DPWSs, $$\lub_{\leq_K}(\mc{S}) = \left(\bigcap_{\upws\in\mc{S}} \upws_A \right)_{A\in\A}$$
The notion of three-valued valuations is extended to the distributed setting by evaluating each modal operator with respect to the correct agent. 
\begin{definition}\label{def:conlibvaluation}
The \emph{value} of $\varphi$ with respect to a distributed belief pair $\ubp$, an interpretation $I$ and a variable assignment $\varass$ (notation $\varphi^{\ubp,I,\varass}$) is defined inductively as follows:
\begin{align*}
 &(P(\ttt))^{B,I,\varass}&&= \left\{ \begin{array}{ll}\ltrue &\text{if }\ttt^{I,\varass} \in P^I\\\lfalse&\text{otherwise}\end{array}\right.\\
 &(\neg \varphi)^{B,I,\varass}&&=(\varphi^{B,I,\varass})^{-1}\\
 &(\varphi\land \psi)^{B,I,\varass}&&=\glb_{{\leqt}} (\varphi^{B,I,\varass},\psi^{B,I,\varass})\\
 &(\forall x:\varphi)^{B,I,\varass}&&=\glb_{{\leqt}} \{\varphi^{B,I,\varass [x:d]}\mid d \in D\}\\
&(K_t\varphi)^{\ubp,I,\varass}&&= \left\{\begin{array}{ll}
                                  \Tr & \text{if } t^{I,\varass} \in \A \text{ and } \\ &\  \varphi^{\ubp,J,\varass}=\Tr \text{ for all $J\in \ubp_{t^{I,\varass}}^c$}\\
                                  \Fa & \text{if } t^{I,\varass} \notin \A \text{ or } \\ &\  \varphi^{\ubp,J,\varass}=\Fa \text{ for some $J\in \ubp_{t^{I,\varass}}^l$}\\
                                  \Un & \text{otherwise}
                                 \end{array}\right.
\end{align*}
Note that this definition differs from the recursive definition of the three-valued valuation of an AEL formula only in the the fifth rule.

\end{definition}
This valuation essentially provides us with the means to apply AFT to lift the class of semantics of AEL to dAEL.

%% file: chapters/semantics.tex
\subsection{Semantics of dAEL through AFT}
\label{sec:sem}
Recall that we assume that a  structure $I_o$ interpreting the domain and all symbols in $\Sigma_o$ is fixed and hence shall not be repeated in all definitions.
The two- and three-valued valuations form the building blocks to extend the semantic operator and its approximator from AEL to dAEL.
\begin{definition}\label{def:beliefrevisionupws}
The knowledge revision operator for a distributed theory $\T$ is a mapping from the set of distributed possible world structures to itself, defined by
\begin{align*}
\upwsrevision(\upws)=(\{I\mid (\T_A)^{\upws,I}=\Tr \})_{A\in \A}
\end{align*}
\end{definition}
This revision operator revises the knowledge of all agents simultaneously, given their current states of mind. 
Fixpoints represent states of knowledge of the agents that cannot be revised any further. Or, in other words, distributed possible world structures that are consistent with the theories of all agents.
\begin{definition}
\label{def:appr}
The approximator for a distributed theory $\T$  on a distributed belief pair $\ubp$ is defined by
 $\ubprevision(\ubp)=(\ubprevisionL(\ubp),\ubprevisionU(\ubp)),$
 where
\begin{align*}
&\ubprevisionL(\ubp)=(\{I\mid   (\T_A)^{\ubp,I} \neq \Fa \})_{A\in \A}, \\
&\ubprevisionU(\ubp)=(\{I\mid   (\T_A)^{\ubp,I} = \Tr \})_{A\in \A}. 
\end{align*}
\end{definition}
\begin{theorem}
 $\ubprevision$ is an approximator of $\upwsrevision$.
\end{theorem}
\begin{proof}
 One can easily see from Definition \ref{def:conlibvaluation} that the valuation $\ubp\mapsto(\T_A)^{\ubp,I}$ is $\leqp$-monotone. This implies that  when $\ubp\leqp\ubp'$, $\ubprevisionL(\ubp)\supseteq \ubprevisionL(\ubp')$ and $\ubprevisionU(\ubp) \subseteq \ubprevisionU(\ubp')$, i.e.\ $\ubprevisionL(\ubp)\leqk \ubprevisionL(\ubp')$ and $\ubprevisionU(\ubp) \geqk \ubprevisionU(\ubp')$, i.e.\ $\ubprevision(\ubp) \leqp \ubprevision(\ubp')$. Thus \ubprevision is $\leqp$-monotone.

 The fact that $\ubprevision$ coincides with $\upwsrevision$ on two-valued belief pairs, follows from the fact that if $\ubp = (\upws,\upws)$, then $\T_A^{\ubp,I}=\T_A^{\upws,I}$.
\end{proof}

The stable operator $\upwsstrevision$ is defined for dAEL as $\upwsstrevision(\upws)=\lfp(\ubprevisionL(\cdot,\upws))$. 
Different fixpoints of these operators lead to different semantics as  discussed in Section \ref{ss:AFT};
\begin{definition}
 Let $\T$ be a distributed theory. 
  \begin{itemize}
  \item A \emph{supported model} of $\T$ with respect to $I_o$ is a fixpoint of $\upwsrevision$. 
  \item The \emph{Kripke-Kleene model} of $\T$ with respect to $I_o$ is the $\leq_p$-least fixpoint of $\ubprevision$. 
  \item A \emph{partial stable model} of $\T$ with respect to $I_o$ is a distributed belief pair $\ubp$, such that $\ubp^c=\upwsstrevision(\ubp^l)$ and $\ubp^l=\upwsstrevision(\ubp^c)$.
  \item A \emph{stable model} of $\T$ with respect to $I_o$ is a DPWS $\upws$, such that $(\upws,\upws)$ is a partial stable model of $\T$.
  \item The \emph{well-founded model} of $\T$ with respect to $I_o$ is the least precise partial stable model of $\T$.
 \end{itemize}
 
\end{definition}
We use the abbreviations \textsf{Sup}-\emph{model}, \textsf{KK}-\emph{model}, \textsf{PSt}-\emph{model}, \textsf{St}-\emph{model} and \textsf{WF}-\emph{model} to refer to these five kinds of models respectively.

\begin{example}[Example \ref{ex:vote} continued]\label{ex:vote:cont}
As discussed before, the intended result is that all three agents vote yes. This intended result corresponds to the following DPWS: 
 \[\left( \{ \{\vyes\}\}_A, \{ \{\vyes\}\}_B, \{ \{\vyes\}\}_C\right).\]
Note that in this DPWS, $K_A \vyes$, $K_B \vyes$ and $K_C \vyes$ are all true, while $K_A \neg\vyes$, $K_B \neg\vyes$ and $K_C \neg\vyes$ are all false. We now show that this DPWS is indeed the only model of $\T$. For this we first establish that this DPWS viewed as an exact distributed belief pair is the only fixpoint of~$\ubprevision$.

Let $\ubp$ be any distributed belief pair. Then
\begin{align*}
 \ubprevisionL(\ubp)_C &= \{I\mid (\T_C)^{\ubp,I} \neq \Fa \} \textnormal{ (by Definition \ref{def:appr})}\\
 &= \{I \mid \vyes^{\ubp,I} \neq \Fa \}\\
 &= \{I \mid I = \{\vyes\} \}\\
 &= \{\{\vyes\}\}.
\end{align*}

Similarly, $\ubprevisionU(\ubp)_C = \{\{\vyes\}\}$. 

Now suppose $\ubp$ is a fixpoint of \ubprevision. By the above, $\ubp_C = \{\{\vyes\},\{\vyes\}\}$, i.e.\ $(K_C \vyes)^\ubp = \Tr$ by Definition~\ref{def:conlibvaluation}, i.e.\ $(K_B \vyes \lor K_C \vyes)^\ubp = \Tr$. So for any fixpoint $\ubp$ of \ubprevision, we get
\begin{align*}
 \ubprevisionL(\ubp)_A &= \{I\mid (\T_A)^{\ubp,I} \neq \Fa \}\\
 &= \{I \mid ((K_B \vyes \lor K_C\vyes) \lequiv \vyes)^{\ubp,I} \neq \Fa \}\\
 &= \{I \mid \vyes^{\ubp,I} \neq \Fa\} \textnormal{ (since } (K_B \vyes \lor K_C \vyes)^\ubp = \Tr)\\
 &= \{I \mid I = \{\vyes\} \}\\
 &= \{\{\vyes\}\}.
\end{align*}

Similarly, $\ubprevisionU(\ubp)_A = \{\{\vyes\}\}$. Now from this we get that $(K_A \vyes)^\ubp = \Tr$, i.e.\ by the above we get that $(K_A \vyes \land K_C \vyes)^\ubp = \Tr$. By a similar derivation as above, we can then conclude that $\ubprevisionL(\ubp)_B = \ubprevisionU(\ubp)_B = \{\{\vyes\}\}$, i.e.\ $\ubp = ( \{ \{\vyes\}\}_A, \{ \{\vyes\}\}_B, \{ \{\vyes\}\}_C)$. Thus $(\{ \{\vyes\}\}_A, \{ \{\vyes\}\}_B, \{ \{\vyes\}\}_C)$ is indeed the only fixpoint of~$\ubprevision$.
 
Since it is the unique fixpoint of \ubprevision, it is the \textsf{KK}-\emph{model} and the unique \textsf{PSt}-\emph{model} and thus also the \textsf{WF}-\emph{model} of $\T$. Since this model is exact, it is also the unique \textsf{Sup}-\emph{model} and \textsf{St}-\emph{model} of \T. 
 Thus, we see that in this example, all our semantics coincide with the intended model. 
\end{example}

\begin{example}\label{ex:candy}
Suppose we have two agents, the mother and father of a six-year-old child: $\A=(M,D)$. A common situation is one where the child fancies candy and the father answers ``You can have some candy if it is okay for mom'', while the mother answers ``You can have candy if your father says so''. 
These statements can be modelled in dAEL as 
\begin{align*}
  \T_D& =\{ K_M(c) \Rightarrow c\} &\T_M&=\{ K_D(c) \Rightarrow c\}.
 \end{align*}

There exist four possible world structures for each agent:
 \begin{enumerate}
  \item The empty possible world set or inconsistent belief: $\emptyset$, denoted as $\top$.
  \item The belief of $c$: $\{\{c\}\}$, i.e., the fact that it follows from the public announcements made by the agent in question that the kid can have candy. 
  \item The disbelief of $c$: $\{\emptyset\}$, i.e., the fact that it follows from the public announcements made by the agent in question that the kid cannot have candy. 
  \item The lack of knowledge: $\{\emptyset,\{c\}\}$, denoted as $\bot$, i.e., the fact that no statements about being able to get candy follow from the announcements made by the agent in question.
 \end{enumerate}

The semantic operator associated to this theory is:
\begin{align*}
\upwsrevision(\upws)&=(\{I\mid (\T_D)^{\upws,I}=\Tr \}_D,\{I\mid (\T_M)^{\upws,I}=\Tr \}_M) \textnormal{ (by Definition \ref{def:beliefrevisionupws})}\\
&=(\{I\mid (K_M(c) \Rightarrow c)^{\upws,I}=\Tr \}_D,\{I\mid (K_D(c) \Rightarrow c)^{\upws,I}=\Tr \}_M)\\
&=(\{I\mid (K_M(c))^{\upws,I}=\Fa \textnormal{ or } c^{\upws,I}=\Tr \}_D,\{I\mid (K_D(c))^{\upws,I}=\Fa \textnormal{ or } c^{\upws,I}=\Tr \}_M)\\
&=(\{I\mid \emptyset \in \upws_M \textnormal{ or } I = \{c\} \}_D,\{I\mid \emptyset \in \upws_D \textnormal{ or } I = \{c\} \}_M) \textnormal{ (by Definition \ref{def:2val})}
\end{align*}

Therefore 
\begin{align*}
\upwsrevision(\upws)_D = 
\begin{cases}
 \{c\} & \textnormal{if } \emptyset \notin \upws_M;\\
 \bot & \textnormal{if } \emptyset \in \upws_M.\\
\end{cases}
\hspace{16mm}
\upwsrevision(\upws)_M = 
\begin{cases}
 \{c\} & \textnormal{if } \emptyset \notin \upws_D;\\
 \bot & \textnormal{if } \emptyset \in \upws_D.
\end{cases}
\end{align*}

From this it is obvious that there are two \emph{supported models}, namely $(\{\{c\}\}_D,\{\{c\}\}_M)$ and $(\bot_D,\bot_M)$. In the first model, $K_D c$ and $K_M c$ hold, i.e.\ Mom and Dad agree that the kid can have candy. In the second model, $K_D c$, $K_D \neg c$, $K_M c$ and $K_M \neg c$ are all false, i.e.\ Mom and Dad do not make any claims about the kid being allowed or disallowed to have candy.

For computing the other semantics, we need to determine the approximator $\ubprevision$. The first component of the approximator is:
\begin{align*}
\ubprevisionL(\ubp)&=(\{I\mid (\T_D)^{\ubp,I} \neq \Fa \}_D,\{I\mid (\T_M)^{\ubp,I} \neq \Fa \}_M)\\
&=(\{I\mid (K_M(c) \Rightarrow c)^{\ubp,I} \neq \Fa \}_D,\{I\mid (K_D(c) \Rightarrow c)^{\ubp,I} \neq \Fa \}_M)\\
&=(\{I\mid (K_M(c))^{\ubp,I} \neq \Tr \textnormal{ or } c^{\ubp,I} \neq \Fa \}_D,\{I\mid (K_D(c))^{\ubp,I} \neq \Tr \textnormal{ or } c^{\ubp,I} \neq \Fa \}_M)\\
&=(\{I\mid \emptyset \in \ubp^c_M \textnormal{ or } I = \{c\} \}_D,\{I\mid \emptyset \in \ubp^c_D \textnormal{ or } I = \{c\} \}_M)
\end{align*}

Therefore
\begin{align*}
\ubprevisionL(\ubp)_D = 
\begin{cases}
 \{c\} & \textnormal{if } \emptyset \notin \ubp^c_M;\\
 \bot & \textnormal{if } \emptyset \in \ubp^c_M.\\
\end{cases}
\hspace{16mm}
\ubprevisionL(\ubp)_M = 
\begin{cases}
 \{c\} & \textnormal{if } \emptyset \notin \ubp^c_D;\\
 \bot & \textnormal{if } \emptyset \in \ubp^c_D.
\end{cases}
\end{align*}


Similarly, for the second component of the approximator we get:
\begin{align*}
\ubprevisionU(\ubp)_D = 
\begin{cases}
 \{c\} & \textnormal{if } \emptyset \notin \ubp^l_M;\\
 \bot & \textnormal{if } \emptyset \in \ubp^l_M.\\
\end{cases}
\hspace{16mm}
\ubprevisionU(\ubp)_M = 
\begin{cases}
 \{c\} & \textnormal{if } \emptyset \notin \ubp^l_D;\\
 \bot & \textnormal{if } \emptyset \in \ubp^l_D.
\end{cases}
\end{align*}

The \emph{Kripke-Kleene model} can be computed by iterated applications of \ubprevision to $(\bot,\top)$ until a fixpoint is reached:
\begin{align*}
 \ubprevision(((\bot_D,\bot_M),(\top_D,\top_M))) &= ((\bot_D,\bot_M), (\{c\}_D,\{c\}_M))\\
 \ubprevision(((\bot_D,\bot_M), (\{c\}_D,\{c\}_M))) &= ((\bot_D,\bot_M), (\{c\}_D,\{c\}_M))
\end{align*}

Thus $((\bot_D,\bot_M), (\{c\}_D,\{c\}_M))$ is the Kripke-Kleene model of $\T$. In this model $K_D c$ and $K_M c$ have truth-value \Un, while $K_D \neg c$ and $K_M \neg c$ have truth value \Fa. So intuitively, it is undetermined whether Mom and Dad allow the kid to have candy, but they definitely do not disallow it.

Observation about \ubprevisionL: The value of $\ubprevisionU(\ubp^c,\ubp^l)$ depends only on $\ubp^c$; indeed, $\ubprevisionU(\ubp^c,\ubp^l) =  \upwsrevision(\ubp^c)$.

Now for any DPWS $\upws$, 
\begin{align*}
\upwsstrevision &= \lfp(\ubprevisionL(\cdot,\upws)) \\
&= \lfp(\upwsrevision) \textnormal{ (by the above observation about \ubprevisionL)}\\
&= (\bot_D,\bot_M)
\end{align*}

So the only \emph{partial stable} model is $(\bot,\bot)$. This is therefore also the \emph{well-founded model}. And since it is exact, it is also the only stable model. In this model, $K_D c$, $K_D \neg c$, $K_M c$ and $K_M \neg c$ are all false, i.e.\ Mom and Dad do not make any claims about the kid being allowed or disallowed to have candy. 
\end{example}

In the above example, it can be seen that supported models are very liberal in deriving knowledge, as knowledge may be supported by circular reasoning. For instance, the supported model $(\{\{c\}\}_D,\{\{c\}\}_M)$ essentially states that from the announcements of Mom and Dad, it follows that the kid is allowed to have candy. Whether this is a problematic interpretation in the case of this toy example may be debatable, but is is certainly not acceptable in access control: We do not want to allow access when the only reason to support the access is a circular justification that assumes that the access in question should be granted. Such circular justification could cause security problems! We will discuss an example of this in Section \ref{sec:AC}.

This criticism is similar to what has been said about Moore's original autoepistemic expansions (which are exactly the supported models). Other semantics, such as stable and well-founded semantics are more \emph{grounded} \mycite{GroundedFixpoints} in the sense that they derive only knowledge for which there is ground in the theory: knowledge is only derived if there is a non-self supporting reason. This is a reasonable way of deriving knowledge from the theories.


%% file: chapters/mapping.tex
%
%
%
\newcommand\tauform{\m{\tau_{\textit{formula}}}}
\newcommand\tautheo{\m{\tau_{\textit{theory}}}}
\newcommand\taustruc{\m{\tau_{\textit{structure}}}}
\newcommand\taupws{\m{\tau_{\textit{pws}}}}
\newcommand\taubp{\m{\tau_{\textit{beliefpair}}}}

We now describe a mapping from dAEL to AEL.
We prove that for distributed theories that do not contain any inconsistency, the semantics for dAEL match the corresponding semantics for AEL. 
In the case of partially inconsistent distributed theories, the semantics do not coincide: dAEL allows for a single agent to have inconsistent beliefs, whereas AEL has no mechanism to encapsulate an inconsistency in a similar way. 
This capacity of dAEL to encapsulate inconsistencies is a desirable feature, for instance in access control, where it facilitates to \emph{isolate} a faulty agent. We discuss this kind of isolation in detail in Section \ref{sec:faulty}.   

It has been noted before, by \citet{tocl/VennekensGD06/Fix} and \citet{VlaeminckVBD/KR2012} that natural embeddings of certain ``stratified'' languages in AEL fail when there is the possibility of inconsistent knowledge. 
They have presented the notion of \emph{permaconsistent} theories as a criterion for their embeddings to work. In this section, we show \begin{enumerate}
\item how to generalise permaconsistency to dAEL,
\item that for permaconsistent theories, our mapping indeed preserves semantics, and
\item that a weaker criterion (being universally consistent) works for supported, stable and partial stable semantics. 
\end{enumerate}

\ifthenelse{\boolean{showproofs}}{}{Due to space restrictions, proofs of intermediate lemmas are omitted (they are technical and do not provide new insights).}
We first present a generalisation of the notion of permaconsistency to the distributed case. 
\begin{definition}
 A distributed theory $\T$ is \emph{permaconsistent} if for each  $A\in \A$ and each theory $T'$ that can be constructed from $\T_A$ by replacing all occurrences of formulas $K_B\varphi$ not nested under a modal operator by $\ltrue$ or $\lfalse$, it holds that $T'$ has at least one model that expands $I_o$.
\end{definition}

The mapping from dAEL to AEL consists of a collection of translation functions, defined in Definitions \ref{def:tau_phi} to \ref{def:tau_B}. 
These functions translate syntactic constructs of dAEL like formulas and distributed theories and semantic constructs of dAEL like DPWSs and distributed belief pairs into the corresponding constructs of AEL. 
We denote each of these translation functions by $\tau$ with some subscript, where the subscript indicates the type of the output of the translation function. 

Given a vocabulary $\Sigma$ used for writing a distributed theory $\T$ in dAEL, the AEL translation of $\T$ will be written in a slightly modified vocabulary $\Sigma'$:

\begin{definition}
 Given a vocabulary $\Sigma$, we define $\Sigma'$ to be the vocabulary consisting of
 \begin{itemize}
  \item all symbols in $\Sigma_o$,
  \item all symbols in $\Sigma_s$, but with an arity increased by one. 
 \end{itemize}
\end{definition}

The additional argument of relation and function symbols in $\Sigma_s$ refers to the agent whose beliefs about the relation/function symbol we are using to interpret the symbol. 
Given an $n$-ary function symbol $f \in \Sigma$, we will therefore interpret $f$ as an $n\mathord{+}1$-ary function symbol in AEL, where $f(v_1,\dots,v_n,a)$ should be interpreted as the interpretation of $f(v_1,\dots,v_n)$ according to agent $a$. 
Our mapping from dAEL to AEL will be based on this intuition, namely in each formula, the extra argument will be used to represent the agent whose knowledge is referred to.
Since we assume all functions to be total, $f(a_1,\dots,a_n,a_{n+1})$ also needs to be interpreted when $a_{n+1}$ is not an agent.
Since it does not matter which value we give to $f(a_1,\dots,a_n,a_{n+1})$ in this case (this will follow from our particular translation), we fix an arbitrary element $\dummy$ in our domain $D$ to assign to such defective terms.
\begin{example}[Example \ref{ex:candy} continued]
 In this example, $\A=(M,D)$ and $\voc$ consists of a proposition symbol $c/0$ and two constant symbols, namely, $M/0,D/0$, where $M$ and $D$ refer to mommy and daddy and have a fixed interpretation (i.e., $M,D\in \Sigma_o$). 
 As such, $\voc'$ consists of one unary predicate symbol $c/1$ and the same function symbols as $\voc$. The intended interpretation of $c(d)$ is that the child can have candy according to $d$. 
\end{example}

We use the following notational conventions in this section: $\phi$ denotes an $\fo_d^\Sigma$-formula, $\varphi$ denotes an $\fo_k^{\Sigma'}$-formula, $I$ denotes a $\Sigma$-structure, and $J$ denotes a $\Sigma'$-structure.



\begin{definition}
 Given a $\Sigma$-term $t$ and a $\Sigma'$-term $s$, we define the $\Sigma'$-term $t_s$ recursively as follows:
 \begin{itemize}
  \item $x_s := x$ for each variable $x$
  \item $(f(t_1,\dots,t_n))_s := f({t_1}_s,\dots,{t_n}_s,s)$ for each $f \in \Sigma_s$
  \item $(f(t_1,\dots,t_n))_s := f({t_1}_s,\dots,{t_n}_s)$ for each $f \in \Sigma_o$
 \end{itemize}
\end{definition}

\begin{definition}
\label{def:tau_phi}
We define the function $\tauform:\Terms^{\Sigma'} \times \fo_d^\Sigma \rightarrow \fo_k^{\Sigma'}$ as follows:
\begin{itemize}
 \item $\tauform(s,P(t_1,\dots,t_n)) := P({t_1}_s,\dots,{t_n}_s,s)$
 \item $\tauform(s,\neg\phi) = \neg\tauform(s,\phi)$
 \item $\tauform(s,\phi \land \psi) = \tauform(s,\phi) \land \tauform(s,\psi)$
 \item $\tauform(s,\forall x : \phi) = \forall x : \tauform(s,\phi)$
 \item $\tauform(s,K_t \phi) = \exists x : (x=t_s \land \Apred(x) \land K \tauform(x,\phi))$ for a fresh variable $x$
\end{itemize}
\end{definition}

\begin{definition}
For a distributed theory $\T$, we define $\tautheo(\T) := \bigcup_{A\in\A} \tauform(A,\T_A)$.
\end{definition}

\begin{example}[Example \ref{ex:candy} continued]
 In this example, 
 \begin{align*}
 \tautheo(\T) &= \left\{\begin{array}{l}
 (\exists a: a=M \land \Apred(a) \land K c(a)) \Rightarrow c(D)\\
 (\exists a: a=D \land\Apred(a) \land K c(a)) \Rightarrow c(M) 
 \end{array}\right\}
\end{align*}
 
 Given that $M$ and $D$ are in $\Sigma_o$, i.e.\ are have a fixed interpretation in all models, this AEL theory is equivalent to the following simpler one:
 \begin{align*}
 \tautheo(\T) &= \left\{\begin{array}{l}
  K c(M) \Rightarrow c(D),\\
  K c(D) \Rightarrow c(M) 
 \end{array}\right\}
\end{align*}
 Intuitively, this state that if Mom says candy is allowed, so does Dad and vice versa, i.e.\ this theory contains the same knowledge as the original example. 
\end{example}


\begin{example}
 This example we illustrates why the AEL translation of $K_t \phi$ is $\exists x : (x=t_s \land \Apred(x) \land K \tauform(x,\phi))$ and not the simpler formula $\Apred(t_s) \land K \tauform(t_s,\phi)$. Consider the following dAEL theory $\T$:
 \begin{align*}
 \T_A &=\{(d = A \land p) \lor (d=B \land \neg p)\}\\
 \T_B &= \{p\}
\end{align*}
It can be easily verified that in all semantics defined in this paper, $\T$ has a unique model $\ubp$, and that ${(K_A K_d p)^{\ubp,I,a} = \Fa}$ for all $I,a$. Note that $\ubp$ is exact, i.e.\ of the form $(\upws,\upws)$, so what we just wrote about $K_A K_d p$ implies that ${(K_d p)^{\ubp,I,a} = \Fa}$ for some $I \in \upws$ and some variable assignment $a$.

The AEL translation $\tautheo(\T)$ of $\T$ is as follows:
\begin{align*}
 \tautheo(\T) &= \left\{\begin{array}{l}
 (d(A) = A \land p(A)) \lor (d(A)=B \land \neg p(A))\\
 p(B)
 \end{array}\right\}
\end{align*}
Again, all our semantics agree that this theory has a unique model $B$, and $B$ is exact, i.e.\ of the form $(Q,Q)$. Then $\tauform(A, K_d p)^{\ubp,I,a} = (\exists x : (x = d(A) \land \Apred(x) \land K p(x) ))^{B,I,a} = \Fa$ for some $I \in Q$ and some variable assignment $a$. This is in line with the above semantic analysis of $K_d p$.  On the other hand $(\Apred(d(A)) \land K p(d(A)))^{B,I,a} = \Tr$ for any $I,a$. This shows that the idea to use $\Apred(d(A)) \land K p(d(A))$ as the AEL translation of $K_d p$ would not work.
\end{example}

We now define a mapping of dAEL's semantic notions to AEL's semantic notions.

\begin{definition}
For an indexed family $\mathcal{I}=(I_A)_{A\in\A}$ of $\Sigma$-structures, we define the $\Sigma'$-structure $\taustruc(\mathcal{I})$ as follows: For each $n$-ary function symbol $f \in \Sigma_s$ and all $d_1,\dots,d_n\in D'$,
$$f^{\taustruc(\mathcal{I})}(d_1,\dots,d_n,d) := \begin{cases}    
                                                   f^{I_{d}}(d_1,\dots,d_n) &\textnormal{ if } d \in \A;\\
                                                   \dummy &\textnormal{ otherwise.}                                                                                                                       \end{cases}$$ 
                                                   
For each $n$-ary function symbol $f \in \Sigma_o$ and all $d_1,\dots,d_n\in D'$,
$$f^{\taustruc(\mathcal{I})}(d_1,\dots,d_n) := f^{I_o}(d_1,\dots,d_n).$$

For each $n$-ary relation symbol $R \in \Sigma_s$ and $d_1,\dots,d_n\in D'$,
\begin{align*}
 R^{\taustruc(\mathcal{I})}(d_1,\dots,d_n,d) \textnormal{ iff } d \in \A \textnormal{ and } R^{I_{d}}(d_1,\dots,d_n).
\end{align*}

For each $n$-ary relation symbol $R \in \Sigma_o$ and $d_1,\dots,d_n\in D'$,
$$R^{\taustruc(\mathcal{I})}(d_1,\dots,d_n) \textnormal{ iff } R^{I_o}(d_1,\dots,d_n).$$
\end{definition}

\begin{definition}
\label{def:tau_Q}
For a DPWS $\upws$, we define $\taupws(\upws):= \{\taustruc((I_A)_{A\in\A})\mid I_A \in \upws_A \textnormal{ for every } A \in \A\}$.
\end{definition}

\begin{definition}
\label{def:tau_B}
For a distributed belief pair $\ubp$, define $\taubp(\ubp) := (\taupws(\ubp^c),\taupws(\ubp^l))$. 
\end{definition}


The above mapping from dAEL to AEL preserves all semantics in case $\T$ is permaconsistent.
\begin{theorem}\label{thm:mapping:perm}
  Let $\sigma \in \{\textsf{Sup},\textsf{KK},\textsf{PSt},\textsf{St},\textsf{WF}\}$ be a semantics, let $\T$ be a permaconsistent distributed theory, and let $\ubp$ be a distributed belief pair. Then $\ubp$ is a $\sigma$-model of $\Tt$ iff $\taubp(\ubp)$ is a $\sigma$-model of $\tautheo(\Tt)$.
\end{theorem}

The proof of this theorem as well as the two theorems below is in the appendix.

We also present a weaker criterion that preserves models for three out of the five semantics. 

\begin{definition}
We call a DPWS $\upws$ \emph{universally consistent} if $\upws_A \neq \emptyset$ for all $A \in \A$. 
\end{definition}

\begin{definition}
We call a distributed belief pair $\ubp$ \emph{universally consistent} if $\ubp^l$ is universally consistent.
\end{definition}
Note that since $\ubp^l\subset\ubp^c$, if $\ubp$ is universally consistent, so is $\ubp^c$. 

\begin{definition}
 Let $\sigma \in \{\textsf{Sup},\textsf{KK},\textsf{PSt},\textsf{St},\textsf{WF}\}$ be a semantics. We call a distributed theory $\Tt$ \emph{universally consistent under $\sigma$} iff every $\sigma$-model of $\Tt$ is universally consistent.
\end{definition}


The following theorem states that the mapping from dAEL to AEL is faithful for universally consistent models of a distributed theory for three out of the five semantics. 

\begin{theorem}
\label{thm:mapping}
 Let $\sigma \in \{\textsf{Sup},\textsf{PSt},\textsf{St}\}$ be a semantics, let $\T$ be a distributed theory, and let $\ubp$ be a universally consistent distributed belief pair. Then $\ubp$ is a $\sigma$-model of $\Tt$ iff $\taubp(\ubp)$ is a $\sigma$-model of $\tautheo(\Tt)$.
\end{theorem}

The next theorem clarifies the relationship between permaconsistency and universal consistency. 

\begin{theorem}\label{thm:link}
Let $\sigma \in \{\textsf{Sup},\textsf{KK},\textsf{PSt},\textsf{St},\textsf{WF}\}$.  If $T$ is permaconsistent, then $T$ is universally consistent under $\sigma$.
\end{theorem}

\begin{example}
 The reverse of Theorem \ref{thm:link} does not hold as can be seen for example by a theory $\{p\limplies K_A p, K_A p \limplies p\}$ with one agent $A$. This theory is not permaconsistent because after replacing the first modal subformula by $\lfalse$ and the second by 
 $\ltrue$, we get 
 \[p\limplies \lfalse, \ltrue\limplies p,\] which clearly is not consistent. However, it is
universally consistent under the 3 mentioned semantics. E.g., The unique
stable model is $\{\{\},\{p\}\}$ which is universally stable.
\end{example}

\ignore{
\begin{example}
 Consider the following distributed theory:
 \begin{align*}
  T_A &= \{\}\\
  T_B &= \{p\lequiv \lnot K_A p\}\\
  T_C &= \{p\lequiv \lnot K_B p\}\\
 \end{align*}
This theory maps in AEL to ``Hanne's example''. We know that the well-founded model of the mapping is three-valued (not exact). (see my phd to see it worked out)

What is the well-founded model of this theory, viewed as a diststributed theory??? 

We start, from the distributed belief pair: 

\[ \bot = \left( (\{ \{p\},\emptyset\}, \emptyset)_A,(\{ \{p\},\emptyset\}, \emptyset)_B,(\{ \{p\},\emptyset\}, \emptyset)_C     \right) \]
i.e., for $A,B$ and $C$ it is unknown whether they Know p and whether they know $\lnot p$.

Since the theory of $A$ is trivially satisfied, $D^*$ maps this DBP to:
\[ B_1 = \left( (\{ \{p\},\emptyset\}, \{ \{p\},\emptyset\})_A,(\{ \{p\},\emptyset\}, \emptyset)_B,(\{ \{p\},\emptyset\}, \emptyset)_C     \right) \]

i.e., for $A$, it holds that $\lnot K_A p$, $\lnot K_A\lnot p$.

Again, $D^*$ maps $B_1$ to:

Since the theory of $A$ is trivially satisfied, $D^*$ maps this DBP to:
\[ B_2 = \left( (\{ \{p\},\emptyset\}, \{ \{p\},\emptyset\})_A,(\{ \{p\}\}, \{ \{p\}\})_B,(\{ \{p\},\emptyset\}, \emptyset)_C     \right) \]

i.e., it holds that $K_B p$ and $\lnot K_B \lnot p$. 

Again, $D^*$ maps $B_2$ to:
\[ B_3 = \left( (\{ \{p\},\emptyset\}, \{ \{p\},\emptyset\})_A,(\{ \{p\}\}, \{ \{p\}\})_B,(\{ \emptyset\}, \{ \emptyset\})_C     \right) \]

THUS! $B_3$ is exact. THis means that the least fixpoint of $D^*$ is exact and thus $B_3$ must equal the well-founded model and the KK model of our theory!

Also $B_3$ is universally consistent! all agents have consistsent knowledge!!!! Thus the theorem is wrong for at least two semantics!

\end{example}
}

%

%

Theorems \ref{thm:mapping:perm},  \ref{thm:mapping} and \ref{thm:link} are proven in the appendix.

%% file: chapters/AC.tex
In this section, we discuss application scenarios of dAEL that illustrate the motivations from Section \ref{sec:motiv} and give reasons for our claim that the well-founded semantics is particularly suitable for an application in access control. 

First, we show how a certain access control problem related to the revocation of delegated rights can be modelled in a natural and concise way in dAEL.

In ownership-based frameworks for access control, it is common to allow principals (users or processes) to grant both permissions and administrative rights to other principals in the system. Often it is desirable to grant a principal the right to further grant permissions and administrative rights to other principals. This may lead to delegation chains starting at a \emph{source of authority} (the owner of a resource) and passing on certain permissions to other principals \cite{Li,Tamassia,Chander04,Yao}.

For simplicity, we assume access right and delegation right always go hand in hand. In that case, one can recursively define access right for a resource $r$ as follows: 
\begin{itemize}
 \item The owner of $r$ always has access for $r$.
 \item If a principal $A$ with access right for $r$ has granted an authorization for resource $r$ to another principle $B$, then $B$ has access right for $r$.
\end{itemize}

Equivalently, one can say that a principle $A$ has access right for $r$ if there is a chain of authorizations for $r$ starting in the owner of $r$ and ending in principal $A$.

Furthermore, such frameworks commonly allow a principal to revoke a permission that she granted to another principal \cite{Hagstrom01,Zhang03,Chander04,logcom/BarkerBGG14}. Depending on the reasons for the revocation, different ways to treat the delegation chain can be desirable \cite{Hagstrom01,Cramer15,Cramer17}. Any algorithm that determines which permissions to keep intact and which ones to delete when revoking a permission is called a \emph{revocation scheme}. Of these revocation schemes, the one with the strongest effect is called the \emph{Strong Global Negative} (SGN) revocation scheme: In this scheme, revocation is performed by issuing a negative authorization which dominates over positive authorizations and whose effect propagates forward.

Semi-formally, the effect of an SGN revocation can be characterized recursively as follows:
\begin{itemize}
 \item The owner of $r$ always has access for $r$.
 \item If a principal $A$ with access right for $r$ has issued a positive authorization for resource $r$ to another principle $B$ and no principal with access right for $r$ has issued a negative authorization for $r$ to $B$, then $B$ has access right for $r$.
\end{itemize}

We illustrate the effect of SGN revocations the example depicted in Figure~\ref{fig:ac:0}: In this example, $A$ has issued positive authorizations to $B$ and $C$, $B$ has issued positive authorizations to $D$ and $E$, $E$ has issued a positive authorization to $F$, $C$ has issued a negative authorization to $D$, and $D$ has issued a negative authorization to $F$. Since $A$ is the owner of the resource, $A$ certainly has access by the first bullet item in the above semi-formal characterization of SGN. By the second bullet point, $B$ and $C$ have access, as $A$ has issued positive authorizations to them and no one has issued a negative authorization to them. Similarly, since $E$ have access, since $B$ has issued a positive authorization to $E$, an no one has issued a negative authorization to $E$. Since $C$ has access right and has issued a negative authorization to $D$, $D$ certainly does not have access right despite the positive authorization issued to $D$ by $B$. And since $D$ does not have access, the negative authorization from $D$ to $F$ has no effect, so the access that $E$ has granted to $F$ takes effect.

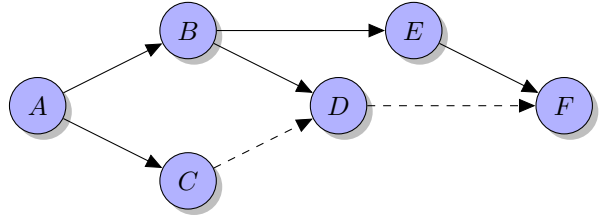
\begin{figure}
\centering
\begin{tikzpicture}[-triangle 45,node distance=2cm, auto,main node/.style={circle, draw=black,drop shadow,minimum size=0.75cm,fill=blue!30}]
 \node[main node] at (0, -1) (A) {$A$};
 \node[main node] at (2, 0) (B) {$B$};
 \node[main node] at (2, -2) (C) {$C$};
 \node[main node] at (4, -1) (D) {$D$};
 \node[main node] at (5, 0) (E) {$E$};
 \node[main node] at (7, -1) (F) {$F$};
 
 \path[]
   (A) edge node  {} (B)
   (A) edge node  {} (C)
   (B) edge node  {} (D)
   (C) edge [dashed] node  {} (D)
   (B) edge node  {} (E)
   (E) edge node  {} (F)
   (D) edge[dashed] node  {} (F);
 \end{tikzpicture}
\caption{First example scenario of SGN. Full arrows represent delegations (positive authorizations), dashed arrows revocations (negative authorizations). $A$ is the owner of the resource in question.}\label{fig:ac:0}
\end{figure}

In this example, the semi-formal characterization of SGN leads to clear results about who has acces and who does not. But this is not always the case. Consider for example the situation depicted in Figure~\ref{fig:ac:1}.

\begin{figure}
\centering
\begin{tikzpicture}[-triangle 45,node distance=2cm, auto,main node/.style={circle, draw=black,drop shadow,minimum size=0.75cm,fill=blue!30}]
 \node[main node] at (0, -1) (A) {$A$};
 \node[main node] at (2, 0) (B) {$B$};
 \node[main node] at (2, -2) (C) {$C$};
 \node[main node] at (4, -1) (D) {$D$};
 
 \path[]
   (A) edge node  {} (B)
   (A) edge node  {} (C)
   (B) edge node  {} (D)
   (C) edge node  {} (D)
   (B) edge[dashed, bend right] node  {} (C)
   (C) edge[dashed, bend right] node  {} (B);
 \end{tikzpicture}
\caption{Second example scenario of SGN. Full arrows represent delegations, dashed arrows revocations. $A$ is the owner of the resource in question. }\label{fig:ac:1}
\end{figure}
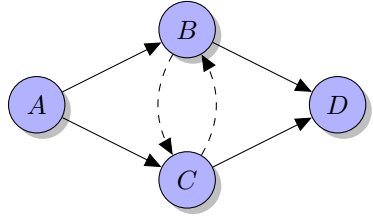

Here $B$ is attempting to revoke $C$'s access right (and vice versa). According to the above characterization of the effect of an SGN revocation, this attempt is only successful if $B$ has access. In other words, $C$ should have access if and only if $B$ does not have access. But since the scenario is symmetric between $B$ and $C$, they should either both be granted or both be denied access right. However, this cannot be achieved without violating the above characterization of SGN revocation. Paradoxical situations like this one can arise whenever the authorization graph contains a cycle that contains at least one negative authorization (one revocation).

Existing papers that have covered SGN revocation have handled this issue in different ways:
\begin{itemize}
 \item In the revocation framework of Hagstr\"om \textit{et al.}~[\citeyear{Hagstrom01}], this problem only arises when their \emph{Strong Global Negative revocation} is combined with a \emph{negative-takes-precedence} conflict resolution policy. In their paper, they do not describe in detail how Strong Global Negative revocation is supposed to work in the context of a \emph{negative-takes-precedence} conflict resolution policy. In other words, their paper implicitly implies the existence of such problematic scenarios, but does not expicitly discuss them.
 \item The paper by Cramer \textit{et al.}~[\citeyear{Cramer15}] is the first one to explicitly discuss this problem. The problem is circumvented by disallowing problematic authorization graphs (basically any authorization graph with a cycle that contains at least one negative authorization).
 \item The paper by Cramer and Casini~[\citeyear{Cramer17}] has a two-part inductive definition (Definitions 3 and 4) that directly corresponds to our above characterization of SGN revocation. In a footnote the paper specifies that this inductive definition is to be interpreted using the well-founded semantics for inductive definitions~\cite{Denecker98}. The paper points out that there exist paradoxical cases in which the well-founded model of the inductive definition is three-valued rather than  two-valued, so that for some principals it may be undecided whether they have access or not. The paper stipulate that in such cases \emph{undecided} is to be treated in the same way as \emph{false}, so that the principals directly affected by such a paradoxical situation will not have access until the paradoxical situation is resolved. (Applied to our above example this approach implies that formally the access right of $B$, $C$ and $D$ is \emph{undefined}, which practically means access gets denied for all three of them.)
\end{itemize}

We will now model delegation and SGN revocation in dAEL. When interpreted with the well-founded semantics for dAEL, this model of delegation and revocation is equivalent to the formalization by Cramer and Casini~\cite{Cramer17}. Furthermore, we will motivate why this behavior of SGN revocation corresponds better to general access control principles than the behavior that one would get if one used dAEL with another semantics than the well-founded semantics.

Our dAEL model of delegation and SGN revocation is based on statements issued by the various principals involved in a system: A principal $k$ can delegate access right to a principal $j$ by issuing the statement $\delegate(j)$, and can revoke access right from $j$ by issuing the statement $\revoke(j)$. We assume that the owner $A$ of the resource wants to ensure that these delegation and revocation statements are interpreted in line with our above characterization of SGN revocation. The owner $A$ can achieve this by issuing the following  statements as part of its theory (together with the $\delegate$ and $\revoke$ statements that $A$ makes):
\begin{align*}
&\access(A,r)\\
&
\left( \exists k \; (K_A \access(k,r) \land K_k \delegate(j)) \land \neg \exists i \; (K_A \access(i,r) \land K_i \revoke(j))\right) 
\limplies \access(j,r)
\end{align*}

Now access is to be granted to a principal $k$ if and only if the owner $A$ believes the statement $\access(k,r)$, i.e., if and only if $K_A\access(k,r)$ holds in the well-founded model of the distributed theory given by the above base theory of owner $A$ and all the statements issued by the various principals in the system. 

We illustrate this using our first  example scenario from Figure \ref{fig:ac:0}. 
Given that $A$ is the owner of the resource in question, the distributed theory that represents the authorizations present in this example scenario is as follows:
\begin{align*}
 \T_A &= \left\{\begin{array}{l}
			\access(A,r)\\
			\left( \exists k \; (K_A \access(k,r) \land K_k \delegate(j)) \land \neg \exists i \; (K_A \access(i,r) \land K_i \revoke(j))\right)\ \limplies \access(j,r)\\
			\delegate(B)\\
			\delegate(C)
                     \end{array}\right\}\\
 \T_B &= \left\{\begin{array}{l}
			\delegate(D)\\
			\delegate(E)
                     \end{array}\right\}\\
 \T_C &= \left\{\begin{array}{l}
			\revoke(D)
                     \end{array}\right\}\\
 \T_D &= \left\{\begin{array}{l}
			\revoke(F)
                     \end{array}\right\}\\
 \T_E &= \left\{\begin{array}{l}
			\delegate(F)
                     \end{array}\right\}\\
 \T_F &= \left\{\right\}\\
\end{align*}

Let $\ubp_\textsf{WF}$ be the well-founded model of $\T$. By formalizing the informal reasoning about this example that we presented above, one can show that $\ubp_\textsf{WF}$ assigns $\Tr$ to the statements $K_A \access(A,r)$, $K_A \access(B,r)$, $K_A \access(C,r)$, $ K_A \access(E,r)$, $K_A \access(F,r)$, while it assigns $\Fa$ to the statement $K_A \access(D,r)$. Therefore $A$, $B$, $C$, $E$ and $F$ will be granted access to resource $r$ and $D$ will be denied access to it. 

In this application of dAEL, the information that we are interested in from a given model is only the information about which truth-values the model assigns to statements of the form $K_A \access(X,r)$, i.e. which agents are given access to the recourse by the owner ($A$). For this reason, we present the relevant information as
a set of expressions $X^t$ where $X$ is a principal and $t$ the truth value of $K_A \access(X,r)$ in the model. 
So in the above example, we would say that the well-founded model $\ubp_\textsf{WF}$ satisfies $\{A^\Tr,B^\Tr,C^\Tr,D^\Fa,E^\Tr,F^\Tr\}$.\footnote{Note that this presentation of a model does not present all the information that is in the model, only the one that is relevant for our discussion about the access control application presented here. But in fact, this information suffices for figuring out the entire models, for more details on this, see the proof of Theorem \ref{thm:poly}} 

For the above example, the other semantics presented in Section \ref{sec:sem} give the same results as the well-founded semantics. However, that is not always the case. In the cases when the various semantics differ, the well-founded semantics is the only one that ensures that decisions about access are \emph{grounded} \cite{ai/BogaertsVD15}, meaning that derivable formulas are supported by cycle-free justifications, and that they satisfy a security principle that has been worded by Garg [\citeyear{Garg09}] as follows: ``When access is granted to a principal $k$, it should be known where $k$'s authority comes from''. For this reason, we consider the well-founded semantics to be preferable to the other semantics for applications of dAEL to access control. We will now illustrate these desirable feature of the well-founded semantics through two example scenarios. 


First, let us consider again the scenario depicted in Figure \ref{fig:ac:1}. Here $A$ is the owner of the resource $r$ and has issued the statements $\delegate(B)$ and $\delegate(C)$, $B$ has issued the statements $\revoke(C)$ and $\delegate(D)$, and that $C$ has issued the statements $\revoke(B)$ and $\delegate(D)$. As explained above, attempting to apply the semi-formal characterization of SGN revocation to this scenario leads to paradoxical arguments about the access rights of $B$ and $C$. So we may say that the scenario contains a conflict that cannot be automatically resolved. At this point, $A$ as the principal with control over $r$ will have to manually resolve the conflict by removing access from at least one of $B$ and $C$ depending on the cause for the conflict between them. 

In practice, it may take $A$ some time to study the situation and perform this manual resolution.
During this time, the system should still respond to access requests.
The intended behavior is that neither $B$ nor $C$ should have access, 
to avoid  security risks. 
The situation for $D$ is less clear: Given that $D$ would have access no matter who of $B$ and $C$ has access, one could make a case for granting $D$ access in this situation.

However, granting access right to $D$ would violate the security principle mentioned above: ``When access is granted to a principal $k$, it should be known where $k$'s authority comes from''~\cite{Garg09}. 

Now consider the statements issued by the principals as a distributed theory $\T$, with the two statements governing access included in the theory of the resource owner $A$. This theory has different models depending on the choice of semantics. 
 There are two supported models satisfying $\{A^{\ltrue},B^{\ltrue},C^{\lfalse},D^{\ltrue}\}$ and $\{A^{\ltrue},B^{\lfalse}, C^{\ltrue},D^{\ltrue}\}$ respectively. These are also the stable models. 
 The Kripke-Kleene model and the well-founded model are identical and satisfy $\{A^{\ltrue},B^{\lunkn},C^{\lunkn},D^{\lunkn}\}$. This model is not exact: the truth-value of the statements $K_A \access(X,r)$, with $X\in\{B,C,D\}$ is unknown.
 
When there is more than one model, the only safe approach in the access control application is to merge the information from the multiple models in a skeptical way, i.e.\ to grant access only if each model justifies granting access. According to this principle, the supported model semantics and stable semantics lead to access being granted to $A$ and $D$ in this scenario. Given our above argument against granting access to $D$, this means that these semantics cannot be considered viable semantics for this application of dAEL. Furthermore, note that when the skeptical way of combining information from multiple models is applied to the partial stable semantics, the result is always the same as the result of the well-founded semantics. For this reason, we do not consider the partial stable semantics separately in this section.

The Kripke-Kleene and well-founded model of this theory gives access precisely to the principal that should have access according to our above discussion.
Furthermore, it exhibits the existing conflict between $B$ and $C$ by making their access right status undefined. 

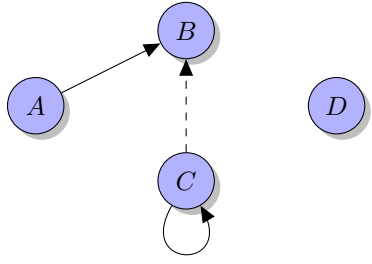
\begin{figure}
\centering
\begin{tikzpicture}[-triangle 45,node distance=2cm, auto,main node/.style={circle, draw=black,drop shadow,minimum size=0.75cm,fill=blue!30}]
 \node[main node] at (0, -1) (A) {$A$};
 \node[main node] at (2, 0) (B) {$B$};
 \node[main node] at (2, -2) (C) {$C$};
 \node[main node] at (4, -1) (D) {$D$};
 
 \path[]
   (A) edge node  {} (B)
   (C) edge[out=-120,in=-60,distance=10mm]  node {} (C)
   (C) edge[dashed] node {} (B);
%
%
%
%
%
 \end{tikzpicture}
\caption{Illustration of the third scenario for SGN. Full arrows represent delegations, dashed arrows revocations. $A$ is the owner of the resource in question. }\label{fig:ac:2}
\end{figure}

Now consider a third scenario as depicted in Figure \ref{fig:ac:2}, in which the resource owner $A$ has issued the statement $\delegate(B)$ and $C$ has issued the statements $\delegate(C)$ and $\revoke(B)$. Here $C$ should clearly not have access, because the only principal granting her access is $C$ herself. 
Hence $C$'s revocation of $B$'s access right does not have any effect, so $B$ should be granted access.
The Kripke-Kleene model of the distributed theory corresponding to this scenario is not exact; it satisfies $\{A^{\ltrue},B^{\lunkn},C^{\lunkn},D^{\lfalse}\}$. In this model, it is unknown whether $B$ and $C$ have access; this clearly diverges from our requirements. 
The well-founded model on the other hand correctly computes this desired outcome: it satisfies  $\{A^{\ltrue},B^{\ltrue},C^{\lfalse},D^{\lfalse}\}$.
The reason why the well-founded semantics leads to a better outcome than the Kripke-Kleene semantics is that it is \emph{grounded} \cite{ai/BogaertsVD15}, meaning that derivable formulas are supported by cycle-free justifications.

From these scenarios, we can see that the only semantics for dAEL that behaves as desired in the access control application is the well-founded semantics.
%
These findings are in line with the findings of \citet{nonmon30/DeneckerMT11}, who strongly argued in favour of the well-founded semantics of AEL. 

\subsection{Faulty agents}\label{sec:faulty}
In a distributed setting, it can happen that one of the agents either deliberately or accidentally \emph{fails}, i.e.\ has a theory that -- together with additional
information present in the system -- implies a contradiction. 
It is to be expected that such a failure has at least some influence on the rest of the system. However, the hope is that the rest of the system does not suffer too much from a failure of a single agent. In this section, we show how dAEL isolates faulty agents and contrast it to what happens when standard AEL is used to model a multi-agent scenario using the translation from Section \ref{sec:mapping}.
Consider again the principals from example \ref{ex:postdoc}. 
$A$ is a professor with theory 
\[ \theory_A = \left\{\begin{array}{l} \access(A,r).\\
\access(C,r).\\
\lnot K_C \lnot \access(B,r)\limplies \access(B,r) 
\end{array}\right\}.\]
Now, first we will consider a situation in which the PhD student $B$ is faulty (making inconsistent claims). I.e., consider the following theories
\[
 \theory_B^1 = \{\access(B,r) \land \lnot \access(B,r) \},\quad \theory_C^1 = \{\}
\]
and the distributed autoepistemic theory 
\[\theory^1 = (\theory_A, \theory_B^1, \theory_C^1).\]
Under all of the semantics we defined for dAEL, a model 
\[(\pws_A,\pws_B,\pws_C)\]
will have the property that $\pws_B=\top$, i.e., that the agent $B$ has inconsistent knowledge. This is to be expected since the theory $\theory_B^1$ that describes his knowledge is inconsistent. 
The interesting thing to investigate is how this inconsistency affects the other agents' knowledge. Luckily, it doesn't!  
In this example supported, Kripke-Kleene, well-founded, partial stable and stable semantics all agree that the unique model is given by
\begin{align*}
 \pws_A &= \{ \{\access(A,r),\access(C,r),\access(B,r)\}\}\\
 \pws_B &= \top \\
 \pws_C &= \bot
\end{align*}
That is, $A$ knows that everyone can access resource $r$, $C$ makes no claims about access to resources, while $B$ has inconsistent knowledge. 
This example is one of the situations where the mapping from dAEL to AEL does \emph{not} preserve semantics. Indeed, AEL has no mechanisms to isolate inconsistencies. If an AEL theory contains an inconsistency, this always results in a globally inconsistent possible world structure. 

Now, let us consider another variation of the same theory, namely with 
\[
 \theory_B^2 = \{\},\quad \theory_C^2 = \{\access(B,r) \land \lnot \access(B,r) \}
\]
and the distributed autoepistemic theory 
\[\theory^2 = (\theory_A, \theory_B^2, \theory_C^2).\]
I.e., we now consider what happens if $C$ is a faulty agent. 
Now, all semantics for dAEL agree that the unique model is given by
\begin{align*}
 \pws_A &= \left\{\begin{array}{l} \{\access(A,r),\access(C,r)\},\{\access(A,r),\access(B,r),\access(C,r)\}\end{array}\right\}\\
 \pws_B &= \bot \\
 \pws_C &= \top
\end{align*}
That is, $A$ knows that $A$ and $C$ have access to the resource $r$ and that it does not follow that $B$ has access. 
In this example, $C$ has inconsistent knowledge and $B$ has no knowledge. 
Thus, in this case, we can see that the inconsistency in $C$'s theory \emph{does} influence the knowledge of other agents. Indeed, $A$ observes that $K_C\lnot \access(B,r)$ holds and thus the last constraint no longer entails access of $B$. However, it only influence knowledge of other players at places where they explicitly refer to the faulty agent. If there were a fourth agent, say another postdoc $D$ and $\theory_A$ also contains the constraint $\access(D,r)$, the result would be that $A$ grants $D$ access, regardless of inconsistencies in other agent's knowledge. 
Thus, we conclude that our proposed formalism manages to isolate faulty agents as desired. 

This desirable behaviour with respect to faulty agents is the same behaviour that other \say-based acess control logic such as BL \cite{Garg12} exhibit. The point we made in this subsection is that if one tried to use standard AEL as an access control logic by modelling the multi-agent features using the translation from Section \ref{sec:mapping}, one would not get this desirable behaviour concerning faulty agents, so that the extension of AEL to dAEL is really necessary for the access control application.

%

%% file: chapters/procedure.tex
We now study complexity of reasoning in dAEL. 
Given the argumentation in the previous section, we focus on the well-founded semantics. More particularly, we are concerned with the the following decision task: 
\begin{task}\label{task:dec}
 Given a finite set of agents $\A$, a finite $\Sigma_o$-structure $I_o$ with domain $D$, a  finite dAEL theory $\theory$ and a sentence $\varphi$ of the form $K_t\psi$ with $t$ a $\Sigma_o$ term, determine if $\varphi$ holds in the well-founded model of $\theory$.
\end{task}
This task is well-defined: since $t$ is a $\Sigma_o$ term, $\varphi$ can be evaluated in a belief pair. 
Stated in words, we are interested in evaluating whether a certain formula ($\psi$) holds in the knowledge of a certain agent (represented here by the term $t$). In the context of access control, this decision problem is indeed the one we are interested in: there the formula is $\psi$ typically of the form $acces(b,r)$ and the term $t$ is typically $owner(r)$. I.e., there we wish to query whether according to  the owner of a given recourse $r$, a certain agent has access to that resource. 

More concretely, we will be interested in the \emph{data complexity} of this task, i.e., all complexity results will be for a fixed $T$ and $\varphi$, and thus are measured in terms of the size of the domain of $I_o$. 

After the publication of the conference version of this paper, 
\citet{Ambrossio19} already defined a query-driven decision procedure for dAEL that tackles exactly Task \ref{task:dec}. Their decision procedure is designed in such a way that it allows one to determine access rights while avoiding redundant information flow between principals in order to enhance security and reduce privacy concerns. 
Their decision procedure is query-driven in the following sense: A query in the form of a dAEL formula $\varphi$ is posed to a principal $A$. $A$ determines whether her theory contains enough information to verify $\varphi$. It can happen that $A$ cannot verify $\varphi$ just on the basis of her theory, but can determine that if a certain other principal supports a certain formula, her theory implies the query. For example, $A$'s theory may contain the formula $K_Bp \Rightarrow \varphi$. In this case, $A$ can forward a remote sub-query to $B$ concerning the status of $p$ in $B$'s theory. If $B$ verifies the sub-query $p$ and informs $A$ about this, $A$ can complete her verification of the original query $\varphi$.

In this generation of subqueries, loops may occur. For this reason, the decision procedure includes a loop detection mechanism. When a loop is detected, the query causing the loop (by being identical to a query that is an ancestor of it in the call graph) is labelled either with $\Fa$ or $\Un$, depending on whether the loop is over a negation or not. The details are described in \cite{Ambrossio19}.

Keeping the distributed theory $\T$ fixed and varying the size of the domain, this decision procedure for dAEL and its restriction to dAEL have a worst case runtime that is exponential in the size of the domain.
From a practical perspective, this is not an encouraging result.  
In the following theorem, we show that this complexity is not a coincidence.

\begin{theorem}
 Task \ref{task:dec} is NP-hard and co-NP-hard.
\end{theorem}
\begin{proof}
 
 It is well known that the graph coloring problem (the problem of determining whether a given graph is colorable by a given set of colors) is NP-complete. In fact, this is one of Karp's original 21 NP-complete problems \cite{Kar72}.  
We reduce both this problem and its negation to the Task \ref{task:dec}. 
 
 Consider a set of agents $\A=\{a,b\}$, the vocabulary $\Sigma_o$  with two constants $a,b$ and predicates $Node/1$, $Color/1$, $Edge/2$. 
   Also consider the vocabulary $\Sigma_s$ consisting of  predicate symbols $Coloring/2$ (with the informal interpretation that $Coloring(n,c)$ holds if node $n$ is colored with color $c$) and $p/0$. Furthermore, as before, let $\Sigma$ denote $\Sigma_o\cup\Sigma_s$. 
%
 Let $\varphi$ denote the first-order $\Sigma$-formula
 \begin{align*}
 &(\forall n: Node(n)\limplies  \exists c: Color(c)\land  Coloring(n,c))\,\land\\
 &(\forall n1, n2: Edge(n_1,n_2)\limplies \lnot \exists c: Coloring(n1,c)\land Coloring(n2,c)).
 \end{align*}
 It can be seen that $\varphi$ holds in an interpretation $I$ if and only if $\mathit{Coloring}^I$ is a coloring of $Edge^I$ with the colors $Color^I$. 
 Let $T_a$ be the empty theory  and
 \[T_b = \{p\lequiv K_a\lnot \varphi\}.\]
 Now consider $\theory=(T_a,T_b)$. 
 For any graph $(V,E)$ and set of colors $C$, let $I_{(V,E),C}$ denote the $\Sigma_o$ interpretation with domain $\{a,b\}\cup V$ interpreting $a$ as $a$, $b$ as $b$ and $Edge$ as $E$, $Node$ as $V$, and $Color$ as $C$. 
 
 Now, we claim that, given a $\Sigma_o$-interpretation $I_o$ 
 \begin{itemize}
  \item $K_a\lnot \varphi$ holds in the well-founded model of $\theory$ under $I_{(V,E),C}$ if and only if $ (V,E)$ is not colorable with $C$
  \item For any graph $(V,E)$ and set of colors $C$, $K_b p$ holds in the well-founded model of $\theory$ under $I_{(V,E),C}$ if and only if $ (V,E)$ admits no coloring with $C$, and
  \item $K_b \lnot p$ holds in the well-founded model of $\theory$ under $I_{(V,E),C}$ if and only if $(V,E)$ admits a coloring with $C$.
 \end{itemize}
 Let $\upws$ denote the well-founded model of $\theory$. First of all, we note that $\upws_a = \bot$, i.e.\ $\upws_a$ is the set of all $\Sigma$-interpretations with domain  $\{a,b\}\cup V$ that  coincide with $I_{(V,E),C}$ on $\Sigma_o$. 
 This is easy to see, since $a$ has no knowledge whatsoever. Hence, the only way it can know $\neg\varphi$ is if $(V,E)$ admits no coloring with $C$. 
 From this fact, the first claim easily follows. 
 Now, the other two claims follow from the first, since in $b$'s theory, $p$ is defined to be $K_a\lnot\varphi$. 
\end{proof}
It can be seen from the previous proof that the result also holds for the Kripke-Kleene model.

Given this complexity result, one might wonder how a logic like dAEL can be useful in practice, for instance for applications like access control. To this end, we develop a fragment of our logic, for which Task \ref{task:dec} can be solved in polynomial time. 
This restriction of dAEL has to be chosen in a careful way in order to find a good balance between expressivity and efficiency. In this subsection, we present one reasonable choice, called \daelr, from the rich space of potential restrictions of dAEL, and show that for this fragment, the task we are interested in, has polynomial data complexity. 

In \daelr, the theories of the various agents may not contain arbitrary formulas, but only formulas that follow a certain syntax akin to that of rules in a logic program. We define these \emph{rule formulas} as follows:

%

\begin{definition}
A \emph{modal literal} is a dAEL formula of the form $K_A l$ or $\lnot K_A l$ where $A$ is an agent and $l$ is a literal (an atom or its negation).
\end{definition}

\begin{definition}
 A \emph{modal complex} is a dAEL formula in which every atom is a subformula of a modal literal.
\end{definition}

Note that for a model complex $\varphi$, the interpretation $\varphi^{\ubp,I,a}$ does not depend on $I$, so we may also write $\varphi^{\ubp,a}$ for it.

\begin{definition}
A \emph{rule formula} is a dAEL formula of the form $\forall \xxx: (\varphi \Rightarrow l)$, where $\varphi$ is a modal complex and $l$ is a literal. 
A \daelr theory is a \dael theory in which all formulas are rule formulas. 
\end{definition}

While \daelr is significantly more restrictive than \dael, we believe that for most access control applications it is sufficient. First of all, note that the access control application of dAEL discussed in Section \ref{sec:AC} lies fully within \daelr. As a further example, it seems in principle to be possible to use \daelr to handle the detailed case study that \citet{Garg09} presented in order to illustrate the applicability of his access control logic BL.\footnote{The BL theory that formalizes the access control policy of this case study consists of formulas of the form $A \textit{ claims } (\varphi_1 \land \dots \land \varphi_n \land B_1 \textit{ says } \psi_1 \land \dots \land B_m \textit{ says } \psi_m \Rightarrow \chi)$. Given the intuitionistic nature of BL discussed in Section \ref{sec:related AC} below, a BL theory consisting of such formulas leads to the same access control decisions as the \daelr distributed theory $\T$ such that for each formula of the form just mentioned, there is a corresponding rule formula $K_A \varphi_1 \land \dots \land K_A \varphi_n \land K_{B_1}  \psi_1 \land \dots \land K_{B_m} \psi_m \Rightarrow \chi$ in $\T_A$.}

Restricting our attention to \daelr, the task we are interested in has polynomial data complexity:

\begin{theorem}\label{thm:poly}
If each formula in $T$ is a rule formula, then 
 Task \ref{task:dec} can be performed in polynomial time.
\end{theorem}
\begin{proof}
 The clue to this proof is that in the special case where $T$ only consists of rule formulas, only a small subset of all the possible world structures (and belief pairs) is relevant, in the sense that all others need not be considered in order to compute the well-founded model of $T$. Evaluating $\varphi$ in this model is then also efficient. 
 
 To this end, we call a distributed possible world structure \upws \emph{literal-determined} if for each agent $A$, there exists a set of ground literals $L_A$ such that $\upws_A$ is the least informative possible world structure in which all literals $L_A$ are known. Stated differently 
 \[\upws_A = \{ I \mid I \models l \text{ for all $l\in L_A$}\}.\]
 We call a distributed belief pair \emph{literal-determined} if both its possible world structures are.

 \textbf{Claim:} \textit{If $T$ consists only of rule formulas, then for each distributed belief pair \ubp, it holds that $\ubprevision(\ubp)$ is literal-determined}. 
 To see that this claim indeed holds, note that given a distributed belief pair \ubp and an agent $A$, it holds that 
 \begin{align*}
 &\ubprevisionL(\ubp)_A=\{I\mid   (\T_A)^{\ubp,I} \neq \Fa \} \\
 &=\{I\mid   (\forall \xxx: (\varphi \Rightarrow l))^{\ubp,I} \neq \Fa \text{ for all rule formulas $\forall \xxx: (\varphi \Rightarrow l)$ in $\T_A$} \}  \\
&= \{I\mid  I \models l[\xxx:\ddd] \text{ for all rule formulas $\forall \xxx: (\varphi \Rightarrow l)$ and all instantiations $\xxx:\ddd$ for which } \varphi^{\ubp,[\xxx:\ddd]}\geqt \lunkn\}
 \end{align*}
 I.e., that this set is literal-determined. An analogous argument yields that also $\ubprevisionU(\ubp)_A$ is literal-determined. 
 
 Now, from this claim, it follows that a well-founded induction exists that only uses literal-determined distributed belief pairs (indeed, an example of such induction is the maximal one). 
 Each increasing (in precision) chain of distributed belief pairs has only polynomial length (there are only polynomally many ground literals).
 Furthermore, for each literal-determined distributed belief pair $\ubp$,  $\ubprevision(\ubp)_A$ can be computed in polynomial time (this is easy to see by the equation in the proof of our claim: indeed, it suffices to evaluate the modal complexes $\varphi$ occurring in all rule formulas for all instantiations $\ddd$ of the $\xxx$). 
 From this, it follows that a well-founded induction can be constructed in polynomial time, and hence, the well-founded model of such a theory can be computed in polynomial time.
 Finally, in order to execute Task \ref{task:dec}, we need to evaluate a single query in the well-founded model, also that consists of simply evaluating a single formula and hence, Task \ref{task:dec} is indeed polynomial.
 \end{proof}

 The above proof actually suggests a mechanism to implement reasoning engines for \daelr based on logic programming. 
 While (e.g., using grounding), this idea applies in the general setting, let us illustrate it in the propositional case. 
 In other words, we assume that we are given a theory $ T$ in which all formulas are of the form $\varphi\limplies l$ with $\varphi$ a propositional modal complex.
 Now, given such a \daelr theory over $\voc$, we can construct a logic-programming vocabulary $\voc'$ consisting of symbols
 \[ \voc' := \{p_A^+, p_A^- \mid p\in \voc, A \in \A\},\]
 where intuitively $p_A^+$ means that agent $A$ knows that $p$ holds, and $p_A^-$ means that agent $A$ knows that $p$ is false. 
 It is not hard to see that there is a canonical one-to-one mapping $\mu$ between interpretations of $\voc'$ and literal-determined possible world structures of $\voc$. 
 This mapping $\mu$ is extended pointwise to three- or four-valued interpretation, which it maps to distributed belief pairs. 
 Moreover, we can construct a logic program $\PP$ over $\voc'$ such that for every three-valued $\voc'$ interpretation $\struct$, 
 \[ \mu(\Psi_\PP(\struct)) = \ubprevision(\mu(\struct)),\]
 i.e., the mapping between $\voc'$-interpretations and literal-determined possible world structures of $\voc$ preserves the approximator.
 It then follows directly that $\mu$ maps the well-founded model of $\PP$ to the well-founded model of $\ubprevision$. 
 The precise way to derive $\PP$ from $T$ is: 
 \begin{itemize}
  \item In each formula $\varphi\limplies l$ in $T$, push all negations inwards in $\varphi$ using de Morgan's laws (making them appear only in front of atoms or modal operators). 
  \item In each formula $\varphi\limplies l$ in $T$, replace each literal $\lnot p$ in $\varphi$ by $p^-_B$ and each literal $p$ by $p^+_B$ where $B$, where $K_B$ is the innermost knowledge operator in the scope of which this literal appears
  \item Subsequently, drop all knowledge operators from $\varphi$, turning $\varphi$ in propositional formula over $\voc'$,
  \item Finally, replace the formula $\varphi\limplies l$ occurring in the theory of agent $A$ by a rule $l_A \lrule \varphi$ (where $l_A$ stands for $p_A^+$ if $l=p$ and $p_A^-$ if $l=\lnot p$.   
 \end{itemize}
 Let us illustrate this by means of an example.
 \begin{example}
  Consider the following \daelr theory 
  \begin{align*}
   T_A &= \{K_B \lnot (\lnot p \lor \lnot K_A \lnot q) \limplies p.\qquad K_B q\limplies \lnot q.\}\\
   T_B &= \{q. \qquad K_Ap \limplies p.\}
  \end{align*}
  corresponds to the following logic program:
  \begin{align*}
   p^+_A &\lrule p_B^+\land q_A^-\\
   q^-_A &\lrule q_B^+\\
   q_B^+.&\\
   p_B^+ &\lrule p_A^+.
   \end{align*}
  In its well-founded model, $q_B^+$ and $q_A^-$ are true. All other atoms are false. 
  This corresponds to the possible world structure in which agent $B$ knows $q$, and agent $A$ knows $\lnot q$ and no-one knows anything else. 
 \end{example}

%

%% file: chapters/related.tex
In this section, we discuss two kinds of related work: In Subsection \ref{sec:mAELs}, we present other multi-agent extensions of autoepistemic logic that have been proposed in the literature, and compare them to dAEL. In Subsection \ref{sec:related AC}, we discuss approaches in access control logic related to ours.

\subsection{Other multi-agent extensions of AEL}
\label{sec:mAELs}
Several extensions of autoepistemic logic, and other non-monotonic reasoning formalisms to the multi-agent case have been made \cite{aaai/Morgenstern90,ijcai/BelleL15,clima/ToyamaKI02,wocfai/PermpoontanalarpJ95}. 
Each of them starts from a particular dialect of the non-monotonic logic and generalizes it to multiple agents.
\citet{aaai/Morgenstern90} made an extension to Moore's AEL~\cite{Moore85} and studied a centralized theory containing statements about the knowledge of different agents. She does not consider distributed theories and does not assume introspection. \citet{ijcai/BelleL15} also studied multi-agent theories in the same setting but added \emph{only knowing} and \emph{common knowledge} constructs. 
\citet{clima/ToyamaKI02} developed a distributed variant on autoepistemic logic that also assumes introspection. Compared to our logic, it is quite limited in the sense that it is propositional and only introduces one of the many semantics we discussed, namely the supported model semantics (which corresponds to Moore's original expansions); as such, it also easily encountered the kind of problems with groundedness and cyclic support the original AEL suffered from \cite{nato/HalpernM85,Konolige88,phd/Bogaerts15}.
\citet{wocfai/PermpoontanalarpJ95} studied a number of logics and developed a proof theory that extends the logic of Morgenstern. 
Their main motivation is that the logic of Morgenstern  has some undesirable properties if reduced to the single agent case, where it differs from AEL.
Our logic on the other hand, when instantiated with only one agent, exactly coincides with AEL. 
\citet{VlaeminckVBD/KR2012} defined two extensions to AEL with multiple agents, namely ordered epistemic logic (OEL) and distributed ordered epistemic logic (dOEL). 
Both of these logics require a partial order on the agents, where agents can only refer to knowledge of agents strictly lower in the order. If we add this restriction to our logic, we get exactly dOEL, i.e., dOEL is the fragment of dAEL for which there exists a stratification on the agents such that agents only refer to knowledge of ``lower'' agents. 
For such theories, all AFT semantics coincide and are equal to the semantics of dOEL as defined by \citet{VlaeminckVBD/KR2012}. 
The logic OEL is close to dOEL, with  the difference being that in OEL an agents knows everything any agent lower in the order knows. 
This behavior can be simulated in dAEL by adding the axiom scheme $K_A\phi\limplies \phi$ to the theory of each agent greater than $A$ in the order.    
In the context of an application to access control, the restriction of dOEL and OEL that a a global stratification on the agents is required is undesirable for a truly distributed system. 

Our most important contribution with respect to other approaches that define multi-agent extensions of AEL is that we present a uniform, fundamental principle to lift various of those dialects to the multi-agent case using AFT. 
In this paper, we already lift 5 dialects, and it easily extends to more semantics. 
We can use the same approach to lift the family of \emph{ultimate} semantics \cite{DeneckerMT00},  \emph{(partial) grounded fixpoint semantics} \cite{ai/BogaertsVD15,ijcai/BogaertsVD15}, \emph{well-founded set semantics} \cite{tocl/BogaertsVD16}, \emph{conflict-freeness}, \emph{$M$-stable semantics} and \emph{$L$-stable semantics} \cite{journals/ai/Strass13} from AEL to dAEL.  
This approach not only allows us to lift many semantics, it also provides a uniform principle for \emph{comparing} various semantics and hence it brings \emph{order} in the zoo of semantics for multi-agent AEL.

\subsection{Related approaches in access control logic}
\label{sec:related AC}

Most access control logics proposed in the literature have been defined in a proof-theoretical way, i.e., by specifying which axioms and inference rules they satisfy. This contrasts with our approach of defining dAEL model-theoretically rather than proof-theoretically. 
Our main motivation for defining dAEL model-theoretically is that model-theoretic definitions are more basic: from a model-theoretic definition, a notion of entailment, and hence a proof-theoretic characterization can be derived, but not the other way around.
We have already motivated the application of autoepistemic logic to access control in section \ref{sec:autoepistemic}, and the use of the well-founded semantics in section \ref{sec:use-cases}.

\citet{Garg2008} and 
\citet{Genovese12} have defined Kripke semantics for many of the access control logics discussed in the literature. However, these semantics are not meant to specify the meaning of the \say-modality, but to be a tool for defining decision procedures for those access control logics. This contrasts with our approach of studying the meaning of the \say-modality by showing that its intended use in access control justifies an application of the semantic principles of autoepistemic logic.

\citet{Clarkson13} have defined a so-called \emph{belief semantics} as well as a standard Kripke semantics for their access control logic FOCAL, arguing that the belief semantics corresponds better than the Kripke semantics to how principals reason in real-world systems. However FOCAL does not support mutual positive or negative introspection between principals, making it difficult to naturally model both delegation and denial.

We are not aware of any other \say-based access control logic that allows to model the non-monotonic behavior of denials as straightforwardly as dAEL by allowing to derive formulas of the form $\neg k \says \phi$ and supporting mutual negative introspection between principals. However, most state-of-the-art access control logics allow for mutual positive introspection between principals. For example BL, an access control logic with support for system state and explicit time, supports mutual positive introspection \cite{Garg09,Garg12}. 

The only approach to access control based on a non-monotonic logical formalism that we are aware of is the unifying access control meta-model proposed by 
\citet{sacmat/Barker2009}. This proposed meta-model is based on a rule language interpreted using Clark's completion, a non-monotonic logic programming semantics. Unlike dAEL, Barker's meta-model is not designed for distributed access control. If it is extended to support distributed access control policies and used in a straightforward way to implement our example, its behavior would correspond to the behavior that dAEL would have with the supported model semantics, which we have shown in section \ref{sec:use-cases} to give undesirable results. 

 \citet{Secommunity} describe \emph{Secommunity}, a framework for distributed access control based on Barker's access control meta-model. Secommunity is implemented in the Answer Set Programming system DLV, which works with the stable semantics. Thus in this distributed framework based on the Barker's meta-model, Clark's completion semantics has been replaced by stable semantics. But as described in section \ref{sec:use-cases}, the stable semantics also gives undesirable results when applied to a standard access control problem.
 To the best of our knowledge, Barker's meta-model has never been used under the well-founded semantics. 
 We expect that, given the strong correspondences between logic programming and autoepistemic logic, induced by AFT, there will be a close relation between such a usage of Barker's model and our logic dAEL. Researching this is a topic for future work. 


\paragraph{Classical vs. intuitionistic logic}

Many state-of-the-art access control logics are based on intuitionistic rather than classical logic.
\citet{Garg09} justifies the use of intuitionistic logic in access control
on the basis of the security principle that when access is granted to a principal $k$, it should be known where $k$'s authority comes from. 
Autoepistemic logic, on the other hand, is based on classical logic. However, in Section \ref{sec:use-cases} we argued that under the well-founded semantics, with its constructive semantics, this security principle is still satisfied.

Another justification for the use of intuitionistic logic has been put forward by \citet{Abadi08}, who gives an overview over the design-space of access control logics, discussing advantages and disadvantages of certain axioms. One axiom discussed by Abadi is the Unit axiom $\phi \Rightarrow k \says \phi$. Note that Unit implies mutual full introspection (i.e.\ mutual positive and mutual negative introspection) between principals, but is strictly stronger than mutual full introspection. Abadi showed that in classical logic, Unit implies Escalation, i.e.\ the property that $k \says \phi$ implies that either $\phi$ or $k \says \bot$. \citet{Abadi08} argued that Escalation embodies a rather degenerate interpretation of
the \say-modality, because it means that if a statement supported by a principal is actually false, the principal can be considered to support all statements.

Since dAEL builds on top of classical logic, we need to discuss the status of Unit and Escalation in dAEL. In dAEL, there is no objective interpretation of objective formulas, i.e., formulas without the \say-modality. 
All we have is an introspective agent's interpretations of formulas, from which we can derive an objective interpretation of formulas of the form $k \says \psi$, or, in our notation $K_k\psi$. Hence, when $\phi$ is an objective formula, neither Unit nor Escalation can be evaluated objectively in AEL. What we can do instead is to ask the following two questions:
\begin{enumerate}
 \item Do Unit and Escalation hold for a formula $\phi$ of the form $K_k \psi$ or $\neg K_k \psi$? I.e.,  are the following formulas tautologies for all agents $j$ and $k$ and every formula $\psi$? 
 \begin{align*}
&   K_k\psi \limplies K_jK_k\psi &\text{(unit)}\\
&   \lnot K_k\psi \limplies K_j \lnot K_k\psi&\text{(unit)}\\
& K_j K_k\psi \limplies (K_k\psi \lor K_j \bot)&\text{(escalation)}\\
& K_j \lnot K_k\psi \limplies (\lnot K_k\psi \lor K_j \bot)&\text{(escalation)}
 \end{align*}
 \item Do Unit and Escalation hold within the belief of an agent? I.e., are the following formulas tautologies? 
 \begin{align*}
  &K_k(\psi \limplies K_j\psi) &\text{(unit)}\\
  &K_k   (   K_j \psi \limplies (\psi \lor K_k \bot))&\text{(escalation)}
 \end{align*}
\end{enumerate} 

The answer to question 1 is ``yes''. Indeed, Unit means that the agents in question have introspection in each other's knowledge. Escalation also holds, but simply due to the fact that $K_j K_k \psi$ and $K_k\psi$ are equivalent in our logic, as can easily be seen from the truth evaluation. In this case, escalation boils down to stating that it is impossible for an agent $j$ to consistently make the claim that another agent $k$ said something $k$ did not actually say. 

The answer to question 2 is ``no'': Principal $A$ can know $p$, but not know that principal $B$ knows $p$ (unit). Also, $A$ can know that $B$ claims some property holds (say, that $B$ has access to a recourse) and in the meanwhile $A$ can claim that $B$ does not have access to this recourse without thereby implying that $B$'s claims are inconsistent.

%% file: chapters/conclusion.tex
\label{sec:conclusion}
Motivated by an application in access control, we have extended AEL to a distributed setting, resulting in a logic called distributed autoepistemic logic (dAEL).
dAEL allows for a set of agents to each have their own theory in which they refer to each others knowledge.
For this, the knowledge operator $K$ of AEL is replaced by an indexed operator $K_A$, where $A$ refers to an agent. 
We have defined the semantics of this logic building on approximation fixpoint theory (AFT), a lattice-theoretic framework that captures the semantics of many non-monotonic logics. 
Using AFT has many practical advantages: first of all, it allows for a uniform lifting of many different semantics. 
Secondly, it ensures that all fundamental principles underlying these semantics remain preserved.
And third, in doing so, we immediately obtain access to a wide variety of theoretic results. For instance, properties such as the fact that the well-founded model approximates all stable models was obtained \emph{by definition}, since the corresponding result holds in the algebraic setting. Similarly, we can (but did not do so) apply algebraic stratification results \cite{tocl/VennekensGD06,tocl/BogaertsVD16}, predicate introduction results \cite{VennekensMWD07}, modularity results \cite{amai/Truszczynski06}, or results on constructive characterisations of semantics \cite{ai/BogaertsVD18} without effort. Also future progress in AFT will be directly applicable.  

We have illustrated how dAEL can be applied to access control and have argued that one semantics is particularly suitable for modelling access control policies, namely the well-founded semantics. The non-monotonic behaviour of dAEL allowed us to model denial and revocation of access in dAEL, something that previous access control logics could not achieve due to their monotonicity. We have thus built a bridge between non-monotonic logic and access control. One of the tasks left for future research is to study whether this bridge may lead to further fruitful interaction between these fields additionally to the one already considered in this paper.

We have studied the complexity of reasoning with the well-founded semantics of dAEL and came to the unsettling conclusion that complexity of the considered task is quite high (both NP and coNP hard), thus making it unpractical for use in an access control setting. 
To overcome this limitation, we defined a fragment of our logic, which we called \daelr, in which reasoning becomes polynomial and that suffices to model the kind of application that motivated the paper in the first place.

%% file: chapters/appendix.tex
\section{Proofs of Theorems \ref{thm:mapping:perm},  \ref{thm:mapping} and \ref{thm:link}}

In order to prove Theorems \ref{thm:mapping:perm},  \ref{thm:mapping} and \ref{thm:link}, we first need to define some additional notions and prove some lemmas. 

\begin{lemma}\label{lem:perm}
 If $\T$ is permaconsistent, then $D^*_\T(\bot,\top)$ is universally consistent. 
\end{lemma}
\aproof{\begin{proof}
         If $\T$ is permaconsistent, then for each agent $A\in \A$ and each theory $T'$ that can be constructed from $\T_A$ by replacing all non-nested occurrences of modal literals by $\ltrue$ or $\lfalse$ is consistent. 
         Now, for each agent $A$, let $\T_A'$ the theory constructed from $\T_A$ by replacing all non-nested occurrences of modal literals by $\ltrue$ if they occur in a negative context (under an odd number of negations) and by $\lfalse$ otherwise. This theory is clearly stronger than $\T_A$. 
         Since $\T$ is permaconsistent, $\T_A'$ is satisfiable, let $I_A$ be a model of $\T_A'$. 
         In this case, it holds that 
         $\T_A^{(\bot,\top),I_A} = \ltrue$ (since $\T_A$ is weaker than $\T_A'$). 
         
         From this, we find that for each agent $A$, $\{I\mid \T_A^{(\bot,\top), I}\}$ is non-empty and thus that $D^*(\bot,\top)$ is indeed universally consistent.
        \end{proof}
}

\begin{proof}[Proof of Theorem \ref{thm:link}]
 Suppose $\T$ is permaconsistent. By Lemma \ref{lem:perm}, $D^*_\T(\bot,\top)$ is universally consistent. It follows directly from the definitions in AFT that each model of $\T$ (under any of the semantics), is more precise than $D^*_\T(\bot,\top)$. Furthermore, if $\ubp'\geqp\ubp$ and $\ubp$ is universally consistent, then so is $\ubp'$.
\end{proof}

\begin{definition}
Given a $\Sigma'$-structure $J$ and a $\Sigma'$-term $t$, we write $J_t$ for the $\Sigma$-structure defined by $s^{J_t}(d_1,\dots,d_n) := s^{J}(d_1,\dots,d_n,t^J)$ for every $s \in \Sigma$.
\end{definition}

The following lemma states that for an indexed family of structures, the mapping does not discard any information, i.e., after applying the mapping, we can recover each agent's structure. 
\begin{lemma}
 \label{lem:J_A}
 Let $(I_A)_{A\in\A}$ be an indexed family of $\Sigma$-structures, and let $J = \taustruc((I_A)_{A\in\A})$. Then $J_A = I_A$.
\end{lemma}
\begin{proof}
 Let $s \in \Sigma$. Then $s^{J_{A}}(d_1,\ldots,d_n)=s^{J}(d_1,\ldots,d_n,A)=s^{I_{A}}(d_1,\ldots,d_n)$. 
\end{proof}

The following lemma generalizes Lemma \ref{lem:J_A} to DPWSs. 

\begin{lemma}
	\label{DAEL:lem:transBP}
	Let $\upws$ be a universally consistent DPWS, and let $A \in \A$. Then 
	$$\{J_{A}\mid J\in \taupws(\upws)\}=\upws_{A}$$ for each $A\in\A$.
\end{lemma}
\begin{proof}
	We prove the equality by proving the subset relation in both directions.
	
	Let $J\in \taupws(\upws)$. Then there is an indexed family $(I_{A'})_{A'\in \A}$ s.t. $I_{A'}\in \upws_{A'}$ for each $A'\in \A$ and $J=\taustruc((I_{A'})_{A'\in \A})$.
	Then by Lemma \ref{lem:J_A}, $J_{A}=I_{A}\in \upws_{A}$, as required.
	
	To prove the other direction, let $I\in \upws_{A}$. Since $\upws$ is universally consistent, there is some indexed family $(I_{A'})_{A'\in \A}$ s.t. $I_{A}=I$ and $I_{A'}\in \upws_{A'}$ for all $A'\in \A$.
	Define $J:=\taustruc((I_{A'})_{A'\in\A})$. Note that $J \in \taupws(\upws)$. Now $J_A=I_A$ by Lemma \ref{lem:J_A}, so $J_A=I$, as required.
\end{proof}

The following lemma says that the mapping is faithful to the valuations of AEL and dAEL formulas:

\begin{lemma}
\label{lem:mapval}
For a $\Sigma'$-term $t$, a formula $\phi \in \fo_d^\Sigma$, a universally consistent distributed belief pair $\ubp$ and a $\Sigma'$-strucutre $J$, 
$$\tauform(t,\phi)^{\taubp(\ubp),J} = \phi^{\ubp,J_t}.$$
\end{lemma}
\aproof{\begin{proof}We prove the lemma by induction over the structure of $\phi$. \\
		\noindent \underline{$\phi = P(t_1,\dots,t_n)$.} Then $\tauform(t,\phi)^{\taubp(\ubp),J} = \Tr$ 
		
		iff $P({t_1}_t,\dots,{t_n}_t,t)^{\taubp(\ubp),J} = \Tr$
		
		iff $({t_1}_t^{J},\dots,{t_n}_t^{J},t^J) \in P^{J}$
		
		iff $(t_1^{J_t},\dots,t_n^{J_t},t^J) \in P^{J_t}$
		
		iff $\phi^{\ubp,J_t} = \Tr$.
		
		\noindent
		Analogously, we find $\tauform(t,\phi)^{\taubp(\ubp),J} = \Fa$ iff $\phi^{\ubp,J_t} =\Fa$.
		
		\noindent \underline{$\phi = \neg \psi$.} Then $\tauform(t,\phi)^{\taubp(\ubp),J} = \lnot \tauform(t,\psi)^{\taubp(\ubp),J}$ 
		and the result follows by the induction hypothesis.
		%
		%
		%
		%
		
		\noindent \underline{$\phi = \phi_1 \land \phi_2$.} Similarly.
		
		\noindent \underline{$\phi = \forall x : \psi$.} Similarly.
		
		\noindent \underline{$\phi = K_s \psi$.} Then $\tauform(t,\phi)^{\taubp(\ubp),J} = \Tr$ 
		
		iff $\exists x : (x=s_t \land \Apred(x) \land K \tauform(x,\psi))^{\taubp(\ubp),J} = \Tr$
		
		iff there is a $d \in \A$  with $d=s_t^J$ such that  for each $J' \in \taupws(\ubp^c)$, $\tauform(d,\psi)^{\taubp(\ubp),J'} = \Tr$ (by Definition \ref{def:AEL3val})
		
		iff $s^{J_t} \in \A$ and for each $J' \in \taupws(\ubp^c)$, $\tauform(s^{J_t},\psi)^{\taubp(\ubp),J'} = \Tr$ (since $s_t^J=s^{J_t}$)
		
		iff $s^{J_t} \in \A$ and for each $J' \in \taupws(\ubp^c)$, $\psi^{\ubp,J'_{s^{J_t}}} = \Tr$ (by the induction hypothesis)
		
		iff $s^{J_t} \in \A$ and for each $I \in \ubp^c_{s^{J_t}}$, $\psi^{\ubp,I} = \Tr$ (since $\ubp$ is universally consistent, using Lemma \ref{DAEL:lem:transBP})
		
		iff $\phi^{\ubp,J_t} = \Tr$ (by Definition \ref{def:conlibvaluation}).
		
		\noindent Analogously, we find $\tauform(t,\phi)^{\taubp(\ubp),J} = \Fa$ iff $\phi^{\ubp,J_t} =\Fa$.
		
\end{proof}}

The following lemma states that the dAEL approximator $\ubprevision$ is mapped to the AEL approximator $\ubprev{\tautheo(\Tt)}$, when restricted to universally consistent distributed belief pairs:

\begin{lemma}
\label{lem:ubprev}
 For every distributed theory $\Tt$ and every universally consistent distributed belief pair $\ubp$, 
 $$\taubp(\ubprev{\Tt}(\ubp)) = \ubprev{\tautheo(\Tt)}(\taubp(\ubp)).$$
\end{lemma}
\aproof{\begin{proof}
\begin{align*}
\taubp(\ubprev{\T}(\ubp))& = (\taupws((\{I \mid \phi^{\ubp,I} = \Tr \textnormal{ for each } \phi \in \T_A\})_{A\in\A}),\\
& \hspace{6.1mm} \taupws((\{I \mid \phi^{\ubp,I} \neq \Fa \textnormal{ for each } \phi \in \T_A\})_{A\in\A}))\\
& = (\{J \mid \phi^{\ubp,J_A} = \Tr \textnormal{ for each } A\in\A \textnormal{ and } \phi \in \T_A\},\\
& \hspace{6.1mm} \{J \mid \phi^{\ubp,J_A} \neq \Fa \textnormal{ for each } A\in\A \textnormal{ and } \phi \in \T_A\})\\
& \hspace{6.1mm} \textnormal{(by Definition \ref{def:tau_Q} and Lemma \ref{lem:J_A})}\\
& = (\{J \mid \tauform(A,\phi)^{\taubp(\ubp),J} = \Tr \textnormal{ for each } A\in\A\\
& \hspace{11.9mm} \textnormal{and } \phi \in \T_A\},\\
& \hspace{6.1mm} \{J \mid \tauform(A,\phi)^{\taubp(\ubp),J} \neq \Fa \textnormal{ for each } A\in\A \\
& \hspace{11.9mm} \textnormal{and } \phi \in \T_A\}) \textnormal{ (by Lemma \ref{lem:mapval})}\\
& = (\{J \mid \varphi^{\taubp(\ubp),J} = \Tr \textnormal{ for each } \varphi \in \tautheo(\Tt)\},\\
& \hspace{6.1mm} \{J \mid \varphi^{\taubp(\ubp),J} \neq \Fa \textnormal{ for each } \varphi \in \tautheo(\Tt)\})\\
& = \ubprev{\tautheo(\Tt)}(\taubp(\ubp))
\end{align*}\qedhere
\end{proof}}

%

The mapping maps the dAEL knowledge revision operator $\upwsrevision$ to the corresponding AEL knowledge revision operator $\mc{D}_T$:
\begin{lemma}
\label{lem:upwsrev}
  For every distributed theory $\Tt$ and every DPWS $\upws$, 
 $$\taupws(\upwsrev{\Tt}(\upws)) = \upwsrev{\tautheo(\Tt)}(\taupws(\upws)).$$
\end{lemma}
\aproof{\begin{proof}
Follows from Lemma \ref{lem:ubprev} and the fact that $ D^*_\theory(\upws,\upws)=(D_\theory(\upws),D_\theory(\upws))$ for each $\upws$.
%
\end{proof}}

The mapping is faithful to the (universal) consistency of (distributed) possible world structures:

\begin{lemma}
\label{lem:cons}
  A DPWS $\upws$ is universally consistent iff $\taupws(\upws) \neq \emptyset$.
\end{lemma}
\aproof{\begin{proof}
 Trivial.
\end{proof}}

The following lemma states that the restriction of $\taupws$ to universally consistent DPWSs is injective:

\begin{lemma}
\label{lem:inj}
 If $\upws$ and $\upws'$ are DPWSs such that $\upws$ is universally consistent and $\taupws(\upws) = \taupws(\upws')$, then $\upws = \upws'$.
\end{lemma}
\aproof{\begin{proof}
 By Lemma \ref{lem:cons}, $\upws'$ is universally consistent too. By symmetry, it is enough to show that $\upws_A \subseteq \upws'_A$ for all $A \in \A$. 
 
 So let $I_A \in \upws_A$. Given that $\upws$ is universally consistent, we can choose an $I_B \in \upws_B$ for every $B \in \A \setminus \{A\}$. Then $\taustruc((I_A)_{A\in\A}) \in \taupws(\upws) = \taupws(\upws')$, so $I_A \in \upws'_A$, as required.
\end{proof}}

The following lemma makes an analogous statement of injectivity for $\taubp$:

\begin{lemma}
\label{lem:injB}
 If $\ubp$ and $\ubp'$ are universal belief pairs such that $\ubp$ is universally consistent and $\taubp(\ubp) = \taubp(\ubp')$, then $\ubp = \ubp'$.
\end{lemma}
\aproof{\begin{proof}
 Follows trivially from Lemma \ref{lem:inj}.
\end{proof}}
 
The mapping is faithful to the knowledge order:

\begin{lemma}
\label{lem:order}
 If $\upws \leq_K \upws'$, then $\taupws(\upws) \leq_K \taupws(\upws')$. 
\end{lemma}
\aproof{\begin{proof}
 Let $J \in \taupws(\upws')$, i.e. for every $A \in \A$, $J_A \in \upws'_A$, i.e. $J_A \in \upws$. So $J \in \taupws(\upws)$.
\end{proof}}

The following lemma states that the mapping is faithful to $\leq_K$-least upper bounds and greatest lower bounds:

\begin{lemma}
\label{lem:klub}
 For a set $\mc{S}$ of DPWSs, $\taupws(\lub_{\leq_K}(\mc{S})) = lub_{\leq_K}(\taupws(\mc{S}))$ and $\taupws(\glb_{\leq_K}(\mc{S})) = \glb_{\leq_K}(\taupws(\mc{S}))$.
\end{lemma}
\aproof{\begin{proof}
 We prove the first equality; the second one can be proven similarly.
 
 First we show that $\taupws(\lub(\mc{S}))$ is an upper bound of $\taupws(\mc{S})$: Let $Q \in \taupws(\mc{S})$. Then there is a DPWS $\upws \in \mc{S}$ such that $Q = \taupws(\upws)$. Since $\upws \leq_K \lub \mc{S}$, Lemma \ref{lem:order} implies that $Q \leq_K \taupws(\lub(\mc{S}))$.
  
  Now we show that for each upper bound $Q'$ of $\taupws(\mc{S})$, $\taupws(\lub(\mc{S})) \leq_K Q'$: Suppose that for every $Q \in \taupws(\mc{S})$, $Q \leq_K Q'$, i.e. $Q' \subseteq Q$. We need to show that $\taupws(\lub(\mc{S})) \leq_K Q'$, i.e. that $Q' \subseteq \taupws(\lub(\mc{S}))$. So let $J \in Q'$. Let $\upws \in \mc{S}$. Then $\taupws(\upws) \in \taupws(\mc{S})$, so $Q' \subseteq \taupws(\upws)$. Hence $J \in \taupws(\upws)$, i.e. $J \mid _A \in \upws_A$ for each $A\in\A$. Given that $\upws$ was an arbitrary element of $\mc{S}$, we have that $J \mid _A \in \bigcap \{Q \mid \textnormal{for some } \upws\in\mc{S}, Q=\upws_A\}$. So $J \in \taupws((\bigcap \{Q \mid \textnormal{for some } \upws\in\mc{S}, Q=\upws_A\})_{A\in\A}) = \taupws(\lub(\mc{S}))$, as required.
\end{proof}}

The mapping is faithful to $\leq_p$-least upper bounds:

\begin{lemma}
\label{lem:plub}
 For a set $\mc{S}$ of distributed belief pairs, $\taubp(\lub_{\leq_p}(\mc{S})) = \lub_{\leq_p}(\taubp(\mc{S}))$. 
\end{lemma}
\aproof{\begin{proof}
This follows immediately from Lemma \ref{lem:klub} since 
\begin{align*}
 (\mc{P},\mc{S}) \leqp (\mc{P}',\mc{S}')
\end{align*}
if and only if
\[\mc{P}\leq_K \mc{P}'\text{ and } \mc{S}\geq_K \mc{S}'.\]
\end{proof}}

The following states that the stable revision of an element of the image of $\taupws$ is itself in the image of $\taupws$:

\begin{lemma}
\label{lem:image}
 For any DPWS $\upws$, there is a DPWS $\upws'$ such that $\taupws(\upws') = \upwsrevst{\tautheo(\T)}(\taupws(\upws))$.
\end{lemma}
\aproof{\begin{proof}
We know that $\upwsrevst{\tautheo(\T)}(\taupws(\upws)) = \lfp \upwsrevL{\tautheo(\T)}(\cdot,\taupws(\upws))$. By induction, it is enough to show that for each ordinal number $\alpha$, there is a DPWS $\upws'$ such that $\taupws(\upws') = \upwsrevL{\tautheo(\T)}(\cdot,\taupws(\upws))^\alpha(\bot)$. 

For $\alpha = 0$, let $\upws':= (\bot)_{A\in\A}$. Then $\taupws(\upws') = (\bot) = \upwsrevL{\tautheo(\T)}(\cdot,\taupws(\upws))^0(\bot)$.

Suppose the result holds for $\alpha$, i.e.\ there is a DPWS $\upws'$ such that $\taupws(\upws') = \upwsrevL{\tautheo(\T)}(\cdot,\taupws(\upws))^\alpha(\bot)$. By Lemma \ref{lem:ubprev}, $\taupws(\upwsrevL{\T}(\upws',\upws)) = \upwsrevL{\tautheo(\T)}(\taupws(\upws'),\taupws(\upws)) = \upwsrevL{\tautheo(\T)}(\cdot,\taupws(\upws))^{\alpha+1}(\bot)$.

Let $\lambda$ be a limit ordinal such that the result holds for every $\alpha < \lambda$. Define $\mc{S} := \{\upws' \mid \taupws(\upws') = \upwsrevL{\tautheo(\T)}(\cdot,\taupws(\upws))^\alpha(\bot) \textnormal{ for some } \alpha < \lambda\}$. By Lemma \ref{lem:klub}, $\taupws(\lub(\mc{S})) = \lub(\taupws(\mc{S})) = \upwsrevL{\tautheo(\T)}(\cdot,\taupws(\upws))^\lambda(\bot)$. 
\end{proof}}

The following lemma states that the mapping maps the stable dAEL knowledge revision operator $\upwsrevst{\T}$ to the stable AEL knowledge revision operator $S_{\mc{D}_T^*}$:

\begin{lemma}
\label{lem:ubprevst}
  For every universally consistent distributed theory $\Tt$ and every DPWS $\upws$, 
  $$\upwsrevst{\tautheo(\T)}(\taupws(\upws)) = \taupws(\upwsrevst{\T}(\upws)).$$
\end{lemma}
\aproof{\begin{proof}
Let $Q$ denote $\upwsrevst{\tautheo(\T)}(\taupws(\upws))$. 

First suppose $Q = \emptyset$. We need to show that $\taupws(\upwsrevst{\T}(\upws)) = \emptyset$, i.e. that $\upwsrevst{\T}(\upws) = \lfp(\upwsrevL{\T}(\cdot,\upws))$ is not universally consistent. For this it is enough to show that every universally consistent DPWS is not a fixpoint of $\upwsrevL{\T}(\cdot,\upws)$. So suppose $\upws'$ is universally consistent. Then $\taupws(\upws') \neq \emptyset$, i.e. $\taupws(\upws') <_K Q$. Since $Q$ is the least fixpoint of $\upwsrevL{\tautheo(\T)}(\cdot,\taupws(\upws))$, $\upwsrevL{\tautheo(\T)}(\taupws(\upws'),\taupws(\upws)) \neq \taupws(\upws')$. So by Lemma \ref{lem:ubprev}, $\taupws(\upwsrevL{\T}(\upws',\upws)) \neq \taupws(\upws')$, i.e.\ $\upwsrevL{\T}(\upws',\upws) \neq \upws'$, as required.

Now suppose $Q \neq \emptyset$. By Lemma \ref{lem:image}, there is a DPWS $\upws'$ such that $\taupws(\upws') = Q$. Note that by Lemma \ref{lem:cons}, $\upws$ is universally consistent. $Q = \taupws(\upws')$ is a fixpoint of $\upwsrevL{\tautheo(\T)}(\cdot,\taupws(\upws))$, i.e.\ $\upwsrevL{\tautheo(\T)}(\taupws(\upws'),\taupws(\upws)) \neq \taupws(\upws')$. By Lemma \ref{lem:ubprev}, $\taupws(\upwsrevL{\T}(\upws',\upws)) = \taupws(\upws')$. By Lemma \ref{lem:inj}, $\upwsrevL{\T}(\upws',\upws) = \upws'$, i.e.\ $\upws'$ is a fixpoint of $\upwsrevL{\T}(\cdot,\upws)$. Let $\upws''$ denote the least fixpoint of $\upwsrevL{\T}(\cdot,\upws)$. Then $\upws'' \leq_K \upws'$, so by Lemma \ref{lem:order}, $\taupws(\upws'') \leq_K \taupws(\upws')$. Additionally, $\upwsrevL{\T}(\upws'',\upws) = \upws''$, so $\taupws(\upwsrevL{\T}(\upws'',\upws)) = \taupws(\upws'')$, so by Lemma \ref{lem:ubprev}, $\taupws(\upws'')$ is a fixpoint of $\upwsrevL{\tautheo(\T)}(\cdot,\taupws(\upws))$. Since $Q=\taupws(\upws')$ is the least fixpoint of $\upwsrevL{\tautheo(\T)}(\cdot,\taupws(\upws))$, $\taupws(\upws') \leq_K \taupws(\upws'')$. Combining the two inequalities, we get $\taupws(\upws') = \taupws(\upws'')$, so by Lemma \ref{lem:inj}, $\upws' = \upws''$. So $\upws' = \lfp(\upwsrevL{\T}(\cdot,\upws)) = \upwsrevst{\T}(\upws)$, i.e.\ $Q = \taupws(\upws') = \taupws(\upwsrevst{\T}(\upws))$, as required.
\end{proof}}

We are now ready to present the proofs of Theorems \ref{thm:mapping:perm} and \ref{thm:mapping}.

\begin{proof}[Proof of Theorem \ref{thm:mapping}]\ 

\noindent \underline{Case 1: $\sigma = \textsf{Sup}$:} Suppose $\upws$ is a universally consistent DPWS. $\upws$ is a $\textsf{Sup}$-model of $\T$\\
iff $\upwsrev(\upws)=\upws$\\
iff $\taupws(\upwsrev(\upws)) = \taupws(\upws)$ by Lemma \ref{lem:inj}\\
iff $\upwsrev{\tautheo(\T)}(\taupws(\upws))$ by Lemma \ref{lem:upwsrev}\\
iff $\taupws(\upws)$ is a \textsf{Sup}-model of $\tautheo(\T)$.

\noindent \underline{Case 2: $\sigma = \textsf{PSt}$:} Similar to Case 1, but using Lemma \ref{lem:ubprevst} instead of Lemma \ref{lem:upwsrev}.

\noindent \underline{Case 3: $\sigma = \textsf{St}$:} follows from Case 2 since \textsf{St}-models are two-valued \textsf{PSt}-models.
%
\end{proof}

\begin{proof}[Proof of Theorem \ref{thm:mapping:perm}]\ 

\noindent \underline{Case 1: $\sigma \in \{\textsf{Sup},\textsf{PSt},\textsf{St}\}$:} follows by combining Theorems \ref{thm:mapping} and \ref{thm:link}.

 \noindent \underline{Case 2: $\sigma = \textsf{KK}$:} The \textsf{KK}-model of $\T$ is the $\leq_p$-least fixpoint of $\ubprev{\T}$ and the \textsf{KK}-model of $\tau_\T(\T)$ is the  $\leq_p$-least fixpoint of $\ubprev{\tau_\T(\T)}$. So by Lemma \ref{lem:injB}, it is enough to show that for each ordinal number $\alpha>0$,
 ${\ubprev{\T}}^\alpha((\bot,\top)_{A\in\A})$ is universally consistent and
 $\taubp({\ubprev{\T}}^\alpha((\bot,\top)_{A\in\A}) = {\ubprev{\tautheo(\T)}}^\alpha((\bot,\top))$. We prove this by transfinite induction.


For $\alpha = 1$, this is follows from Lemma \ref{lem:perm}.

Suppose it is true for $\alpha$. Then 
\begin{align*}
 \taubp(&{\ubprev{\T}}^{\alpha+1}((\bot,\top)_{A\in\A}) = \taubp(\ubprev{\T}({\ubprev{\T}}^{\alpha}((\bot,\top)_{A\in\A})) \\
 &= \ubprev{\tautheo(\T)}(\taubp({\ubprev{\T}}^{\alpha}((\bot,\top)_{A\in\A})) \textnormal{ by Lemma \ref{lem:ubprev}}\\
 &= {\ubprev{\tautheo(\T)}}^{\alpha+1}((\bot,\top)) \textnormal{ by assumption about } \alpha.
\end{align*}

Now suppose it is true for all $\alpha < \lambda$. Then
\begin{align*}
 \tau&_B({\ubprev{\T}}^\lambda((\bot,\top)_{A\in\A})) \\
 &= \taubp(\lub(\{{\ubprev{\T}}^\alpha((\bot,\top)_{A\in\A}) \mid \alpha<\lambda\}))\\
 &= \lub(\taubp[\{{\ubprev{\T}}^\alpha((\bot,\top)_{A\in\A}) \mid \alpha<\lambda\}]) \textnormal{ by Lemma \ref{lem:plub}}\\
 &= \lub(\{{\ubprev{\tautheo(\T)}}^\alpha((\bot,\top)) \mid \alpha<\lambda\})\\
 &= {\ubprev{\tautheo(\T)}}^\lambda((\bot,\top))
\end{align*}

\noindent \underline{Case 3: $\sigma = \textsf{WF}$:} Similar to Case 2, but using Lemma \ref{lem:ubprevst} instead of Lemma \ref{lem:ubprev}.
\end{proof}

%% file: dAEL_journal.bbl
\begin{thebibliography}{67}
\providecommand{\natexlab}[1]{#1}
\providecommand{\url}[1]{\texttt{#1}}
\expandafter\ifx\csname urlstyle\endcsname\relax
  \providecommand{\doi}[1]{doi: #1}\else
  \providecommand{\doi}{doi: \begingroup \urlstyle{rm}\Url}\fi

\bibitem[Abadi(2003)]{Abadi03}
Mart{\'i}n Abadi.
\newblock {Logic in Access Control}.
\newblock In \emph{{Proceedings of the Eighteenth Annual IEEE Symposium on
  Logic in Computer Science}}, pages 228--233, 2003.

\bibitem[Abadi(2008)]{Abadi08}
Mart{\'i}n Abadi.
\newblock {Variations in Access Control Logic}.
\newblock In \emph{{9th International Conference on Deontic Logic in Computer
  Science}}, pages 96--109, 2008.

\bibitem[Ambrossio and Cramer(2019)]{Ambrossio19}
Diego~Agust{\'\i}n Ambrossio and Marcos Cramer.
\newblock {A Query-Driven Decision Procedure for Distributed Autoepistemic
  Logic with Inductive Definitions}.
\newblock \emph{arXiv e-prints}, art. arXiv:1910.04010, Oct 2019.

\bibitem[Antic(2020)]{ai/Antic20}
Christian Antic.
\newblock Fixed point semantics for stream reasoning.
\newblock \emph{Artif. Intell.}, 288:\penalty0 103370, 2020.
\newblock \doi{10.1016/j.artint.2020.103370}.
\newblock URL \url{https://doi.org/10.1016/j.artint.2020.103370}.

\bibitem[Antić(2023)]{Antic}
Christian Antić.
\newblock Neural logic programs and neural nets, 03 2023.
\newblock preprint.

\bibitem[Appel and Felten(1999)]{ccs/AppelF99}
Andrew~W. Appel and Edward~W. Felten.
\newblock Proof-carrying authentication.
\newblock In Juzar Motiwalla and Gene Tsudik, editors, \emph{{CCS} '99,
  Proceedings of the 6th {ACM} Conference on Computer and Communications
  Security, Singapore, November 1-4, 1999}, pages 52--62. {ACM}, 1999.
\newblock ISBN 1-58113-148-8.
\newblock \doi{10.1145/319709.319718}.
\newblock URL \url{https://doi.org/10.1145/319709.319718}.

\bibitem[Barker(2009)]{sacmat/Barker2009}
Steve Barker.
\newblock The next 700 access control models or a unifying meta-model?
\newblock In \emph{Proceedings of the 14th ACM symposium on Access control
  models and technologies}, SACMAT '09, pages 187--196. ACM, 2009.
\newblock \doi{http://dx.doi.org/10.1145/1542207.1542238}.
\newblock URL \url{http://dx.doi.org/10.1145/1542207.1542238}.

\bibitem[Barker and Genovese(2011)]{Secommunity}
Steve Barker and Valerio Genovese.
\newblock {Secommunity: A Framework for Distributed Access Control}.
\newblock In James~P. Delgrande and Wolfgang Faber, editors, \emph{{LPNMR}},
  volume 6645 of \emph{Lecture Notes in Computer Science}, pages 297--303.
  Springer, 2011.
\newblock ISBN 978-3-642-20894-2.
\newblock \doi{10.1007/978-3-642-20895-9}.
\newblock URL \url{http://dx.doi.org/10.1007/978-3-642-20895-9}.

\bibitem[Barker et~al.(2014)Barker, Boella, Gabbay, and
  Genovese]{logcom/BarkerBGG14}
Steve Barker, Guido Boella, Dov~M. Gabbay, and Valerio Genovese.
\newblock Reasoning about delegation and revocation schemes in answer set
  programming.
\newblock \emph{J. Log. Comput.}, 24\penalty0 (1):\penalty0 89--116, 2014.
\newblock \doi{10.1093/logcom/exs014}.
\newblock URL \url{http://dx.doi.org/10.1093/logcom/exs014}.

\bibitem[Belle and Lakemeyer(2015)]{ijcai/BelleL15}
Vaishak Belle and Gerhard Lakemeyer.
\newblock Only knowing meets common knowledge.
\newblock In  \citet{ijcai/2015}, pages 2755--2761.
\newblock ISBN 978-1-57735-738-4.
\newblock URL \url{http://ijcai.org/papers15/Abstracts/IJCAI15-390.html}.

\bibitem[Bi et~al.(2014)Bi, You, and Feng]{RR/BiJF14}
Yi~Bi, Jia{-}Huai You, and Zhiyong Feng.
\newblock A generalization of approximation fixpoint theory and application.
\newblock In Roman Kontchakov and Marie{-}Laure Mugnier, editors, \emph{Web
  Reasoning and Rule Systems - 8th International Conference, {RR} 2014, Athens,
  Greece, September 15-17, 2014. Proceedings}, volume 8741 of \emph{Lecture
  Notes in Computer Science}, pages 45--59. Springer, 2014.
\newblock ISBN 978-3-319-11112-4.
\newblock \doi{10.1007/978-3-319-11113-1_4}.
\newblock URL \url{http://dx.doi.org/10.1007/978-3-319-11113-1_4}.

\bibitem[Bogaerts(2015)]{phd/Bogaerts15}
Bart Bogaerts.
\newblock \emph{Groundedness in logics with a fixpoint semantics}.
\newblock PhD thesis, Department of Computer Science, KU Leuven, June 2015.
\newblock URL \url{https://lirias.kuleuven.be/handle/123456789/496543}.
\newblock Denecker, Marc (supervisor), Vennekens, Joost and Van den Bussche,
  Jan (cosupervisors).

\bibitem[Bogaerts(2019)]{aaai/Bogaerts19}
Bart Bogaerts.
\newblock Weighted abstract dialectical frameworks through the lens of
  approximation fixpoint theory.
\newblock In \emph{The Thirty-Third {AAAI} Conference on Artificial
  Intelligence, {AAAI} 2019, The Thirty-First Innovative Applications of
  Artificial Intelligence Conference, {IAAI} 2019, The Ninth {AAAI} Symposium
  on Educational Advances in Artificial Intelligence, {EAAI} 2019, Honolulu,
  Hawaii, USA, January 27 - February 1, 2019}, pages 2686--2693. {AAAI} Press,
  2019.
\newblock ISBN 978-1-57735-809-1.
\newblock \doi{10.1609/aaai.v33i01.33012686}.
\newblock URL \url{https://doi.org/10.1609/aaai.v33i01.33012686}.

\bibitem[Bogaerts and Cruz-Filipe(2018)]{ai/BogaertsC18}
Bart Bogaerts and {Lu\'\i s} Cruz-Filipe.
\newblock Fixpoint semantics for active integrity constraints.
\newblock \emph{Artif. Intell.}, 255:\penalty0 43--70, 2018.
\newblock \doi{10.1016/j.artint.2017.11.003}.
\newblock URL \url{https://doi.org/10.1016/j.artint.2017.11.003}.

\bibitem[Bogaerts and Jakubowski(2021)]{iclp/BogaertsJ21}
Bart Bogaerts and Maxime Jakubowski.
\newblock Fixpoint semantics for recursive {SHACL}.
\newblock In Andrea Formisano, Yanhong~Annie Liu, Bart Bogaerts, Alex Brik,
  Ver{\'{o}}nica Dahl, Carmine Dodaro, Paul Fodor, Gian~Luca Pozzato, Joost
  Vennekens, and Neng{-}Fa Zhou, editors, \emph{Proceedings 37th International
  Conference on Logic Programming (Technical Communications), {ICLP} Technical
  Communications 2021, Porto (virtual event), 20-27th September 2021}, volume
  345 of \emph{{EPTCS}}, pages 41--47, 2021.
\newblock \doi{10.4204/EPTCS.345.14}.
\newblock URL \url{https://doi.org/10.4204/EPTCS.345.14}.

\bibitem[Bogaerts et~al.(2014)Bogaerts, Vennekens, Denecker, and {Van den
  Bussche}]{iclp/BogaertsVDV14}
Bart Bogaerts, Joost Vennekens, Marc Denecker, and Jan {Van den Bussche}.
\newblock {FO(C)}: A knowledge representation language of causality.
\newblock \emph{TPLP}, 14\penalty0 (4--5-Online-Supplement):\penalty0 60--69,
  2014.
\newblock URL \url{https://lirias.kuleuven.be/handle/123456789/459436}.

\bibitem[Bogaerts et~al.(2015{\natexlab{a}})Bogaerts, Vennekens, and
  Denecker]{ai/BogaertsVD15}
Bart Bogaerts, Joost Vennekens, and Marc Denecker.
\newblock Grounded fixpoints and their applications in knowledge
  representation.
\newblock \emph{Artif. Intell.}, 224:\penalty0 51--71, 2015{\natexlab{a}}.
\newblock ISSN 0004-3702.
\newblock \doi{10.1016/j.artint.2015.03.006}.
\newblock URL \url{http://dx.doi.org/10.1016/j.artint.2015.03.006}.

\bibitem[Bogaerts et~al.(2015{\natexlab{b}})Bogaerts, Vennekens, and
  Denecker]{ijcai/BogaertsVD15}
Bart Bogaerts, Joost Vennekens, and Marc Denecker.
\newblock Partial grounded fixpoints.
\newblock In  \citet{ijcai/2015}, pages 2784--2790.
\newblock ISBN 978-1-57735-738-4.
\newblock URL \url{http://ijcai.org/papers15/Abstracts/IJCAI15-394.html}.

\bibitem[Bogaerts et~al.(2016)Bogaerts, Vennekens, and
  Denecker]{tocl/BogaertsVD16}
Bart Bogaerts, Joost Vennekens, and Marc Denecker.
\newblock On well-founded set-inductions and locally monotone operators.
\newblock \emph{ACM Trans. Comput. Logic}, 17\penalty0 (4):\penalty0
  27:1--27:32, September 2016.
\newblock ISSN 1529-3785.
\newblock \doi{10.1145/2963096}.
\newblock URL \url{http://doi.acm.org/10.1145/2963096}.

\bibitem[Bogaerts et~al.(2018)Bogaerts, Vennekens, and
  Denecker]{ai/BogaertsVD18}
Bart Bogaerts, Joost Vennekens, and Marc Denecker.
\newblock Safe inductions and their applications in knowledge representation.
\newblock \emph{Artificial Intelligence}, 259:\penalty0 167 -- 185, 2018.
\newblock ISSN 0004-3702.
\newblock \doi{10.1016/j.artint.2018.03.008}.
\newblock URL
  \url{http://www.sciencedirect.com/science/article/pii/S000437021830122X}.

\bibitem[Chander et~al.(2004)Chander, Dean, and Mitchell]{Chander04}
Ajay Chander, Drew Dean, and John~C. Mitchell.
\newblock {Reconstructing trust management}.
\newblock \emph{Journal of Computer Security}, 2004.

\bibitem[Charalambidis et~al.(2018)Charalambidis, Rondogiannis, and
  Symeonidou]{tplp/CharalambidisRS18}
Angelos Charalambidis, Panos Rondogiannis, and Ioanna Symeonidou.
\newblock Approximation fixpoint theory and the well-founded semantics of
  higher-order logic programs.
\newblock \emph{{TPLP}}, 18\penalty0 (3-4):\penalty0 421--437, 2018.
\newblock \doi{10.1017/S1471068418000108}.
\newblock URL \url{https://doi.org/10.1017/S1471068418000108}.

\bibitem[Cramer and Casini(2017)]{Cramer17}
Marcos Cramer and Giovanni Casini.
\newblock {Postulates for Revocation Schemes}.
\newblock In Matteo Maffei and Mark Ryan, editors, \emph{{Principles of
  Security and Trust}}, pages 232--252, Berlin, Heidelberg, 2017. Springer
  Berlin Heidelberg.
\newblock ISBN 978-3-662-54455-6.

\bibitem[Cramer et~al.(2015)Cramer, Ambrossio, and {Van Hertum}]{Cramer15}
Marcos Cramer, Diego~Agustin Ambrossio, and Pieter {Van Hertum}.
\newblock {A Logic of Trust for Reasoning about Delegation and Revocation}.
\newblock In \emph{{Proceedings of the 20th ACM Symposium on Access Control
  Models and Technologies }}, pages 173--184. ACM, 2015.

\bibitem[Denecker(1998)]{Denecker98}
Marc Denecker.
\newblock The well-founded semantics is the principle of inductive definition.
\newblock In J{\"u}rgen Dix, Luis~Fari{\~n}as del Cerro, and Ulrich Furbach,
  editors, \emph{JELIA}, volume 1489 of \emph{LNCS}, pages 1--16. Springer,
  1998.
\newblock ISBN 3-540-65141-1.

\bibitem[Denecker and Vennekens(2007)]{lpnmr/DeneckerV07}
Marc Denecker and Joost Vennekens.
\newblock Well-founded semantics and the algebraic theory of non-monotone
  inductive definitions.
\newblock In Chitta Baral, Gerhard Brewka, and John~S. Schlipf, editors,
  \emph{{LPNMR}}, volume 4483 of \emph{Lecture Notes in Computer Science},
  pages 84--96. Springer, 2007.
\newblock ISBN 978-3-540-72199-4.
\newblock \doi{10.1007/978-3-540-72200-7\_9}.
\newblock URL \url{http://dx.doi.org/10.1007/978-3-540-72200-7_9}.

\bibitem[Denecker et~al.(1998)Denecker, Marek, and
  Truszczy{\'n}ski]{aaai/DMT98}
Marc Denecker, Victor Marek, and Miros{\l}aw Truszczy{\'n}ski.
\newblock Fixpoint 3-valued semantics for autoepistemic logic.
\newblock In Jack Mostow and Chuck Rich, editors, \emph{AAAI'98}, pages
  840--845, Madison, Wisconsin, July 26-30 1998. MIT Press.
\newblock ISBN 0-262-51098-7.
\newblock URL \url{http://www.aaai.org/Papers/AAAI/1998/AAAI98-119.pdf}.

\bibitem[Denecker et~al.(2000)Denecker, Marek, and
  Truszczy{\'n}ski]{DeneckerMT00}
Marc Denecker, Victor Marek, and Miros{\l}aw Truszczy{\'n}ski.
\newblock Approximations, stable operators, well-founded fixpoints and
  applications in nonmonotonic reasoning.
\newblock In Jack Minker, editor, \emph{Logic-Based Artificial Intelligence},
  volume 597 of \emph{The Springer International Series in Engineering and
  Computer Science}, pages 127--144. Springer US, 2000.
\newblock ISBN 978-1-4613-5618-9.
\newblock \doi{10.1007/978-1-4615-1567-8_6}.
\newblock URL \url{http://dx.doi.org/10.1007/978-1-4615-1567-8_6}.

\bibitem[Denecker et~al.(2003)Denecker, Marek, and
  Truszczy{\'n}ski]{DeneckerMT03}
Marc Denecker, Victor Marek, and Miros{\l}aw Truszczy{\'n}ski.
\newblock Uniform semantic treatment of default and autoepistemic logics.
\newblock \emph{Artif. Intell.}, 143\penalty0 (1):\penalty0 79--122, 2003.
\newblock URL \url{http://dx.doi.org/10.1016/S0004-3702(02)00293-X}.

\bibitem[Denecker et~al.(2004)Denecker, Marek, and
  Truszczy{\'n}ski]{DeneckerMT04}
Marc Denecker, Victor Marek, and Miros{\l}aw Truszczy{\'n}ski.
\newblock Ultimate approximation and its application in nonmonotonic knowledge
  representation systems.
\newblock \emph{Information and Computation}, 192\penalty0 (1):\penalty0
  84--121, July 2004.
\newblock \doi{10.1016/j.ic.2004.02.004}.
\newblock URL \url{https://lirias.kuleuven.be/handle/123456789/124562}.

\bibitem[Denecker et~al.(2011)Denecker, Marek, and
  Truszczy{\'n}ski]{nonmon30/DeneckerMT11}
Marc Denecker, Victor Marek, and Miros{\l}aw Truszczy{\'n}ski.
\newblock Reiter's default logic is a logic of autoepistemic reasoning and a
  good one, too.
\newblock In Gerd Brewka, Victor Marek, and Miros{\l}aw Truszczy{\'n}ski,
  editors, \emph{Nonmonotonic Reasoning -- Essays Celebrating Its 30th
  Anniversary}, pages 111--144. College Publications, 2011.
\newblock URL \url{http://arxiv.org/abs/1108.3278}.

\bibitem[Fitting(2002)]{tcs/Fitting02}
Melvin Fitting.
\newblock Fixpoint semantics for logic programming --- {A} survey.
\newblock \emph{Theoretical Computer Science}, 278\penalty0 (1-2):\penalty0
  25--51, 2002.
\newblock \doi{10.1016/S0304-3975(00)00330-3}.
\newblock URL \url{http://dx.doi.org/10.1016/S0304-3975(00)00330-3}.

\bibitem[Garg(2009)]{Garg09}
Deepak Garg.
\newblock \emph{{Proof Theory for Authorization Logic and Its Application to a
  Practical File System}}.
\newblock PhD thesis, 2009.

\bibitem[Garg and Abadi(2008)]{Garg2008}
Deepak Garg and Mart{\'i}n Abadi.
\newblock {A Modal Deconstruction of Access Control Logics}.
\newblock In Roberto Amadio, editor, \emph{{Foundations of Software Science and
  Computational Structures: 11th International Conference, FOSSACS 2008, Held
  as Part of the Joint European Conferences on Theory and Practice of Software,
  ETAPS 2008, Budapest, Hungary, March 29 - April 6, 2008. Proceedings}}, pages
  216--230, Berlin, Heidelberg, 2008. Springer Berlin Heidelberg.
\newblock ISBN 978-3-540-78499-9.
\newblock \doi{10.1007/978-3-540-78499-9_16}.
\newblock URL \url{http://dx.doi.org/10.1007/978-3-540-78499-9_16}.

\bibitem[Garg and Pfenning(2012)]{Garg12}
Deepak Garg and Frank Pfenning.
\newblock {Stateful Authorization Logic -- Proof Theory and a Case Study}.
\newblock \emph{Journal of Computer Security}, 20\penalty0 (4):\penalty0
  353--391, 2012.

\bibitem[Gelfond and Lifschitz(1988)]{iclp/GelfondL88}
Michael Gelfond and Vladimir Lifschitz.
\newblock The stable model semantics for logic programming.
\newblock In Robert~A. Kowalski and Kenneth~A. Bowen, editors, \emph{ICLP/SLP},
  pages 1070--1080. MIT Press, 1988.
\newblock ISBN 0-262-61056-6.
\newblock URL
  \url{http://citeseer.ist.psu.edu/viewdoc/summary?doi=10.1.1.24.6050}.

\bibitem[Genovese(2012)]{Genovese12}
Valerio Genovese.
\newblock \emph{{Modalities in Access Control: Logics, Proof-theory and
  Application}}.
\newblock PhD thesis, 2012.

\bibitem[Gurevich and Neeman(2008)]{Gurevich07}
Yuri Gurevich and Itay Neeman.
\newblock {DKAL: Distributed-knowledge authorization language}.
\newblock In \emph{{Proceedings of the 2008 21st IEEE Computer Security
  Foundations Symposium}}, pages 149--162. IEEE, 2008.

\bibitem[Hagstr{\"o}m et~al.(2001)Hagstr{\"o}m, Jajodia, Parisi-Presicce, and
  Wijesekera]{Hagstrom01}
{\AA}sa Hagstr{\"o}m, Sushil Jajodia, Francesco Parisi-Presicce, and Duminda
  Wijesekera.
\newblock {Revocations -- A Classification}.
\newblock In \emph{{Proceedings of the 14th IEEE Workshop on Computer Security
  Foundations}}, pages 44--58. IEEE, 2001.

\bibitem[Halpern and Moses(1985)]{nato/HalpernM85}
Joseph~Y. Halpern and Yoram Moses.
\newblock Towards a theory of knowledge and ignorance: Preliminary report.
\newblock In Krzysztof~R. Apt, editor, \emph{Logics and Models of Concurrent
  Systems}, volume~13 of \emph{NATO ASI Series}, pages 459--476. Springer
  Berlin Heidelberg, 1985.
\newblock ISBN 978-3-642-82455-5.
\newblock \doi{10.1007/978-3-642-82453-1_16}.
\newblock URL \url{http://dx.doi.org/10.1007/978-3-642-82453-1_16}.

\bibitem[Hirsch and Clarkson(2013)]{Clarkson13}
Andrew~K. Hirsch and Michael~R. Clarkson.
\newblock {Belief Semantics of Authorization Logic}.
\newblock \emph{CoRR}, abs/1302.2123, 2013.
\newblock URL \url{http://arxiv.org/abs/1302.2123;
  https://dblp.org/rec/bib/journals/corr/abs-1302-2123}.

\bibitem[Hughes and Cresswell(1996)]{hughescresswell.niml}
G.~E. Hughes and M.~J. Cresswell.
\newblock \emph{A New Introduction To Modal Logic}.
\newblock Routledge, 1996.

\bibitem[Karp(1972)]{Kar72}
R.~Karp.
\newblock Reducibility among combinatorial problems.
\newblock In R.~Miller and J.~Thatcher, editors, \emph{Complexity of Computer
  Computations}, pages 85--103. Plenum Press, 1972.

\bibitem[Kleene(1938)]{Kleene38}
Stephen~Cole Kleene.
\newblock On notation for ordinal numbers.
\newblock \emph{The Journal of Symbolic Logic}, 3\penalty0 (4):\penalty0
  150--155, 1938.
\newblock ISSN 00224812.
\newblock URL \url{http://www.jstor.org/stable/2267778}.

\bibitem[Konolige(1988)]{Konolige88}
Kurt Konolige.
\newblock On the relation between default and autoepistemic logic.
\newblock \emph{Artif. Intell.}, 35\penalty0 (3):\penalty0 343--382, 1988.
\newblock \doi{10.1016/0004-3702(88)90021-5}.
\newblock URL \url{http://dx.doi.org/10.1016/0004-3702(88)90021-5}.

\bibitem[Levesque(1990)]{ai/Levesque90}
Hector~J. Levesque.
\newblock All {I} know: A study in autoepistemic logic.
\newblock \emph{Artif. Intell.}, 42\penalty0 (2-3):\penalty0 263--309, 1990.
\newblock \doi{10.1016/0004-3702(90)90056-6}.
\newblock URL \url{http://dx.doi.org/10.1016/0004-3702(90)90056-6}.

\bibitem[Lewis and Langford(1932)]{Lewis32}
C.I. Lewis and C.H. Langford.
\newblock \emph{Symbolic logic}.
\newblock Century philosophy series. The Century co., 1932.
\newblock URL \url{https://books.google.be/books?id=jLotAAAAMAAJ}.

\bibitem[Li et~al.(2003)Li, Grosof, and Feigenbaum]{Li}
Ninghui Li, Benjamin~N. Grosof, and Joan Feigenbaum.
\newblock {Delegation Logic: A Logic-based Approach to Distributed
  Authorization}.
\newblock \emph{ACM Transaction on Information and System Security}, 2003.

\bibitem[Moore(1985{\natexlab{a}})]{Moore85}
Robert~C. Moore.
\newblock {A Formal Theory of Knowledge and Action}.
\newblock In J.~R. Hobbs and R.~C. Moore, editors, \emph{Formal Theories of the
  Commonsense World}, pages 319--358. Springer-Verlag, 1985{\natexlab{a}}.

\bibitem[Moore(1985{\natexlab{b}})]{mo85}
Robert~C. Moore.
\newblock Semantical considerations on nonmonotonic logic.
\newblock \emph{Artif. Intell.}, 25\penalty0 (1):\penalty0 75--94,
  1985{\natexlab{b}}.
\newblock \doi{10.1016/0004-3702(85)90042-6}.
\newblock URL \url{http://dx.doi.org/10.1016/0004-3702(85)90042-6}.

\bibitem[Morgenstern(1990)]{aaai/Morgenstern90}
Leora Morgenstern.
\newblock A formal theory of multiple agent nonmonotonic reasoning.
\newblock In T.~Dieterich and W.~Swartout, editors, \emph{Proceedings of the
  8th National Conference on Artificial Intelligence. Boston, Massachusetts,
  July 29 - August 3, 1990, 2 Volumes.}, pages 538--544. AAAI/MIT Press, 1990.
\newblock ISBN 978-0-262-51057-8.
\newblock URL \url{http://www.aaai.org/Library/AAAI/1990/aaai90-081.php}.

\bibitem[Niemel{\"{a}}(1991)]{ijcai/Niemela91}
Ilkka Niemel{\"{a}}.
\newblock Constructive tightly grounded autoepistemic reasoning.
\newblock In John Mylopoulos and Raymond Reiter, editors, \emph{Proceedings of
  the 12th International Joint Conference on Artificial Intelligence. Sydney,
  Australia, August 24-30, 1991}, pages 399--405. Morgan Kaufmann, 1991.
\newblock ISBN 1-55860-160-0.

\bibitem[Permpoontanalarp and Jiang(1995)]{wocfai/PermpoontanalarpJ95}
Yongyuth Permpoontanalarp and John~Yuejun Jiang.
\newblock On multi-agent autoepistemic reasoning.
\newblock In \emph{{WOCFAI}}, pages 307--318, 1995.

\bibitem[Strass(2013)]{journals/ai/Strass13}
Hannes Strass.
\newblock Approximating operators and semantics for abstract dialectical
  frameworks.
\newblock \emph{Artif. Intell.}, 205:\penalty0 39--70, 2013.
\newblock \doi{10.1016/j.artint.2013.09.004}.
\newblock URL \url{http://dx.doi.org/10.1016/j.artint.2013.09.004}.

\bibitem[Tamassia et~al.(2004)Tamassia, Yao, and Winsborough]{Tamassia}
Roberto Tamassia, Danfeng Yao, and William~H. Winsborough.
\newblock {Role-Based Cascaded Delegation}.
\newblock In \emph{{Proceedings of the 9th ACM symposium on Access control
  models and technologies }}, 2004.

\bibitem[Toyama et~al.(2002)Toyama, Kojima, and Inagaki]{clima/ToyamaKI02}
Katsuhiko Toyama, Takahiro Kojima, and Yasuyoshi Inagaki.
\newblock Translating multi-agent autoepistemic logic into logic program.
\newblock In J{\"{u}}rgen Dix, Jo{\~{a}}o~Alexandre Leite, and Ken Satoh,
  editors, \emph{Computational Logic in Multi-Agent Systems: 3rd International
  Workshop, CLIMA'02, Copenhagen, Denmark, August 1, 2002, Pre-Proceedings},
  volume~93 of \emph{Datalogiske Skrifter}, pages 49--62. Roskilde University,
  2002.

\bibitem[Truszczy{\'n}ski(2006)]{amai/Truszczynski06}
Miros{\l}aw Truszczy{\'n}ski.
\newblock Strong and uniform equivalence of nonmonotonic theories - an
  algebraic approach.
\newblock \emph{Ann. Math. Artif. Intell.}, 48\penalty0 (3-4):\penalty0
  245--265, 2006.
\newblock \doi{10.1007/s10472-007-9049-2}.
\newblock URL \url{http://dx.doi.org/10.1007/s10472-007-9049-2}.

\bibitem[{van Emden} and Kowalski(1976)]{jacm/EmdenK76}
Maarten~H. {van Emden} and Robert~A. Kowalski.
\newblock The semantics of predicate logic as a programming language.
\newblock \emph{J. ACM}, 23\penalty0 (4):\penalty0 733--742, 1976.
\newblock \doi{10.1145/321978.321991}.
\newblock URL \url{http://dx.doi.org/10.1145/321978.321991}.

\bibitem[{Van Gelder} et~al.(1991){Van Gelder}, Ross, and Schlipf]{GelderRS91}
Allen {Van Gelder}, Kenneth~A. Ross, and John~S. Schlipf.
\newblock The well-founded semantics for general logic programs.
\newblock \emph{J. ACM}, 38\penalty0 (3):\penalty0 620--650, 1991.
\newblock \doi{10.1145/116825.116838}.
\newblock URL \url{http://dx.doi.org/10.1145/116825.116838}.

\bibitem[{Van Hertum} et~al.(2016){Van Hertum}, Cramer, Bogaerts, and
  Denecker]{ijcai/HertumCBD16}
Pieter {Van Hertum}, Marcos Cramer, Bart Bogaerts, and Marc Denecker.
\newblock Distributed autoepistemic logic and its application to access
  control.
\newblock In Subbarao Kambhampati, editor, \emph{Proceedings of the
  Twenty-Fifth International Joint Conference on Artificial Intelligence,
  {IJCAI} 2016, New York, NY, USA, 9-15 July 2016}, pages 1286--1292.
  {IJCAI/AAAI} Press, 2016.
\newblock ISBN 978-1-57735-770-4.
\newblock URL \url{http://www.ijcai.org/Abstract/16/186}.

\bibitem[Vennekens et~al.(2006)Vennekens, Gilis, and
  Denecker]{tocl/VennekensGD06}
Joost Vennekens, David Gilis, and Marc Denecker.
\newblock Splitting an operator: Algebraic modularity results for logics with
  fixpoint semantics.
\newblock \emph{ACM Trans. Comput. Log.}, 7\penalty0 (4):\penalty0 765--797,
  2006.
\newblock \doi{10.1145/1182613.1189735}.
\newblock URL \url{http://dx.doi.org/10.1145/1182613.1189735}.

\bibitem[Vennekens et~al.(2007{\natexlab{a}})Vennekens, Gilis, and
  Denecker]{tocl/VennekensGD06/Fix}
Joost Vennekens, David Gilis, and Marc Denecker.
\newblock Erratum to splitting an operator: {A}lgebraic modularity results for
  logics with fixpoint semantics (vol 7, pg 765, 2006), January
  2007{\natexlab{a}}.
\newblock URL \url{https://lirias.kuleuven.be/handle/123456789/124415}.

\bibitem[Vennekens et~al.(2007{\natexlab{b}})Vennekens, Mari{\"e}n, Wittocx,
  and Denecker]{VennekensMWD07}
Joost Vennekens, Maarten Mari{\"e}n, Johan Wittocx, and Marc Denecker.
\newblock Predicate introduction for logics with a fixpoint semantics. {Parts I
  and II}.
\newblock \emph{Fundamenta Informaticae}, 79\penalty0 (1-2):\penalty0 187--227,
  2007{\natexlab{b}}.

\bibitem[Vlaeminck et~al.(2012)Vlaeminck, Vennekens, Bruynooghe, and
  Denecker]{VlaeminckVBD/KR2012}
Hanne Vlaeminck, Joost Vennekens, Maurice Bruynooghe, and Marc Denecker.
\newblock Ordered {E}pistemic {L}ogic: {S}emantics, complexity and
  applications.
\newblock In Gerhard Brewka, Thomas Eiter, and Sheila~A. McIlraith, editors,
  \emph{Principles of Knowledge Representation and Reasoning: Proceedings of
  the Thirteenth International Conference, KR 2012, Knowledge Representation
  and Reasoning, Rome, 10-14 July 2012}, pages 369--379. AAAI Press, 2012.
\newblock URL \url{http://www.aaai.org/ocs/index.php/KR/KR12/paper/view/4513}.

\bibitem[Yang and Wooldridge(2015)]{ijcai/2015}
Qiang Yang and Michael Wooldridge, editors.
\newblock \emph{Proceedings of the Twenty-Fourth International Joint Conference
  on Artificial Intelligence, {IJCAI} 2015, Buenos Aires, Argentina, July
  25-31, 2015}, 2015. {AAAI} Press.
\newblock ISBN 978-1-57735-738-4.

\bibitem[Yao and Tamassia(2009)]{Yao}
Danfeng Yao and Roberto Tamassia.
\newblock {Compact and Anonymous Role-Based Authorization Chain}.
\newblock \emph{ACM Transactions on Information and System Security}, 2009.

\bibitem[Zhang et~al.(2003)Zhang, Ahn, and Chu]{Zhang03}
Longhua Zhang, Gail-Joon Ahn, and Bei-Tseng Chu.
\newblock {A rule-based framework for role-based delegation and revocation}.
\newblock \emph{ACM Transactions on Information and System Security}, 2003.

\end{thebibliography}
